\documentclass[11pt,oneside,reqno]{article}

\RequirePackage[OT1]{fontenc}
\RequirePackage{amsthm}
\usepackage[paper=a4paper]{geometry}
\usepackage[tbtags]{amsmath}
\usepackage{mathrsfs}
\usepackage{amssymb}
\usepackage{mathtools}
\usepackage{enumerate}
\usepackage{bm}
\usepackage{changes}
\usepackage{natbib}
\usepackage[english]{babel}
\usepackage[utf8x]{inputenc}
\usepackage{graphicx}
\usepackage{setspace}
\usepackage{color,soul}
\usepackage[section]{placeins}
\usepackage{subfig}
\usepackage{appendix}
\usepackage{float}
\usepackage{titlesec}
\usepackage{setspace}
\usepackage{tabularx}

\theoremstyle{plain}
\newtheorem*{thm*}{Theorem}
\newtheorem{thm}{Theorem}[section]
\newtheorem{prop}[thm]{Proposition}

\theoremstyle{definition}
\newtheorem*{defn*}{Definition}

\title{EWF : simulating exact paths of the Wright–Fisher diffusion}
\date{\today}
\author{Jaromir Sant$^{1}$, Paul A.~Jenkins$^{1,2,3}$, Jere Koskela$^{1}$, Dario Span{\`o}$^{1}$ \\ \\ \normalsize{Department of Statistics$^{1}$ \& Department of Computer Science$^{2}$} \\ \normalsize{University of Warwick, Coventry CV4 7AL, United Kingdom} \\ \normalsize{The Alan Turing Institute$^{3}$, British Library, London NW1 2DB, United Kingdom}}

\begin{document}
	\maketitle
	
	\begin{abstract}
	\noindent The Wright--Fisher diffusion is important in population genetics in modelling the evolution of allele frequencies over time subject to the influence of biological phenomena such as selection, mutation, and genetic drift. Simulating paths of the process is challenging due to the form of the transition density. We present EWF, a robust and efficient sampler which returns exact draws for the diffusion and diffusion bridge processes, accounting for general models of selection including those with frequency-dependence. Given a configuration of selection, mutation, and endpoints, EWF returns draws at the requested sampling times from the law of the corresponding Wright--Fisher process. Output was validated by comparison to approximations of the transition density via the Kolmogorov--Smirnov test and QQ plots. All software is available at https://github.com/JaroSant/EWF
	\end{abstract}
	
	\section{Introduction}
	
	The Wright--Fisher diffusion is a central model for the temporal fluctuation of allele frequencies in a large population evolving under random mating and in the presence of mutation and selection.
	Despite its importance, it remains difficult to work with from a computational perspective, both in the absence of selection (where the transition density admits an infinite series expansion) and the non-neutral case (where the corresponding infinite series expansion has intractable terms).
	Additionally, in a diallelic model the diffusion lives on the bounded interval $[0,1]$ and thus even simple approximate sampling techniques such as the Euler--Maruyama scheme require sophisticated modifications to respect its boundary behaviour \citep{Dangerfield}.
	Existing approaches in the literature have tackled this by resorting to a combination of discretisation and numerical approximation, e.g.\ solving the Kolmogorov backwards equation numerically \citep{Bollback,Malaspinas}, approximating through more tractable processes \citep{MathiesonMcVean}, truncating a spectral expansion of the transition density \citep{SteinruckenTDF}, and using Riemann sum approximations \citep{Schraiber}, all of which induce a bias which is hard to quantify.
	
	In some cases, \emph{exact} sampling routines making use of rejection sampling are available.
	This class of techniques has been extended to certain variants of the Wright--Fisher diffusion: \cite{JenkinsSpano} showed that neutral Wright--Fisher diffusion paths and bridges can be simulated exactly via simulation techniques tailored for infinite series, and that neutral paths are the natural proposal mechanism for simulating non-neutral paths by rejection.
	Their work assumes that the mutation parameters are strictly positive and the endpoints for both the diffusion and diffusion bridge lie in the interior of $[0, 1]$.
	The case of diffusion bridges that start and end at 0 was tackled by \cite{GriffithsJenkinsSpano}, but several other combinations of startpoint, endpoint, and parameters remain unaddressed.
	Moreover, no simulation package implementing all of the cases of interest exists.
	
	We present EWF, a C++ package producing exact draws from both neutral and non-neutral Wright--Fisher diffusions. The method properly accounts for all types of boundary (entrance, reflecting, and absorbing), incorporates a wide class of selection models, and allows for arbitrary endpoints, substantially extending previous work by \cite{JenkinsSpano, GriffithsJenkinsSpano}.
	These new theoretical details can be found in the accompanying supplement.
	Additionally, EWF preserves accuracy over long times, in contrast to Euler--Maruyama type schemes where errors accumulate over the simulated path. 
	
	\section{Models}
	
	Consider the two-allele non-neutral Wright--Fisher diffusion $(X_t)_{t\geq 0}$ with mutation parameter $\boldsymbol{\theta}=(\theta_1,\theta_2)$, which is given by the solution to the following stochastic differential equation
	\begin{align}\label{WFDiff}
		dX_t = {}& \frac{1}{2}\left[\sigma X_{t}(1-X_{t})\eta(X_t) -\theta_2 X_t + \theta_1(1-X_t)\right]dt \nonumber \\
		&{}+ \sqrt{X_t(1-X_t)}dW_t 
	\end{align}
	for $t \geq 0$ with $X_0 \in [0,1]$, and $\eta(x)=\sum_{i=0}^{n}a_i x^{i}$ for $n$ finite (e.g.\ for genic selection $\eta(x)=1$, and for diploid selection $\eta(x)=h+x(1-2h)$ with $h$ the dominance parameter).
	When the mutation parameter $\boldsymbol{\theta}$ has positive entries, the corresponding neutral (i.e.\ $\sigma = 0$) transition density can be decomposed into a mixture distribution 
	\begin{align*}
		p^{(\theta_1,\theta_2)}(x,y;t) = \sum_{m=0}^{\infty} q_{m}^{\theta}(t)\sum_{l=0}^{m}\mathrm{Bin}_{m,x}(l)\mathrm{Beta}_{\theta_1+l,\theta_2+m-l}(y),
	\end{align*}
	where $(q_m^{\theta}(t))_{m\in\mathbb{N}}$ is a distribution on the integers and $\theta:=\theta_1+\theta_2$. This allows for exact simulation \citep[Section 2]{JenkinsSpano}. EWF extends this approach to the $\theta_1 = 0$ and/or $\theta_2=0$ cases, when the diffusion is absorbed on hitting 0 and/or 1 in finite time almost surely.
	
	It is often of interest to consider the evolution of a de novo mutation which appears at time $t_0$ and is observed in the population at a sampling time $t > t_0$.
	If $\boldsymbol{\theta}=\mathbf{0}$, one needs to condition the diffusion on non-absorption to recover a non-degenerate transition density. The resulting density can be found in Section 1 in the Supplementary Information (together with the respective details), as well as the corresponding transition densities for the cases when $\boldsymbol{\theta}=(0,\theta)$ or $\boldsymbol{\theta}=(\theta, 0)$.
	
	The transition density for a diffusion \emph{bridge} can be similarly derived (see Section 2 in the Supplementary Information), whilst in the presence of selection (i.e.\ $\sigma \neq 0$ in \eqref{WFDiff}), draws from the corresponding non-neutral process can be returned by simulating neutral paths as candidates in an appropriate rejection scheme \citep[Section 5]{JenkinsSpano}.
	
		\section{Methods}
		
		The expression for $p^{(\theta_1,\theta_2)}(x,y;t)$ tells us that draws from the transition density can be achieved by the following:
		\begin{enumerate}
			\item Draw $M \sim \{q^{\theta}_{m}(t)\}_{m\in\mathbb{N}}$
			\item Conditional on $M=m$, draw $L \sim \textnormal{Bin}(m,x)$
			\item Conditional on $M=m, L=l$, draw $Y \sim \textnormal{Beta}(\theta_1+l,\theta_2+m-l)$
		\end{enumerate}
		Steps 2 and 3 are simple.
		Step 1 is more involved since each $q_m^{\theta}(t)$ is an infinite series (see Supplementary information Section 3 where we have extended the procedure to generate samples when $\boldsymbol{\theta}=\mathbf{0}$ or $\boldsymbol{\theta}=(0,\theta)$).
		
		If the time increment $t$ is small, approximations are necessary due to numerical instabilities in computing $q_{m}^{\theta}(t)$. EWF employs a Gaussian approximation of $q_{m}^{\theta}(t)$ for small $t$ \citep[Theorem 4]{Griffiths84} ($t \leq 0.08$ by default), with similar approximations used for bridges whenever subsequent time increments fall below some threshold. For full details see Section 5 in the Supplementary Information.
		
		The implementation was tested extensively and validated through a combination of QQ plots and the Kolmogorov--Smirnov test (see Supplementary Information Section 7). An example is shown in Fig.\ \ref{HorsesSimulation}.
		
		\begin{figure}[!ptb]
			\centering
			\centering{{\includegraphics[width=.98\linewidth]{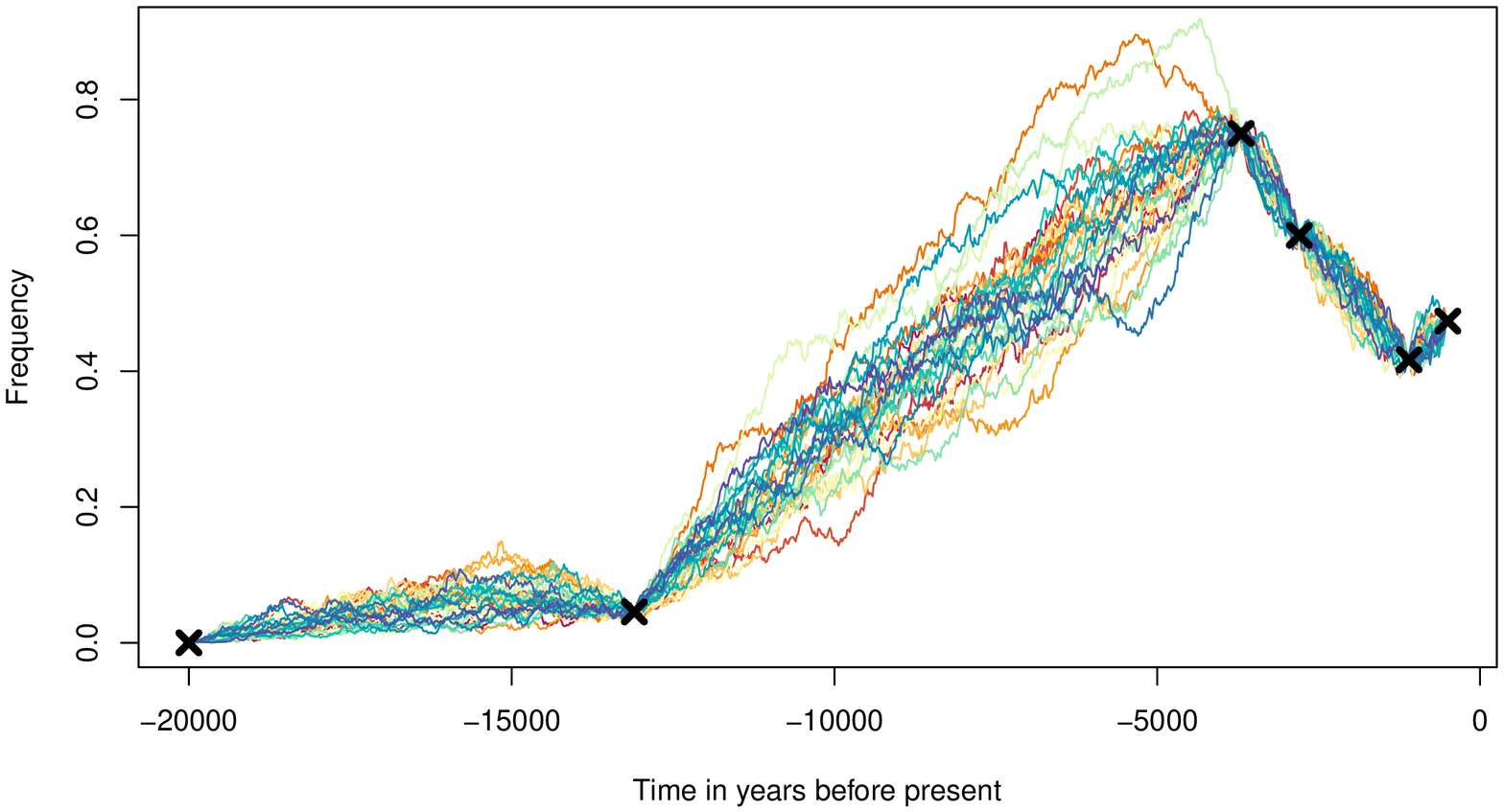} }}
			\caption{Illustration of 30 candidate trajectories for the horse coat color data found in \cite{Ludwig} simulated using EWF (note that the observed frequencies (black crosses) are assumed to be exact observations of the underlying diffusion). Simulations used the inferred selection coefficient $s = 0.0007$ with a consensus effective population size $N_e=10,000$ \citep{Ludwig,Malaspinas,Schraiber}, giving $\sigma= 2N_es = 14$. We used $\boldsymbol{\theta}=\mathbf{0}$ and a generation time of 5 years.}
		\label{HorsesSimulation}
	\end{figure}

\section{Discussion}

EWF provides a robust, efficient, and exact sampling routine to target a wide family of Wright--Fisher diffusions featuring a broad class of selective regimes, any mutation parameters, and any start/end points. The implementation can be used as a stand-alone package, or incorporated into simulation-based inference pipelines from time series allele frequency data. This is particularly useful in view of the recent increase in availability of such data \citep{Wutke2016,Fages2019}. 

\section*{Funding}

This work has been supported by the EPSRC and the Alan Turing Institute under
grants EP/R044732/1, EP/V049208/1, EP/N510129/1. \\ \newline

\noindent \huge\textbf{Supplementary Information} \normalsize

\setcounter{section}{0}

\section{Transition densities for neutral Wright--Fisher diffusions}\label{TransitionDensitiesDiffusion}
\noindent Consider a Wright--Fisher diffusion started from some arbitrary initial point $x\in[0,1]$ with one of the mutation parameters set to 0, say $\boldsymbol{\theta}=(0,\theta)$. Under such a setup, the diffusion  survives up to a time $T_0 := \inf\{ t \geq 0 : X_t = 0 \}$, when it hits 0 and remains there. In this section we derive the transition density when the hitting time $T_0$ is both allowed to occur at any time, and when the sampling time is conditioned on $\{t < T_0\}$. The latter case is slightly harder to tackle because it is necessary to incorporate this conditioning. \newline

\noindent Similar arguments apply for the case when mutation is absent (i.e.\ $\boldsymbol{\theta}=\boldsymbol{0}$), and we further point out that the case $\boldsymbol{\theta}=(\theta,0)$ follows immediately from the case $\boldsymbol{\theta}=(0,\theta)$ by considering the symmetric mapping $x \mapsto 1-x$ and observing that the resulting process is once again a Wright--Fisher diffusion with mutation parameter $\boldsymbol{\theta}'=(\theta_2,\theta_1)$ and selection parameter $\sigma' = -\sigma$. 

\subsection{Neutral diffusion with strictly positive mutation}
\noindent We begin by considering $\theta_1,\theta_2 > 0$ such that both 0 and 1 are non-absorbing boundaries. In this case the transition density can be expressed (\cite{Griffiths79,Tavare84}) as
\begin{align}\label{TransitionDensityDiffusion(theta1,theta2)}
	p^{(\theta_1,\theta_2)}(x,y;t) = \sum_{m=0}^{\infty} q_{m}^{\theta}(t)\sum_{l=0}^{m}\mathcal{B}_{m,x}(l)\mathcal{D}_{\theta_1+l,\theta_2+m-l}(y),
\end{align}
\noindent where $\theta = |\boldsymbol{\theta}| = \theta_1+\theta_2$, $\mathcal{B}_{m,x}(\cdot)$ denotes the binomial probability mass function with parameters $m$ and $x$, $\mathcal{D}_{\theta_1+l,\theta_2+m-l}(\cdot)$ denotes the beta probability density function with parameters $\theta_1+l$ and $\theta_2+m-l$, and 
\begin{align*}
	q_{m}^{\theta}(t) := \sum_{k=m}^{\infty} (-1)^{k-m}\frac{\theta+2k-1}{k!(k-m)!}\frac{\Gamma(\theta+m+k-1)}{\Gamma(\theta+m)}e^{\frac{-k(k+\theta-1)t}{2}},
\end{align*}
\noindent with $\Gamma(\cdot)$ denoting the gamma function. We point out that $\{q_{m}^{\theta}(t)\}_{m\in\mathbb{N}}$ correspond to the transition probabilities of the number of lineages in Kingman's coalescent (which is the moment dual to the Wright--Fisher diffusion), such that $q_{m}^{\theta}(t)$ is the probability that $m$ lineages survive up to time $t$ when one starts with an infinite number of lineages at time 0. For more details, we refer the interested reader to \cite{Griffiths79,Tavare84}. The inclusion of the mutation parameters on the LHS of \eqref{TransitionDensityDiffusion(theta1,theta2)} makes explicit the dependence of the transition density on these quantities, however in an effort to reduce on encumbrance, we shall suppress this notation henceforth and simply write $p(x,y;t)$ for the transition density of the diffusion, with the specific mutation regime being considered specified exogenously.  

\subsection{Neutral diffusion with one sided mutation}
For $\boldsymbol{\theta}=(0,\theta)$, the diffusion is absorbed upon hitting 0 and the transition density can be expressed as
\begin{align}\label{TransitionDensityDiffusion(0,theta)}
	p(x,y;t) = \sum_{m=0}^{\infty} q_{m}^{\theta}(t)\left[\sum_{l=1}^{m}\mathcal{B}_{m,x}(l)\mathcal{D}_{l,\theta+m-l}(y)+(1-x)^{m}\delta_{0}(y)\right],
\end{align}
\noindent where $\delta_{0}(y)$ denotes a point mass at 0 and represents the case when the diffusion is absorbed at 0. In cases like this we reinterpret `density' appropriately, with respect to a dominating measure containing both a Lebesgue component and an atom at each of 0 and 1. \newline

\noindent If we condition on the event $\{t < T_0\}$, standard conditional probability gives us that the transition density of the diffusion \emph{conditioned} on non-absorption until time $t$ is given by 
\begin{align*}
	\tilde{p}(x,y;t) &= \frac{p(x,y;t)}{\mathbb{P}_{x}\left[T_0 > t\right]},
\end{align*}
for $y\in(0,1]$, where we use the notation $\tilde{p}(\cdot,\cdot;\cdot)$ to make explicit the fact that this is the transition density of the \emph{conditioned} diffusion process. Additionally, we have that 
\begin{align}\label{HittingTimeProbability}
	\mathbb{P}_{x}\left[T_0>t\right] &= \int_{(0,1]} p(x,u;t) du \nonumber \\
	&= \sum_{m=1}^{\infty}q_{m}^{\theta}(t)\sum_{l=1}^{m}\mathcal{B}_{m,x}(l),
\end{align}
and we note that the contributions from $m=0$ above are missing as the corresponding beta density collapses to a point mass at 0. Thus for $x,y\in (0,1]$ we have
\begin{align}\label{WFOneSidedMut1}
	\tilde{p}(x,y;t) = \sum_{m=1}^{\infty}\frac{ q_{m}^{\theta}(t)}{\sum_{l=1}^{m}\mathcal{B}_{m,x}(l)\sum_{d=1}^{\infty} q_{d}^{\theta}(t)(1-(1-x)^{d})}\mathcal{D}_{l,\theta+m-l}(y).
\end{align}
\noindent For small $x$, we have the following leading order expansion in $x$
\begin{align}\label{0ThetaDensity}
	p(x,y;t) = x\sum_{m=1}^{\infty}q_{m}^{\theta}(t)m(\theta+m-1)(1-y)^{\theta+m-2} + O(x^{2}),
\end{align}
\noindent and note further \eqref{HittingTimeProbability} is also of leading order $x$ for $x$ small. Thus upon taking the limit $x\rightarrow 0$ in \eqref{WFOneSidedMut1} we get that
\begin{align}\label{WFOneSidedMut2}
	\tilde{p}(0,y;t) &= \sum_{m=1}^{\infty}\frac{m q_{m}^{\theta}(t)}{\sum_{d=1}^{\infty}d q_{d}^{\theta}(t)}\mathcal{D}_{1,\theta+m-1}(y).
\end{align}
\noindent Putting all of the above together we get that the conditioned diffusion has transition density given by
\begin{align}\label{OneSidedMutDen}
	\tilde{p}(x,y;t) = \begin{cases} \displaystyle\sum_{m=1}^{\infty}\frac{m q_{m}^{\theta}(t)}{\sum_{d=1}^{\infty}d q_{d}^{\theta}(t)}\mathcal{D}_{1,\theta+m-1}(y) & x = 0, \\ & \\
		\displaystyle\sum_{m=1}^{\infty}\frac{ q_{m}^{\theta}(t)\sum_{l=1}^{m}\mathcal{B}_{m,x}(l)}{\sum_{d=1}^{\infty} q_{d}^{\theta}(t)(1-(1-x)^{d})}\mathcal{D}_{l,\theta+m-l}(y) & x \in (0,1].
	\end{cases}
\end{align}
\noindent We point out that as the diffusion is conditioned on avoiding 0, there will always be at least one surviving lineage in the moment-dual Kingman coalescent, and thus the index for $m$ starts at 1. 

\subsection{Diffusion without mutation}
\noindent If $\boldsymbol{\theta}=\boldsymbol{0}$, then the diffusion is absorbed upon hitting either boundary, and the corresponding transition density is given by
\begin{align} \label{TransitionDensityDiffusion(0,0)}
	p(x,y;t) = \sum_{m=2}^{\infty} q_{m}^{\theta}(t)\left[\sum_{l=1}^{m-1}\mathcal{B}_{m,x}(l)\mathcal{D}_{l,m-l}(y)+(1-x)^{m}\delta_{0}(y)+x^{m}\delta_{1}(y)\right],
\end{align}

\noindent Conditioning the diffusion on remaining inside the interior of $[0,1]$, and again employing a leading order analysis of the resulting numerator and denominator allows us to conclude that the transition density in this case is given by
\begin{align}\label{ZeroMutDen}
	\tilde{p}(x,y;t) = \begin{cases} \displaystyle\sum_{m=2}^{\infty}\frac{m q_{m}^{0}(t)}{\sum_{d=2}^{\infty}d q_{d}^{0}(t)}\mathcal{D}_{1,m-1}(y) & x = 0, \\ & \\
		\displaystyle\sum_{m=2}^{\infty}\frac{m q_{m}^{0}(t)}{\sum_{d=2}^{\infty}d q_{d}^{0}(t)}\mathcal{D}_{m-1,1}(y) & x = 1, \\ & \\
		\displaystyle\sum_{m=2}^{\infty}\frac{ q_{m}^{0}(t)\sum_{l=1}^{m-1}\mathcal{B}_{m,x}(l)}{\sum_{d=2}^{\infty} q_{d}^{0}(t)(1-x^{d}-(1-x)^{d})}\mathcal{D}_{l,m-l}(y) & x \in (0,1).
	\end{cases}
\end{align}
Note that as $\boldsymbol{\theta} = \boldsymbol{0}$ and we are conditioning on non-absorption, the indices $m$ and $d$ are now forced to start from 2. This follows from the fact that the derivations performed above assume the starting point $x$ to be within $(0,1)$ and subsequently send $x$ to the corresponding boundary from within the interior of $(0,1)$, which corresponds to starting the diffusion arbitrarily close to the boundary. Thus at all times there is a fraction $x$ of the population having one type, with the other fraction $1-x$ having the other, neither of which can be lost by mutation. 

\section{Transition densities for neutral Wright--Fisher diffusion bridges}\label{TransitionDensitiesBridge}
\noindent We now derive the density of a point $y\in[0,1]$ sampled at time $s\in(0,t)$ from the law of a Wright--Fisher diffusion bridge started at $x$ at time 0 and ending at $z$ at time $t$. In addition to considering each mutation regime separately, we further split our considerations based on the values the start and end points $x$ and $z$ assume. As in the diffusion case, we derive the relevant expressions in the case $\boldsymbol{\theta}=(0,\theta)$, as the other cases ($\boldsymbol{\theta}=(0,0)$ or $\boldsymbol{\theta}=(\theta,0)$) follow using similar arguments. We further consider both cases when (i) the bridge is allowed to be absorbed at any time point within the time interval $(0,t)$, and (ii) the bridge is conditionally non-absorbing: $X_s \in (0,1)$ for all $s \in (0,t)$. We make use of the following short-hand notation for the different possible end-point combinations.
\begin{center}
	\begin{tabular}{c|c|c|c}\label{BridgeNotationTable}
		& $x=0$ & $x=1$ & $x \in(0,1)$ \\ \hline
		$z=0$ & A1 & B1 & C1 \\
		$z\in(0,1)$ & A2 & B2 & C2 \\
		$z=1$ & A3 & B3 & C3 
	\end{tabular}
\end{center}
We further introduce a letter at the front of each of the above to differentiate between the cases $\boldsymbol{\theta}=\boldsymbol{0}$ (`Z' for zero), $\boldsymbol{\theta}=(0,\theta)$ (`O' for one sided), and $\boldsymbol{\theta}$ with strictly positive entries (`P' for strictly positive). \newline

\noindent Before proceeding with deriving the transition densities for all the above outlined cases, observe that the transition density for a Wright--Fisher diffusion bridge started from $x\in[0,1]$ at time 0, ending at $z\in[0,1]$ at time $t$ and sampled at time $s$ can be factorised as follows for $y \in [0,1]$:
\begin{align}\label{BridgeDecomposition}
	p^{x,z;t}(y;s) = \frac{p(x,y;s)p(y,z;t-s)}{p(x,z;t)},
\end{align} 
\noindent where again, for simplicity the dependence of \eqref{BridgeDecomposition} on the mutation parameters is omitted from the notation. 

\subsection{Neutral diffusion bridge with one sided mutation $\boldsymbol{\theta}=(0,\theta)$}
\noindent We start by noting that if the diffusion bridge is allowed to be absorbed at 0 at any time within the interval $(0,t)$, then the only cases of interest are when the left endpoint $x\in(0,1]$, for otherwise the bridge stays at 0. Additionally if $z \in (0,1]$, the bridge could not have been absorbed within the time interval $(0,t)$, and is therefore equivalent to conditioning it on non-absorption (which shall be tackled shortly). Thus we take $x\in(0,1)$ and $z = 0$, substitute \eqref{TransitionDensityDiffusion(0,theta)} into \eqref{BridgeDecomposition}, and re-group terms to get that
\begin{align}\label{TransitionDensityBridge(0,theta)}
	p^{x,z;t}(y;s) = \sum_{m,k=1}^{\infty}\frac{q^{\theta}_{m}(s)q^{\theta}_{k}(t-s)}{\sum_{d=1}^{\infty}q_{d}^{\theta}(t)(1-x)^{d}}\Bigg[&\sum_{l=1}^{m}\mathcal{B}_{m,x}(l)\frac{B(l,\theta+m-l+k)}{B(l,\theta+m-l)}\mathcal{D}_{l,\theta+m-l+k}(y) \nonumber \\
	&\qquad{}+ (1-x)^{m}\delta_{0}(y)\Bigg].
\end{align}
\noindent where $B(\cdot,\cdot)$ is the beta function. \newline

\noindent To derive the transition density when $z\in(0,1]$, we first point out that conditioning a diffusion (or conditioning a diffusion bridge) on non-absorption is a special case of taking an $h$-transform for said process (see for instance \cite{Fitzsimmons92,GriffithsJenkinsSpano}). Furthermore, diffusion bridges are invariant under $h$-transforms (see equation (10) in \cite{GriffithsJenkinsSpano}), and thus the distribution of a diffusion bridge conditioned on non-absorption is the same as that of the corresponding unconditioned process. We therefore need not differentiate between the transition density of the conditioned or unconditioned diffusion bridge, and simply use $p^{x,z;t}(y;s)$ throughout. \newline

\noindent Expanding \eqref{BridgeDecomposition} for $x,z\in(0,1]$ gives
\begin{align}\label{XZOneSidedMut}
	p^{x,z;t}(y;s) ={} & \sum_{m,k=1}^{\infty}\frac{q_{m}^{\theta}(s)q_{k}^{\theta}(t-s)}{\sum_{d=1}^{\infty}q_{d}^{\theta}(t)\sum_{f=1}^{d}\mathcal{B}_{d,x}(f)\mathcal{D}_{f,\theta+d-f}(z)} \nonumber \\ 
	&\qquad{}\times\sum_{l,j=1}^{m,k}\binom{k}{j}\frac{B(l+j,\theta+m-l+k-j)}{B(l,\theta+m-l)}\mathcal{D}_{j,\theta+k-j}(z)\mathcal{D}_{l+j,\theta+m-l+k-j}(y).
\end{align}
\noindent When $x=0$, we make use of \eqref{0ThetaDensity} in both the numerator and denominator above, and subsequently take the limit as $x\rightarrow 0$, to arrive at
\begin{align}\label{X0ZOneSidedMut}
	p^{0,z;t}(y;s) &= \lim_{x\rightarrow0}\left(\frac{x\sum_{m=1}^{\infty}q_{m}^{\theta}(s)m(\theta+m-1)(1-y)^{\theta+m-2} + o(x^{2})}{x\sum_{d=1}^{\infty}q_{d}^{\theta}(t)d(\theta+d-1)(1-z)^{\theta+d-2} + o(x^{2})}\right)p(y,z;t-s) \nonumber \\
	&= \sum_{m,k=1}^{\infty}\frac{q_{m}^{\theta}(s)q_{k}^{\theta}(t-s)}{\sum_{d=1}^{\infty}q_{d}^{\theta}(t)d(d+\theta-1)(1-z)^{\theta+d-2}} \nonumber\\ 
	&\qquad{}\times\sum_{j=1}^{k}\binom{k}{j}\frac{B(j+1,\theta+m-1+k-j)}{B(1,\theta+m-1)}\mathcal{D}_{j,\theta+k-j}(z)\mathcal{D}_{j+1,\theta+m-1+k-j}(y).
\end{align}
The above expression can further be used to derive the expression when $z=0$ by taking leading order terms in $z$ and taking the limit $z\rightarrow 0$, giving
\begin{align}\label{X0Z0OneSidedMut}
	p^{0,0;t}(y;s) &= \sum_{m,k=1}^{\infty}\frac{q_{m}^{\theta}(s)q_{k}^{\theta}(t-s)}{\sum_{d=1}^{\infty}q_{d}^{\theta}(t)d(d+\theta-1)}\frac{m(m+\theta-1)k(k+\theta-1)}{(m+k+\theta-1)(m+k+\theta-2)} \mathcal{D}_{2,\theta+m-1+k-1}(y).
\end{align}
\noindent As previously mentioned, the case $\boldsymbol{\theta}=(\theta,0)$ follows from the above by considering the symmetric map $x \mapsto 1-x$.

\subsection{Neutral diffusion bridge with no mutation}
\noindent We can replicate all of the above arguments for when $\boldsymbol{\theta} = \boldsymbol{0}$ to get that if $x\in(0,1)$ and $z=0$, for $y \in [0,1)$ we have
\begin{align}\label{TransitionDensityBridge(0,0)z0}
	p^{x,0;t}(y;s) =  \sum_{m,k=1}^{\infty}\frac{q^{0}_{m}(s)q^{0}_{k}(t-s)}{\sum_{d=1}^{\infty}q_{d}^{0}(t)(1-x)^{d}}\Bigg[&\sum_{l=1}^{m-1}\mathcal{B}_{m,x}(l)\frac{B(l,m-l+k)}{B(l,m-l)}\mathcal{D}_{l,m-l+k}(y) \nonumber \\
	&\qquad{}+ (1-x)^{m}\delta_{0}(y)\Bigg]
\end{align}
\noindent whilst if $x \in (0,1)$ and $z = 1$, we get for $y \in (0,1]$
\begin{align}\label{TransitionDensityBridge(0,0)z1}
	p^{x,1;t}(y;s) =  \sum_{m,k=1}^{\infty}\frac{q^{0}_{m}(s)q^{0}_{k}(t-s)}{\sum_{d=1}^{\infty}q_{d}^{0}(t)x^{d}}\left[\sum_{l=1}^{m-1}\mathcal{B}_{m,x}(l)\frac{B(l+k,m-l)}{B(l,m-l)}\mathcal{D}_{l+k,m-l}(y) + x^{m}\delta_{1}(y)\right]
\end{align}
\noindent Note that if $z = 0$, then we cannot have $y = 1$ and similarly if $z=1$, $y$ cannot be equal to 0. \newline

\noindent Computing the transition densities conditioned on non-absorption can be done as in the one-sided mutation case illustrated above, by following the same arguments. \newline

\noindent The resulting expressions for the conditioned diffusion bridges under all three mutation regimes can be found below (recall the notation in Table \ref{BridgeNotationTable}). 

\subsection{Bridge diffusion transition density when $\boldsymbol{\theta}=\boldsymbol{0}$}
\begin{align*}
	p^{x,z;t}(y;s) &= \sum\limits_{m,k=2}^{\infty}\frac{ q_{m}^{0}(s)q_{k}^{0}(t-s)}{\sum_{d=2}^{\infty}q_{d}^{0}(t)d(d-1) }\frac{m(m-1)k(k-1)}{(m+k-1)(m+k-2)}\mathcal{D}_{2,m+k-2}(y) & ZA1 \\ & & \\
	p^{x,z;t}(y;s) &= \sum\limits_{m,k=2}^{\infty}\frac{q_{m}^{0}(s)q_{k}^{0}(t-s)}{\sum_{d=2}^{\infty}q_{d}^{0}(t)d\mathcal{D}_{1,d-1}(z)}m\sum\limits_{j=1}^{k-1}\binom{k}{j}\frac{B(j+1,m-1+k-j)}{B(1,m-1)}\mathcal{D}_{j,k-j}(z) & \\
	&\qquad{}\qquad{}\qquad{}\qquad{}\qquad{}\qquad{}\qquad{}\qquad{}\qquad{}\qquad{}\times \mathcal{D}_{j+1,m-1+k-j}(y) & ZA2 \\ & & \\
	p^{x,z;t}(y;s) &= \sum\limits_{m,k=2}^{\infty}\frac{ q_{m}^{0}(s)q_{k}^{0}(t-s)}{2q_{2}^{0}(t)}m(m-1)k(k-1)B(k,m)\mathcal{D}_{k,m}(y) & ZA3 \\ & & \\
	p^{x,z;t}(y;s) &= \sum\limits_{m,k=2}^{\infty}\frac{ q_{m}^{0}(s)q_{k}^{0}(t-s)}{2q_{2}^{0}(t)}m(m-1)k(k-1)B(m,k)\mathcal{D}_{m,k}(y) & ZB1 \\ & & \\
	p^{x,z;t}(y;s) &= \sum\limits_{m,k=2}^{\infty}\frac{ q_{m}^{0}(s)q_{k}^{0}(t-s)}{\sum_{d=2}^{\infty}q_{d}^{0}(t)d\mathcal{D}_{d-1,1}}m\sum\limits_{j=1}^{k-1}\binom{k}{j}\frac{B(m-1+j,k-j+1)}{B(m-1,1)}\mathcal{D}_{j,k-j}(y) & \\
	&\qquad{}\qquad{}\qquad{}\qquad{}\qquad{}\qquad{}\qquad{}\qquad{}\qquad{}\qquad{}\times \mathcal{D}_{m-1+j,1+k-j}(y) & ZB2 \\ & & \\
	p^{x,z;t}(y;s) &= \sum\limits_{m,k=2}^{\infty}\frac{ q_{m}^{0}(s)q_{k}^{0}(t-s)}{\sum_{d=2}^{\infty}q_{d}^{0}(t)d(d-1) }\frac{m(m-1)k(k-1)}{(m+k-1)(m+k-2)}\mathcal{D}_{m+k-2,2}(y) & ZB3 \\ & & \\
	p^{x,z;t}(y;s) &= \sum\limits_{m,k=2}^{\infty}\frac{q_{m}^{0}(s)q_{k}^{0}(t-s)}{\sum_{d=2}^{\infty}q_{d}^{0}(t)(d-1)\mathcal{B}_{d,x}(1)}\sum_{l=1}^{m-1}\mathcal{B}_{m,x}(l)k(k-1)\frac{B(l+1,m-l+k-1)}{B(l,m-l)} & \\
	&\qquad{}\qquad{}\qquad{}\qquad{}\qquad{}\qquad{}\qquad{}\qquad{}\qquad{}\qquad{}\times \mathcal{D}_{l+1,m-l+k-1}(y) & ZC1 \\ & & \\
	p^{x,z;t}(y;s) &= \sum\limits_{m,k=2}^{\infty}\frac{q_{m}^{0}(s)q_{k}^{0}(t-s)}{\tilde{p}(x,z;t)}\sum\limits_{l,j=1}^{m-1,k-1}\mathcal{B}_{m,x}(l)\binom{k}{j}\frac{B(l+j,m-l+k-j)}{B(l,m-l)}\mathcal{D}_{j,k-j} & \\
	&\qquad{}\qquad{}\qquad{}\qquad{}\qquad{}\qquad{}\qquad{}\qquad{}\qquad{}\qquad{}\times\mathcal{D}_{l+j,m-l+k-j}(y) & ZC2 \\
	p^{x,z;t}(y;s) &= \sum\limits_{m,k=2}^{\infty}\frac{q_{m}^{0}(s)q_{k}^{0}(t-s)}{\sum_{d=2}^{\infty}q_{d}^{0}(t)(d-1)\mathcal{B}_{d,x}(d-1)}\sum_{l=1}^{m-1}\mathcal{B}_{m,x}(l)k(k-1)\frac{B(l+k-1,m-l+1)}{B(l,m-l)} \\
	&\qquad{}\qquad{}\qquad{}\qquad{}\qquad{}\qquad{}\qquad{}\qquad{}\qquad{}\qquad{}\times\mathcal{D}_{l+k-1,m-l+1}(y) & ZC3
\end{align*}

\subsection{Bridge diffusion transition density when $\boldsymbol{\theta}=(0,\theta)$}
\begin{align*}
	p^{x,z;t}(y;s) &= \sum\limits_{m,k=1}^{\infty}\frac{ q_{m}^{\theta}(s)q_{k}^{\theta}(t-s)}{\sum_{d=1}^{\infty}q_{d}^{\theta}(t)d(d+\theta-1) }\frac{m(m+\theta-1)k(k+\theta-1)}{(m+k+\theta-1)(m+k+\theta-2)}\mathcal{D}_{2,\theta+m+k-2}(y) & OA1 \\ & & \\
	p^{x,z;t}(y;s) &= \sum_{m,k=1}^{\infty}\frac{q_{m}^{\theta}(s)q_{k}^{\theta}(t-s)}{\sum_{d=1}^{\infty}q_{d}^{\theta}(t)d\mathcal{D}_{1,\theta+d-1}(z)}m\sum_{j=1}^{k}\binom{k}{j}\frac{B(j+1,\theta+m-1+k-j)}{B(1,\theta+m-1)} & \\
	&\qquad{}\qquad{}\qquad{}\qquad{}\qquad{}\qquad{}\qquad{}\qquad{}\qquad{}\qquad{}\times\mathcal{D}_{j,\theta+k-j}(z) \mathcal{D}_{j+1,\theta+m-1+k-j}(y) & OA2 \\ & & \\
	p^{x,z;t}(y;s) &= \sum\limits_{m,k=1}^{\infty}\frac{q_{m}^{\theta}(s)q_{k}^{\theta}(t-s)}{\theta q_1}m(m+\theta-1)\frac{B(k+1,\theta+m-1)}{B(k,\theta)}\mathcal{D}_{k+1,\theta+m-1}(y) & OA3 \\ & & \\
	p^{x,z;t}(y;s) &= \sum\limits_{m,k=1}^{\infty}\frac{q_{m}^{\theta}(s)q_{k}^{\theta}(t-s)}{\theta q_1}k(k+\theta-1)\frac{B(m+1,\theta+k-1)}{B(m,\theta)}\mathcal{D}_{m+1,\theta+k-1}(y) & OB1 \\ & & \\
	p^{x,z;t}(y;s) &= \sum\limits_{m,k=1}^{\infty}\frac{q_{m}^{\theta}(s)q_{k}^{\theta}(t-s)}{\sum_{d=1}^{\infty}q_{d}^{\theta}(t)\mathcal{D}_{d,\theta}}\sum\limits_{j=1}^{k}\binom{k}{j}\frac{B(m+j,\theta+k-j)}{B(m,\theta)}\mathcal{D}_{j,\theta+k-j}(z) & \\ 
	&\qquad{}\qquad{}\qquad{}\qquad{}\qquad{}\qquad{}\qquad{}\qquad{}\times\mathcal{D}_{m+j,\theta+k-j}(y) & OB2 \\ & & \\
	p^{x,z;t}(y;s) &= \sum\limits_{m,k=1}^{\infty} \frac{q_{m}^{\theta}(s)q_{k}^{\theta}(t-s)}{\sum_{d=1}^{\infty}q_{d}^{\theta}(t)\frac{1}{B(d,\theta)}}\frac{B(m+k,\theta)}{B(m,\theta)B(k,\theta)}\mathcal{D}_{m+k,\theta}(y) & OB3 \\ & & \\
	p^{x,z;t}(y;s) &= \sum\limits_{m,k=1}^{\infty}\frac{q_{m}^{\theta}(s)q_{k}^{\theta}(t-s)k(k+\theta-1)}{\sum_{d=1}^{\infty}q_{d}^{\theta}(t)(d+\theta-1)\mathcal{B}_{d,x}(1)}\sum\limits_{l=1}^{m}\mathcal{B}_{m,x}(l)\frac{B(l+1,\theta+k-1+m-l)}{B(l,\theta+m-l)} & \\
	&\qquad{}\qquad{}\qquad{}\qquad{}\qquad{}\qquad{}\qquad{}\qquad{}\qquad{}\times\mathcal{D}_{l+1,\theta+k-1+m-l}(y) & OC1 \\ & & \\
	p^{x,z;t}(y;s) &= \sum_{m,k=1}^{\infty}\frac{q_{m}^{\theta}(s)q_{k}^{\theta}(t-s)}{\tilde{p}(x,z;t)}\sum\limits_{l,j=1}^{m,k}\mathcal{B}_{m,x}(l)\binom{k}{j}\frac{B(l+j,\theta+m-l+k-j)}{B(l,\theta+m-l)}\mathcal{D}_{j,\theta+k-j}(z) & \\
	&\qquad{}\qquad{}\qquad{}\qquad{}\qquad{}\qquad{}\qquad{}\qquad{}\times\mathcal{D}_{l+j,\theta+m-l+k-j} & OC2 \\ & & \\
	p^{x,z;t}(y;s) &= \sum\limits_{m,k=1}^{\infty}\frac{q_{m}^{\theta}(s)q_{k}^{\theta}(t-s)}{\sum_{d=1}^{\infty}q_{d}^{\theta}(t)\frac{x^{d}}{B(d,\theta)}}\sum_{l=1}^{m}\mathcal{B}_{m,x}(l)\frac{B(l+k,\theta+m-l)}{B(l,\theta+m-l)B(k,\theta)}\mathcal{D}_{l+k,\theta+m-l}(y) & OC3
\end{align*}

\subsection{Bridge diffusion transition density when $\boldsymbol{\theta}=(\theta_1,\theta_2)$}
\begin{align*}
	p^{x,z;t}(y;s) &= \sum\limits_{m,k=0}^{\infty}\frac{q_{m}^{\theta}(s)q_{k}^{\theta}(t-s)}{\sum_{d=0}^{\infty}q_{d}^{\theta}\frac{1}{B(\theta_1,\theta_2+d)}}\frac{B(\theta_1,\theta_2+m+k)}{B(\theta_1,\theta_2+m)B(\theta_1,\theta_2+k)}\mathcal{D}_{\theta_1,\theta_2+m+k}(y) & PA1 \\ & & \\
	p^{x,z;t}(y;s) &= \sum\limits_{m,k=0}^{\infty}\frac{q_{m}^{\theta}(s)q_{k}^{\theta}(t-s)}{\sum_{d=0}^{\infty}q_{d}^{\theta}\mathcal{D}_{\theta_1,\theta_2+d}(z)}\sum_{j=0}^{k}\binom{k}{j}\frac{B(\theta_1+j,\theta_2+m+k-j)}{B(\theta_1,\theta_2+m)}\mathcal{D}_{\theta_1+j,\theta_2+k-j}(z) & \\
	&\qquad{}\qquad{}\qquad{}\qquad{}\qquad{}\qquad{}\qquad{}\times\mathcal{D}_{\theta_1+j,\theta_2+m+k-j}(y) & PA2 \\ & & \\
	p^{x,z;t}(y;s) &= \sum\limits_{m,k=0}^{\infty}\frac{q_{m}^{\theta}(s)q_{k}^{\theta}(t-s)}{q_{0}^{\theta}\frac{1}{B(\theta_1,\theta_2)}}\frac{B(\theta_1+k,\theta_2+m)}{B(\theta_1,\theta_2+m)B(\theta_1+k,\theta_2)}\mathcal{D}_{\theta_1+k,\theta_2+m}(y) & PA3 \\ & & \\
	p^{x,z;t}(y;s) &= \sum\limits_{m,k=0}^{\infty}\frac{q_{m}^{\theta}(s)q_{k}^{\theta}(t-s)}{q_{0}^{\theta}\frac{1}{B(\theta_1,\theta_2)}}\frac{B(\theta_1+m,\theta_2+k)}{B(\theta_1+m,\theta_2)B(\theta_1,\theta_2+k)}\mathcal{D}_{\theta_1+m,\theta_2+k}(y) & PB1 \\ & & \\
	p^{x,z;t}(y;s) &= \sum\limits_{m,k=0}^{\infty}\frac{q_{m}^{\theta}(s)q_{k}^{\theta}(t-s)}{\sum_{d=0}^{\infty}q_{d}^{\theta}\mathcal{D}_{\theta_1+d,\theta_2}(z)}\sum_{j=0}^{k}\binom{k}{j}\frac{B(\theta_1+m+j,\theta_2+k-j)}{B(\theta_1+m,\theta_2)}\mathcal{D}_{\theta_1+j,\theta_2+k-j}(z) & \\
	&\qquad{}\qquad{}\qquad{}\qquad{}\qquad{}\qquad{}\qquad{}\times\mathcal{D}_{\theta_1+m+j,\theta_2+k-j}(y) & PB2 \\ & & \\
	p^{x,z;t}(y;s) &= \sum\limits_{m,k=0}^{\infty}\frac{q_{m}^{\theta}(s)q_{k}^{\theta}(t-s)}{\sum_{d=0}^{\infty}q_{d}^{\theta}\frac{1}{B(\theta_1+d,\theta_2)}}\frac{B(\theta_1+m+k,\theta_2)}{B(\theta_1+m,\theta_2)B(\theta_1+k,\theta_2)}\mathcal{D}_{\theta_1+m+k,\theta_2}(y) & PB3 \\ & & \\
	p^{x,z;t}(y;s) &= \sum_{m,k=0}^{\infty}\frac{q_{m}^{\theta}(s)q_{k}^{\theta}(t-s)}{\sum_{d=0}^{\infty}q_{d}^{\theta}(t)\frac{(1-x)^{d}}{B(\theta_1,\theta_2+d)}}\sum_{l=0}^{m}\mathcal{B}_{m,x}(l)\frac{B(\theta_1+l,\theta_2+m-l+k)}{B(\theta_1,\theta_2+k)B(\theta_1+l,\theta_2+m-l)} \\
	&\qquad{}\qquad{}\qquad{}\qquad{}\qquad{}\qquad{}\qquad{}\times\mathcal{D}_{\theta_1+l,\theta_2+m-l+k}(y) & PC1 \\ & & \\
	p^{x,z;t}(y;s) &= \sum_{m,k=0}^{\infty}\frac{q_{m}^{\theta}(s)q_{k}^{\theta}(t-s)}{\tilde{p}(x,z;t)}\sum_{l,j=0}^{m,k}\mathcal{B}_{m,x}(l)\binom{k}{j}\frac{B(\theta_1+l+j,\theta_2+m-l+k-j)}{B(\theta_1+l,\theta_2+m-l)} \\
	&\qquad{}\qquad{}\qquad{}\qquad{}\qquad{}\qquad{}\qquad{}\times\mathcal{D}_{\theta_1+j,\theta_2+k-j}(z)\mathcal{D}_{\theta_1+l+j,\theta_2+m-l+k-j}(y) & PC2 \\ & & \\
	p^{x,z;t}(y;s) &= \sum_{m,k=0}^{\infty}\frac{q_{m}^{\theta}(s)q_{k}^{\theta}(t-s)}{\sum_{d=0}^{\infty}q_{d}^{\theta}(t)\frac{x^{d}}{B(\theta_1+d,\theta_2)}}\sum_{l=0}^{m}\mathcal{B}_{m,x}(l)\frac{B(\theta_1+l+k,\theta_2+m-l)}{B(\theta_1+k,\theta_2)B(\theta_1+l,\theta_2+m-l)} \\
	&\qquad{}\qquad{}\qquad{}\qquad{}\qquad{}\qquad{}\qquad{}\times\mathcal{D}_{\theta_1+l+k,\theta_2+m-l}(y) & PC3
\end{align*}

\section{Sampling schemes}\label{SamplingSchemes}
We now detail how to obtain sample draws from the above transition densities for both the diffusion and diffusion bridge case.
\subsection{Sampling from the law of the diffusion}
\noindent Note that we need only consider the cases $\boldsymbol{\theta} = (0,\theta)$ and $\boldsymbol{\theta}=\boldsymbol{0}$, as the case $\boldsymbol{\theta}=(\theta_1,\theta_2)$ is already covered in \cite{JenkinsSpano}. Furthermore, the transition densities \eqref{TransitionDensityDiffusion(0,theta)}, \eqref{OneSidedMutDen}, \eqref{TransitionDensityDiffusion(0,0)} and \eqref{ZeroMutDen} for $x\in[0,1)$ are similar, allowing for near identical sampling schemes. To this end, we restrict our attention to the case when $\boldsymbol{\theta}=(0,\theta)$, starting with a sampling scheme for \eqref{TransitionDensityDiffusion(0,theta)}. \newline

\noindent In this case, Algorithm 1 in \cite{JenkinsSpano} can be easily adapted to sample from \eqref{TransitionDensityDiffusion(0,theta)}:
\begin{enumerate}
	\item Sample $M \sim \{ q^{\theta}_{m}(t) \}_{m \in\mathbb{N}}$,
	\item Conditionally on $M = m$, sample $L \sim \textnormal{Bin}(m,x)$,
	\item If $L = 0$ return 0, else draw $Y \sim \textnormal{Beta}(l,\theta+m-l)$.
\end{enumerate}
\noindent The only modification to Algorithm 1 in \cite{JenkinsSpano} is the sampling procedure in step 3, where the outcome $L=0$ encodes the event when the diffusion is absorbed before the sampling time. A similar strategy allows for draws from \eqref{TransitionDensityDiffusion(0,0)}, where additionally if $L=m$, then in step 3 we return $Y = 1$. \newline

\noindent For the case when the diffusion is conditioned on non-absorption, both expressions on the RHS of \eqref{OneSidedMutDen} are mixtures of beta distributions, with the weights forming a probability mass function (pmf) on $\mathbb{N}$. When the starting point $x$ is set to 0, one can return a draw $Y$ from the law of the corresponding diffusion process sampled at time $t$, by
\begin{enumerate}
	\item Drawing $M \sim \left\{ \frac{mq_{m}^{0}(t)}{\sum_{d=1}^{\infty}dq_{d}^{0}(t)} \right\}_{m\in\mathbb{N}}$
	\item Conditionally on $M=m$, drawing $Y \sim \textnormal{Beta}(1,m-1)$
\end{enumerate}
Step 2 is straightforward, whilst for step 1 the `alternating series trick' can be employed---this technique requires access to a pair of monotonic sequences of upper and lower bounds for terms in the numerator and denominator, both converging to their exact values. This is immediate for the numerator (Proposition 1 in \cite{JenkinsSpano}), whilst for the denominator we modify slightly the arguments present in Proposition 3 in \cite{JenkinsSpano} (see Section \ref{ExtensionOfDenominator} for further details). \newline

\noindent A similar sampling scheme can be used for drawing samples from the law of the diffusion started from $x\in(0,1]$, where once again appropriate monotonic upper and lower bounds can be constructed for both  numerator and denominator (see Section \ref{ExtensionOfDenominator} for more details). \newline

\noindent The above can be replicated and suitably tweaked to return samples from \eqref{ZeroMutDen}, where an additional scheme is needed to deal with the case $x=1$.

\subsection{Sampling from the law of the diffusion bridge}
\noindent Once again we start by considering the case when the bridge is allowed to be absorbed at the boundary within the time interval $(0,t)$, and the mutation parameter is given by $\boldsymbol{\theta}=(0,\theta)$. \newline

\noindent To sample from \eqref{TransitionDensityBridge(0,theta)}, we follow an approach similar to that illustrated above for the diffusion case. Recall that we need only focus on the case when $z = 0$, for otherwise the bridge cannot have been absorbed during the time interval $(0,t)$ and thus is equivalent to conditioning on non-absorption (for which an appropriate sampling scheme will be provided shortly). Observe that the RHS of \eqref{TransitionDensityBridge(0,theta)} can be viewed as a mixture of beta distributions, with the mixture weights
\begin{align*}
	p_{m,k,l} := \begin{cases}
		\frac{q^{\theta}_{m}(s)q^{\theta}_{k}(t-s)}{\sum_{d=1}^{\infty}q^{\theta}_{d}(t)(1-x)^{d}}\mathcal{B}_{m,x}(l)\frac{B(l,\theta+m-l+k)}{B(l,\theta+m-l)}\mathcal{D}_{l,\theta+m-l+k}(y) & m,k \in \mathbb{N}, l \in \{1,\dots,m\} \\
		\\
		\frac{q^{\theta}_{m}(s)q^{\theta}_{k}(t-s)}{\sum_{d=1}^{\infty}q^{\theta}_{d}(t)(1-x)^{d}}(1-x)^{m} & m,k \in \mathbb{N}, l = 0 \\
		\\
		0 & \textnormal{otherwise}
	\end{cases}
\end{align*}
\noindent defining a pmf on a subspace of $\mathbb{N}^{3}$ (for more details, refer to \cite[Section 3]{JenkinsSpano}). Individual monotonic upper and lower bounds can be constructed for $\{q^{\theta}_{m}(s)\}_{m\in\mathbb{N}}$, $\{q^{\theta}_{k}(t-s)\}_{k\in\mathbb{N}}$ and $\sum_{d=1}^{\infty}q^{\theta}_{d}(t)(1-x)^{d}$ (see Section \ref{ExtensionOfDenominator} for full details with regards to this last quantity), and subsequently these can be put together to obtain monotonic upper and lower bounds on the $\{p_{m,k,l} \}_{m,k,l \in \mathbb{N}}$. Thus the alternating series trick lends itself to return a draw $(M,K,L) \sim \{p_{m,k,l}\}_{m,k,l,\in\mathbb{N}}$, and we use this to draw the relevant sample diffusion bridge point: 
\begin{enumerate}
	\item Sample $(M,K,L) \sim \{p_{m,k,l}\}_{m,k,l,\in\mathbb{N}}$
	\item If $L = 0$, return $Y = 0$, else return $Y \sim \textnormal{Beta}(l,\theta+m-l)$.
\end{enumerate}

\noindent A similar scheme can be derived for the case $\boldsymbol{\theta} = (\theta,0)$ by symmetric arguments, whereas for $\boldsymbol{\theta} = \boldsymbol{0}$ the above can be replicated with the only significant difference being that if $L = m$, then the routine returns $Y = 1$ in step 2. \newline

\noindent We now turn to the case when the diffusion bridge is conditioned on \emph{not} being absorbed within the time interval $(0,t)$. Corollary 2 in \cite{GriffithsJenkinsSpano} gives us that Wright--Fisher diffusion bridges with mutation parameters either $\boldsymbol{\theta} = \mathbf{0}$ or $\boldsymbol{\theta}=(0,\theta)$ are equal (in distribution) to Wright--Fisher bridges with mutation parameters $\boldsymbol{\theta}=(2,2)$ or $\boldsymbol{\theta}=(2,\theta)$ respectively. Thus from now on we shall focus our attention solely on the case when $\theta_1, \theta_2 > 0$. \newline

\noindent The strategy will be very close to the one developed above and based on the method found in \cite[Section 3]{JenkinsSpano}. As in the unconditioned bridge case, the diffusion bridge densities (PA1)--(PC3) can be viewed as mixtures of beta distributions, where the mixture weights now define a pmf on subspaces of $\mathbb{N}^{4}$ and whose exact form depends on the particular density being considered. \newline

\noindent As diffusion bridges are invariant under time reversal, a diffusion bridge that goes from $x$ to $y$ in time $s$ and then proceeds to terminate at $z$ at time $t$ has the same law as a diffusion bridge that starts at $z$, proceeds to $y$ at time $t-s$ and ends at $x$ at time $t$. This, coupled with symmetric arguments allows us to sample from the various transition densities (PA1)--(PC3) using just four different schemes, which we group as follows:

\begin{enumerate}
	\item Start and endpoints are the same (i.e.\ equations (PA1) and (PB3)). 
	\item Start and endpoints are opposite boundary points (i.e.\ equations (PA3) and (PB1)). 
	\item $z$ is in the interior of $[0,1]$, and the starting point is at one of the boundary points (i.e.\ equations (PA2), (PB2), (PC1) and (PC3) --- note that for (PC1) and (PC3) we make use of time reversal). 
	\item Start and endpoints are both inside the interior of $[0,1]$ (i.e.\ equation (PC2)).
\end{enumerate}
Using the above groupings, it remains to show that the resulting four different transition densities consist of mixture weights $\{p_{m,k,l,j}\}_{m,k,l,j\in\mathbb{N}}$ for which one can obtain monotonic sequences of upper and lower bounds. Again constructing these quantities for the numerator is straightforward, whereas the denominator is tackled in Section \ref{ExtensionOfDenominator} (by suitably modifying Proposition 4 from \cite{JenkinsSpano}). \newline

\noindent We point out that for both the diffusion and diffusion bridge case, numerical instabilities present when computing contributions to the infinite series representation of the probabilities $\{q_{m}^{\theta}(t)\}_{m\in\mathbb{N}}$ for small time increments prompt the use of approximations for these quantities. For more details, please refer to Section \ref{SmallTimes}.

\section{Simulation of non-neutral paths}\label{NonneutralPaths}
As observed in \cite[Section 5]{JenkinsSpano}, the neutral Wright--Fisher process can be used as a proposal distribution in an appropriate rejection sampler to returns exact draws from a non-neutral process. We give a brief overview for completeness. \newline

\noindent Denote by $\mathbb{WF}_{\sigma,\boldsymbol{\theta}}^{x_0}$ the law induced by the solution $X^T := (X_t)_{t=0}^{T}$ to the SDE given by equation (1) in the main paper on the space of continuous functions mapping $[0,T]$ into $[0,1]$ for some fixed time $T$, and by $\mathbb{WF}_{0,\boldsymbol{\theta}}^{x_0}$ the corresponding neutral law. The Radon--Nikodym derivative between these two laws is given by
\begin{align}\label{RN}
	\frac{d\mathbb{WF}_{\sigma,\boldsymbol{\theta}}^{x_0}}{d\mathbb{WF}_{0,\boldsymbol{\theta}}^{x_0}}(X^T) \propto \exp\left\{\tilde{A}(X_T)-\tilde{A}^{+}\right\}\exp\left\{-\int_{0}^{T}\left(\varphi(X_s)-\varphi^{-}\right)dt\right\}
\end{align}
where $\tilde{A}(x) := (\sigma / 2)\int_{0}^{x}\eta(z)dz$ with $\tilde{A}(x)\leq\tilde{A}^{+}$ for any $x\in[0,1]$, and 
\begin{align}
	\varphi(x) := \frac{\sigma}{4}\left[\left(-\theta_2 x + \theta_1(1-x)\right)\eta(x) + x(1-x)\left(\frac{\sigma}{2}\eta^{2}(x) + \eta'(x)\right)\right].
\end{align}
Observe that $\varphi(x)$ is a polynomial in $x$ (in view of $\eta(x)$ being a polynomial), and thus we can always find $\varphi^{-}$ and $\varphi^{+}$ such that $\varphi^{-}\leq \varphi(x) \leq\varphi^{+}$ on [0,1], and similarly for $\tilde{A}(x) \leq \tilde{A}^{+}$. The first term on the RHS of \eqref{RN} can be viewed as a simple $e^{\tilde{A}(X_T)-\tilde{A}^{+}}$-coin flip, whilst the second term is precisely the probability that all points in a unit rate Poisson point process $\Phi=\{(t_i,\omega_i)\}_{i\in\mathbb{N}}$ on $[0,T]\times[0,\infty)$ lie in the epigraph of the map $t \mapsto {\varphi(x)-\varphi^{-}}$. Furthermore, because $\varphi(x) \leq \varphi^{+}$, we can thin $\Phi$ to a Poisson point process on $[0,T]\times[0,\varphi^{+}-\varphi^{-}]$ and hence simulate an event with probability given by the RHS of \eqref{RN}. \newline

\noindent This allows for exact paths from the non-neutral Wright--Fisher process to be returned by first simulating the appropriate Poisson point process, subsequently generating draws from the neutral Wright--Fisher process at the time-stamps returned by the Poisson point process, checking whether the generated points all lie in the appropriate region, and and finally running a simple $e^{\tilde{A}(X_T)-\tilde{A}^{+}}$-coin flip. \newline

\noindent In order to calculate $\tilde{A}^{+}$, $\varphi^{-}$ and $\varphi^{+}$, a \texttt{Polynomial} class (with associated root finding algorithm implementing the Jenkins--Traub algorithm, developed by Bill Hallahan\footnote{https://www.codeproject.com/Articles/674149/A-Real-Polynomial-Class-with-Root-Finder}) was used. Whilst the implementation of this routine should work for polynomials of any degree, only polynomials $\eta(x)$ of degree at most 25 were allowed to ensure that the code returns reliable output within a reasonable amount of time.   

\section{Monotonic upper and lower bounds for the new denominators} \label{ExtensionOfDenominator}
\noindent In this section we show that the denominators in the transition densities for both the diffusion (equations \eqref{OneSidedMutDen} and \eqref{ZeroMutDen}) and the diffusion bridge (equations \eqref{TransitionDensityBridge(0,theta)}, \eqref{TransitionDensityBridge(0,0)z0} and \eqref{TransitionDensityBridge(0,0)z1}, as well as equations (PA1) through to (PC3)) allow for monotonic sequences of upper and lower bounds. By comparing \eqref{OneSidedMutDen} and \eqref{ZeroMutDen}, as well as \eqref{TransitionDensityBridge(0,theta)}, \eqref{TransitionDensityBridge(0,0)z0} and \eqref{TransitionDensityBridge(0,0)z1}, it becomes clear that we can consider solely the denominator $\sum_{d=2}^{\infty}q_{d}^{0}(t)(1-x^{d}-(1-x)^{d})$ as the proofs for the other quantities follow using near identical arguments. \newline

\noindent We further emphasise once more (as done in Section \ref{SamplingSchemes}), that for the bridge case we need only need consider the cases (PA1), (PA2), (PA3), and (PC2) in order to be able to simulate draws from any Wright--Fisher diffusion bridge process. Additionally, observe that the denominator for (PA3) is given by $q_{0}^{\theta}(t)$ for which monotonic upper and lower bounds are immediate, whereas (PC2) is precisely the case covered by Proposition 3 in \cite{JenkinsSpano}. It therefore remains to find monotonically converging sequences of upper and lower bounds for each of: \newline
\begin{tabularx}{\textwidth}{XX}
	\begin{equation}\label{ExtensionLine1a}
		\sum_{d=0}^{\infty}q_{d}^{\theta}(t)d,
	\end{equation}
	&
	\begin{equation}\label{ExtensionLine1b}
		\sum_{d=0}^{\infty}q_{d}^{\theta}(t)(1-x^{d}-(1-x)^{d}),
	\end{equation}
	\\
	
	\begin{equation}\label{ExtensionLine2a}
		\sum_{d=0}^{\infty}q_{d}^{\theta}(t)\frac{1}{B(\theta_1+d,\theta_2)},
	\end{equation}
	&
	\begin{equation}\label{ExtensionLine2b}
		\sum_{d=0}^{\infty}q_{d}^{0}(t)\frac{z^{\theta_1+d-1}(1-z)^{\theta_2-1}}{B(\theta_1+d,\theta_2)}.
	\end{equation}
\end{tabularx}\par

\noindent Further, by equation (5) in \cite{GriffithsJenkinsSpano}, \eqref{ExtensionLine1a} admits the required monotonic bounds through analytic expressions for the falling factorial moments of the ancestral process (see Theorem 5 in \cite{GriffithsJenkinsSpano} and the preceding paragraphs for full details). 

\subsection{Calculations for \eqref{ExtensionLine1b}} \label{DiffDenom}
\noindent Dealing with \eqref{ExtensionLine1b} requires some more work; we start by observing that $(1-x^m-(1-x)^m) = \sum_{l=1}^{m-1}\binom{m}{l}x^{l}(1-x)^{m-l}$. We can modify the arguments in Lemma 1 in \cite{JenkinsSpano} to deduce that for $L_m \sim \textnormal{Bin}(m,x)$ we have that
\begin{align}\label{Lemma1}
	\sum_{l=1}^{m}\mathbb{P}\left[L_{m+1} = l\right] \leq (x+2)\sum_{l=1}^{m-1}\mathbb{P}\left[L_{m} = l\right].
\end{align}
To see this, observe that for $l \leq \lfloor mx \rfloor$
\begin{align}\label{Lemma1P1}
	\mathbb{P}\left[L_{m+1}=l\right] = \frac{m+1}{m+1-l}(1-x)\mathbb{P}\left[L_{m}=l\right] \leq \mathbb{P}\left[L_{m}=l\right],
\end{align}
where in the last inequality we used the fact that $l \leq \lfloor mx \rfloor \leq mx$. When $l \geq \lfloor mx \rfloor$, we have that
\begin{align}\label{Lemma1P2}
	\mathbb{P}\left[L_{m+1} = l+1\right] = \frac{m+1}{l+1}x\mathbb{P}\left[L_{m}=l\right] \leq (x+1)\mathbb{P}\left[L_{m}=l\right]
\end{align}
by observing that when $mx > 1$, $\frac{m+1}{l+1} \leq \frac{m+1}{mx} \leq 1+\frac{1}{x}$, whereas for $mx \leq 1$, $\frac{m+1}{l+1} \leq m+1 \leq 1+\frac{1}{x}$. Summing together \eqref{Lemma1P1} and \eqref{Lemma1P2} (and noting the double counting happening at $\lfloor mx \rfloor$) gives the result. With this in hand we can apply Proposition 3 in \cite{JenkinsSpano}, this time setting $c_{k,m} := b_{k}^{(t,\theta)}(m)\sum_{l=1}^{m-1}\mathbb{P}[L_{m}=l]$, and replacing $K^{(x,z)}$ with $x+2$. \newline

\subsection{Calculations for \eqref{ExtensionLine2b}} \label{NewDenom1}
\noindent Note first that
\begin{align*}
	&\sum_{m=1}^{\infty}q_{m}^{0}(t)\frac{z^{\theta_1+m-1}(1-z)^{\theta_2-1}}{B(\theta_1+m,\theta_2)} \\
	&= \sum_{m=1}^{\infty}\left(\sum_{k=m}^{\infty}(-1)^{k-m}\frac{\theta+2k-1}{m!(k-m)!}\frac{(\theta+k+m-2)!}{(\theta+m-1)!}e^{-\frac{k(\theta+k-1)t}{2}}\right)\frac{z^{\theta_1+m-1}(1-z)^{\theta_2-1}}{B(\theta_1+m,\theta_2)},
\end{align*}
and observe that the terms inside the inner summation (excluding the factor $(-1)^{k-m}$) correspond to the terms $b_{k}^{(t,\theta)}(m)$ as defined in Proposition 1 in \cite{JenkinsSpano}. Let $c_{k,m} := b_{k}^{(t,\theta)}(m)\frac{z^{\theta_1+m-1}(1-z)^{\theta_2-1}}{B(\theta_1+m,\theta_2)}$, and observe that we can re-write the above as $\sum_{i=0}^{\infty}(-1)^{i}d_{i}$ with 
\begin{align}
	d_{2m} = \sum_{j=0}^{m}c_{m+j,m-j}, & & d_{2m+1} = \sum_{j=0}^{m}c_{m+1+j,m-j}.
\end{align}
For $\varepsilon\in(0,1)$ fixed, denote by 
\begin{align}
	E^{t} &:= \inf\left\{ m \geq 0 : 2j \geq C_{m-j}^{t} \textnormal{ for all } j=0,\dots,m \right\}, \\
	D_{\varepsilon}^{t,\theta} &:= \inf\left\{ k \geq \left(\frac{1}{t} - \frac{\theta+1}{2}\right) \vee 0 : (\theta+2k+1)e^{-\frac{(\theta+2k)t}{2}} < 1-\varepsilon \right\}.
\end{align}
\noindent Proposition 3 in \cite{JenkinsSpano} can be restated for the case we consider here as follows.
\begin{prop}\label{Prop1}
	For all $m > D_{\varepsilon}^{t,\theta} \vee E^{t} \vee \lfloor\frac{\theta+2}{\varepsilon(\theta_1+1)} - 1\rfloor$,
	\begin{align}
		d_{2m+2} < d_{2m+1} < d_{2m}.
	\end{align}
\end{prop}
\begin{proof}
	The proof proceeds as in \cite{JenkinsSpano}. As $m > E^{t}$, $2j \geq C_{m-j}^{t}$, and thus by Proposition 1 in \cite{JenkinsSpano} $b_{m+j+1}(m-j) < b_{m+j}(m-j)$. Multiplying both sides of the inequality by $\frac{z^{\theta_1+m-1}(1-z)^{\theta_2-1}}{B(\theta_1+m,\theta_2)}$ and summing over $j$ gives
	\begin{align*}
		d_{2m+1} = \sum_{j=0}^{m} c_{m+j+1,m-j} < \sum_{j=0}^{m} c_{m+j,m-j} = d_{2m}.
	\end{align*}
	The above reasoning also leads to
	\begin{align*}
		\sum_{j=1}^{m} c_{m+j+2,m-j} < \sum_{j=1}^{m}c_{m+j+1,m-j},
	\end{align*}
	which coupled with $c_{m+1,m+1} + c_{m+2,m} < c_{m+1,m}$ (which still needs to be proved) gives the required $d_{2m+2} < d_{2m+1}$. Now observe that
	\begin{align*}
		\frac{c_{k+1,m}}{c_{k,m}} = \frac{b_{k+1}^{(t,\theta)}(m)}{b_{k}^{(t,\theta)}(m)} = \frac{\theta+m+k-1}{k-m+1}\frac{\theta+2k+1}{\theta+2k-1}e^{-\frac{(\theta+2k)t}{2}} \leq (\theta+2k+1)e^{-\frac{(\theta+2k)t}{2}},
	\end{align*}
	\noindent setting $k=m+1$ and observing that $m > D_{\varepsilon}^{t}$, we get that $c_{m+2,m} < (1-\varepsilon)c_{m+1,m}$. Similarly
	\begin{align*}
		\frac{c_{m+1,m+1}}{c_{m+1,m}} = \frac{\theta+2m}{(m+1)(\theta+m)}z\frac{B(\theta_1+m,\theta_2)}{B(\theta_1+m+1,\theta_2)} \leq \frac{\theta+2}{(m+1)(\theta_1+1)} < \varepsilon
	\end{align*}
	\noindent if $m > \lfloor \frac{\theta+2}{\varepsilon(\theta_1+1)}-1\rfloor$. The result follows. 
\end{proof}

\subsection{Calculation for \eqref{ExtensionLine2a}} \label{NewDenom6}
\noindent The same arguments used above apply (omitting the presence of the $z^{\theta_1+d-1}(1-z)^{\theta_2-1}$, which simplifies the proof slightly).

\section{Approximations for small times and implementation}\label{SmallTimes}
\noindent Whenever the simulation time increments become too small, numerical instabilities crop up when computing contributions to the quantities $q_{m}^{\theta}(t)$. Thus (as done in \cite{JenkinsSpano}), adequate approximations are necessary which make use of the small time asymptotics of $q_{m}^{\theta}(t)$. Theorem 4 in \cite{Griffiths84} gives that as $t\rightarrow 0$, the ancestral block counting process of the coalescent is well approximated by a Gaussian random variable with mean
\begin{align*}
	\mu = \frac{2\eta}{t}, & & \textnormal{ where }\eta=\begin{cases} 1 & \beta = 0 \\ \frac{\beta}{e^{\beta}-1} & \beta\neq0 \end{cases}, & & \textnormal{ and } \beta = \frac{(\theta-1)t}{2},
\end{align*}
and variance
\begin{align*}
	\sigma^{2} = \begin{cases} \frac{2}{3t} & \beta=0 \\ \frac{2\eta}{t}(\frac{\eta+\beta}{\beta})^{2}\left(1+\frac{\eta}{\eta+\beta}-2\eta\right) & \beta\neq0 \end{cases}
\end{align*}
(note that Theorem 4 in \cite{Griffiths84} is missing a factor of $\beta^{-2}$). In light of this, whenever the time increment $t$ falls below a specific threshold $\varepsilon_{G}$, EWF makes use of the above Gaussian approximation, such that the probabilities $q_{m}^{\theta}(t)$ are replaced by their (suitably rounded) Gaussian counterparts. In the current implementation of EWF, the threshold $\varepsilon_{G}$ was set to 0.08 after extensive testing as it was found that such a cutoff ensured a suitable trade-off between retaining precision by employing the approximation only when necessary, and having a robust and efficient implementation. \newline

\noindent For the diffusion bridge case we apply similar approximations to both $q_{m}^{\theta}(s)$ and $q_{k}^{\theta}(t-s)$, but we also introduce an additional threshold $\varepsilon_{D} < \varepsilon_{G}$ below which we approximate draws for the law of a diffusion \emph{bridge} through draws from the law of a \emph{diffusion}. This is necessary due to the fact that the mean $\mu$ given above for the Gaussian approximation is inversely proportional to the time increment $t$. Thus if either of the time increments $s$ or $t-s$ is small, the pmf $\{p_{m,k,l,j}\}_{m,k,l,j\in\mathbb{N}}$ spreads out very thinly over $\mathbb{N}^{4}$ leading to a loss of precision due to the small quantities involved coupled with infeasible run times, even when the above illustrated Gaussian approximations are used. \newline

\noindent In such cases (i.e.\ $s < \varepsilon_{D}$ or $t-s < \varepsilon_{D}$), EWF first simulates a draw from the corresponding Wright--Fisher diffusion started at $x$ and sampled at time $s$, computes the increment between the generated draw $Y'$ and the start point $x$, and superimposes it onto a linear interpolation between the left and right end-points $x$ and $z$ to generate the required draw $Y$. The linear interpolation employed explicitly make use of time increments $s$ and $t-s$ to account for the fact that the returned draw $Y$ should come from a diffusion bridge starting at $x$ and ending at $z$, with appropriate mechanisms in place to ensure that the output remains within the interval $[0,1]$. When either $s \in [\varepsilon_{D},\varepsilon_{G})$ or $t-s \in [\varepsilon_{D},\varepsilon_{G})$, the above detailed (rounded) Gaussian approximations are used for the corresponding $\{q_{i}^{\theta}\}_{i\in\mathbb{N}}$ within the appropriate time interval, whilst the standard sampling scheme is used for time increments which exceed $\varepsilon_{G}$. A threshold of 0.008 was chosen for $\varepsilon_{D}$ following extensive testing, such that the resulting implementation of EWF retained robustness and efficiency and refrained from using such approximations unless their absence led to unfeasible run times. We mention that both thresholds can be altered if desired through the fields \texttt{g1984} (for $\varepsilon_{G}$) and \texttt{bridgethreshold} (for $\varepsilon_{D}$) of the \texttt{Options} class found in the \texttt{myHelpers.h} header file (although we would advise against this). \newline

\section{Output validation}
Output was validated by generating 10,000 samples for a wide variety of cases and subsequently comparing this to a truncation of the transition density by means of Kolmogorov--Smirnov test as well as QQ-plots. We point out that we present only neutral output here as the non-neutral output is generated using the same rejection procedure as used in \cite{JenkinsSpano}. \newline

\noindent To illustrate how the transition density was truncated, consider the case (PC2), which involves a sum over four indices, two of which are infinite. By using an iterative scheme, the mode over these four indices was found and its contribution to the density for a given point $y \in [0,1]$ was calculated. Subsequently the denominator of (PC2) was evaluated up to machine precision, and an appropriate truncation level was chosen by multiplying together the resulting denominator, the mode's contribution to the density and a tolerance parameter. Similar truncations were employed for all the other diffusion and diffusion bridge cases.

\subsection{Diffusions conditioned on non-absorption}
Samples were generated using 9 different parameters setups featuring starting points $x\in\{0,0.5,1\}$, sampling times $t \in \{0.01,0.05,0.5\}$, and mutation parameter $\boldsymbol{\theta} = (0,1)$. The output is plotted below, starting with the case when $x = 0$, with the sampling time increment $t$ increasing when going left to right across plots. All of the Kolmogorov--Smirnov tests and QQ-plots below confirm that the output is indeed coming from the correct distribution.

\begin{figure}[H]
	\centering
	\subfloat
	\centering{{\includegraphics[width=.28\linewidth]{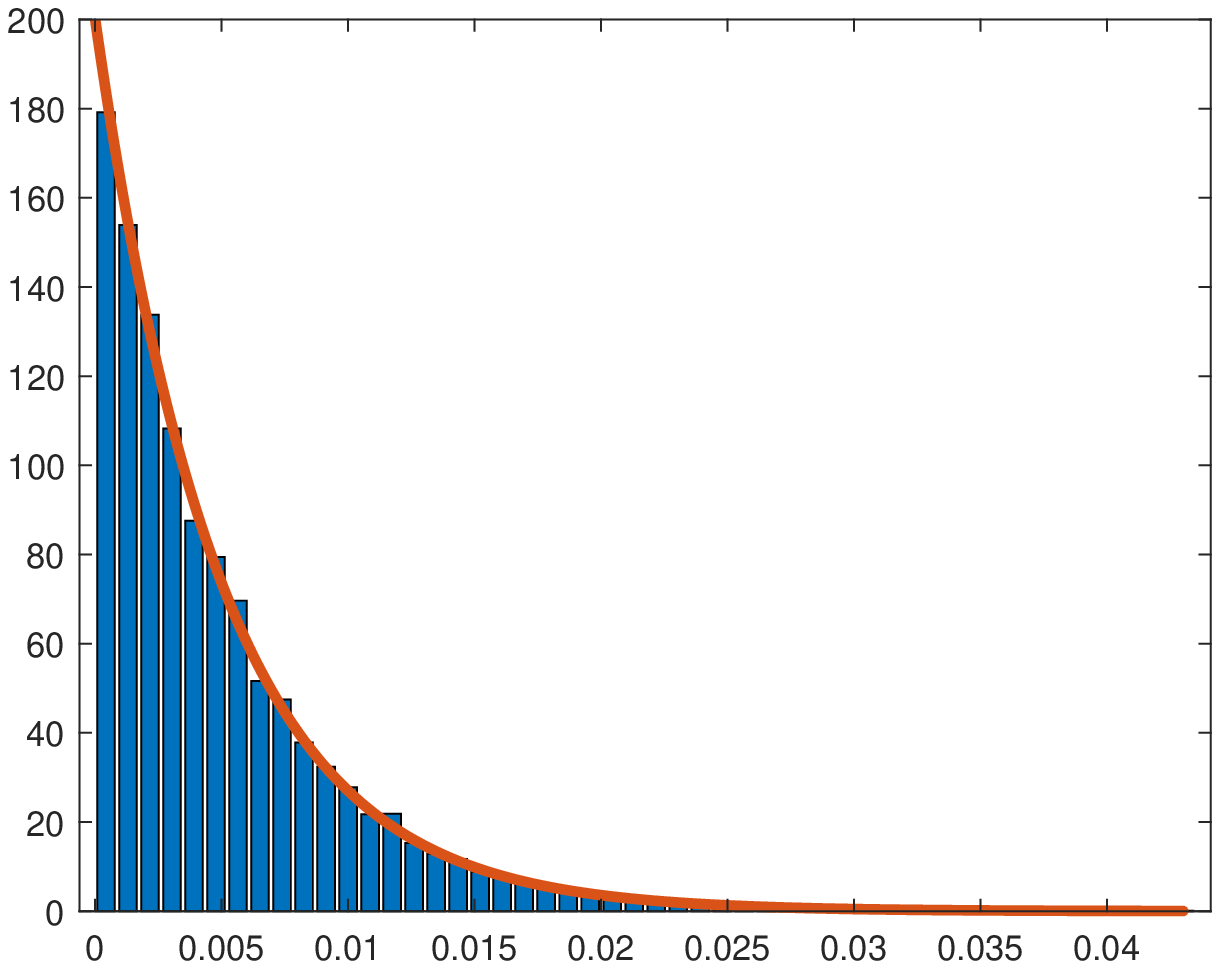} }} \hspace*{5mm}
	\subfloat
	\centering{{\includegraphics[width=.28\linewidth]{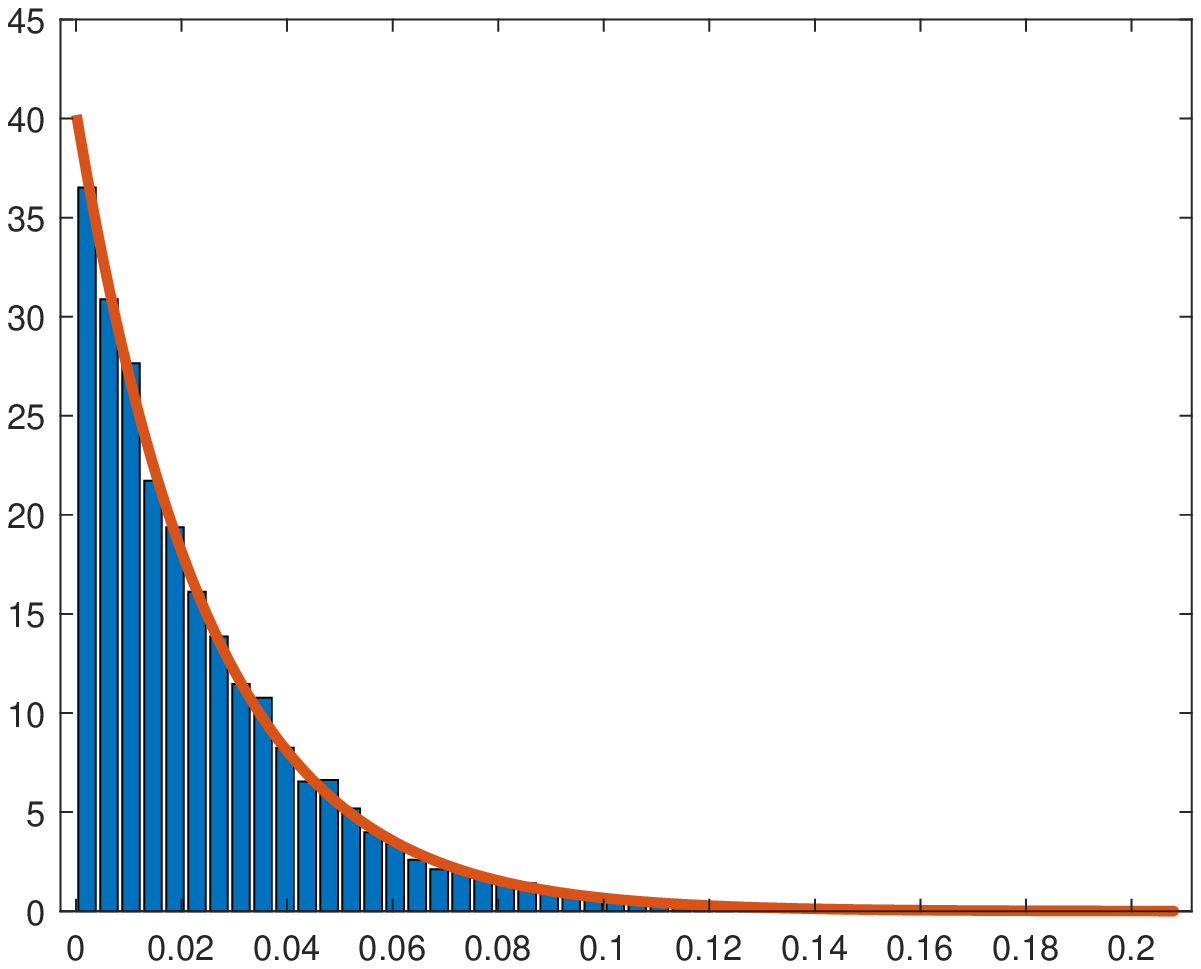} }} \hspace*{5mm}
	\subfloat
	\centering{{\includegraphics[width=.28\linewidth]{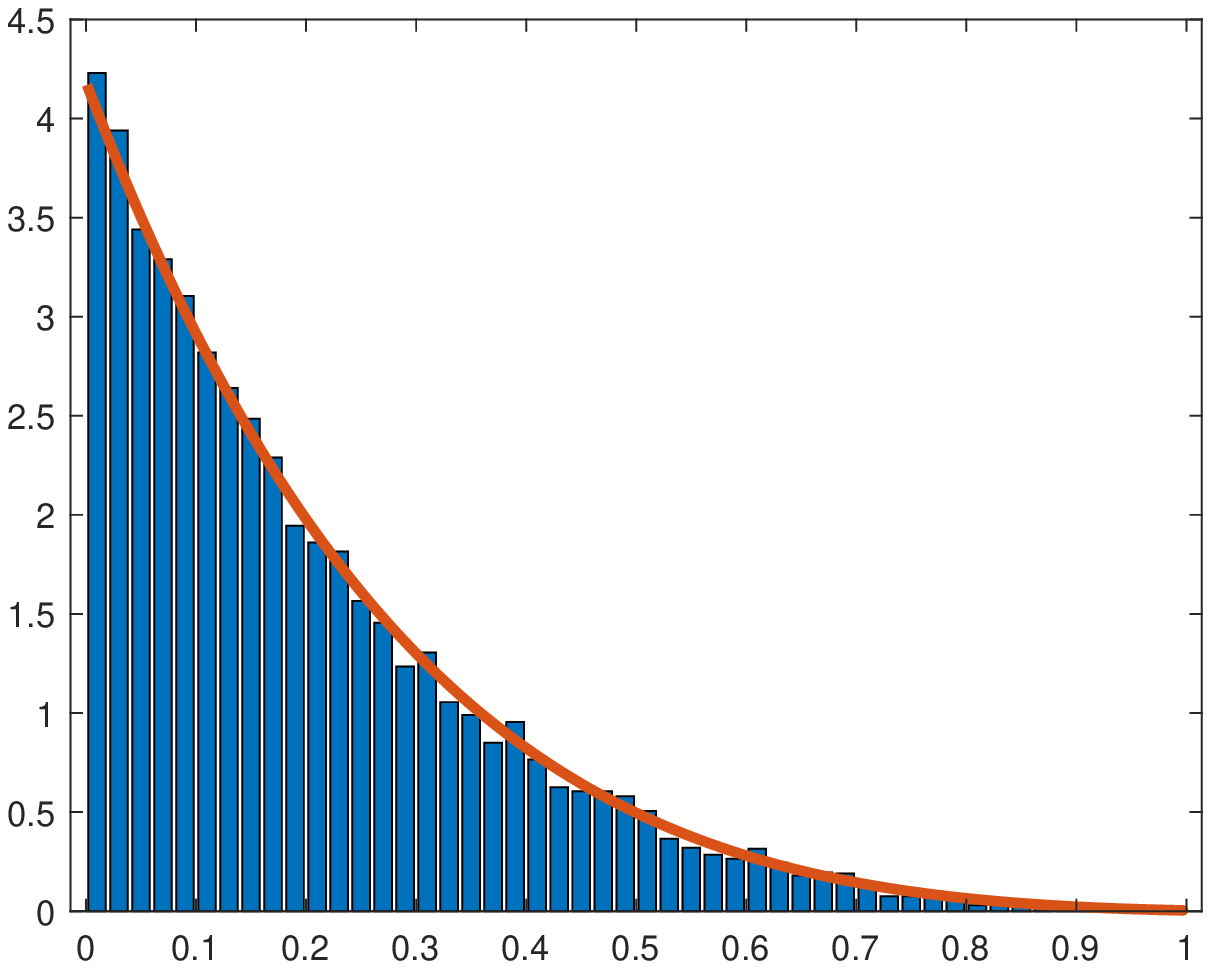}}}
	\\ \vspace*{5mm}
	\centering
	\subfloat
	\centering{{\includegraphics[width=.28\linewidth]{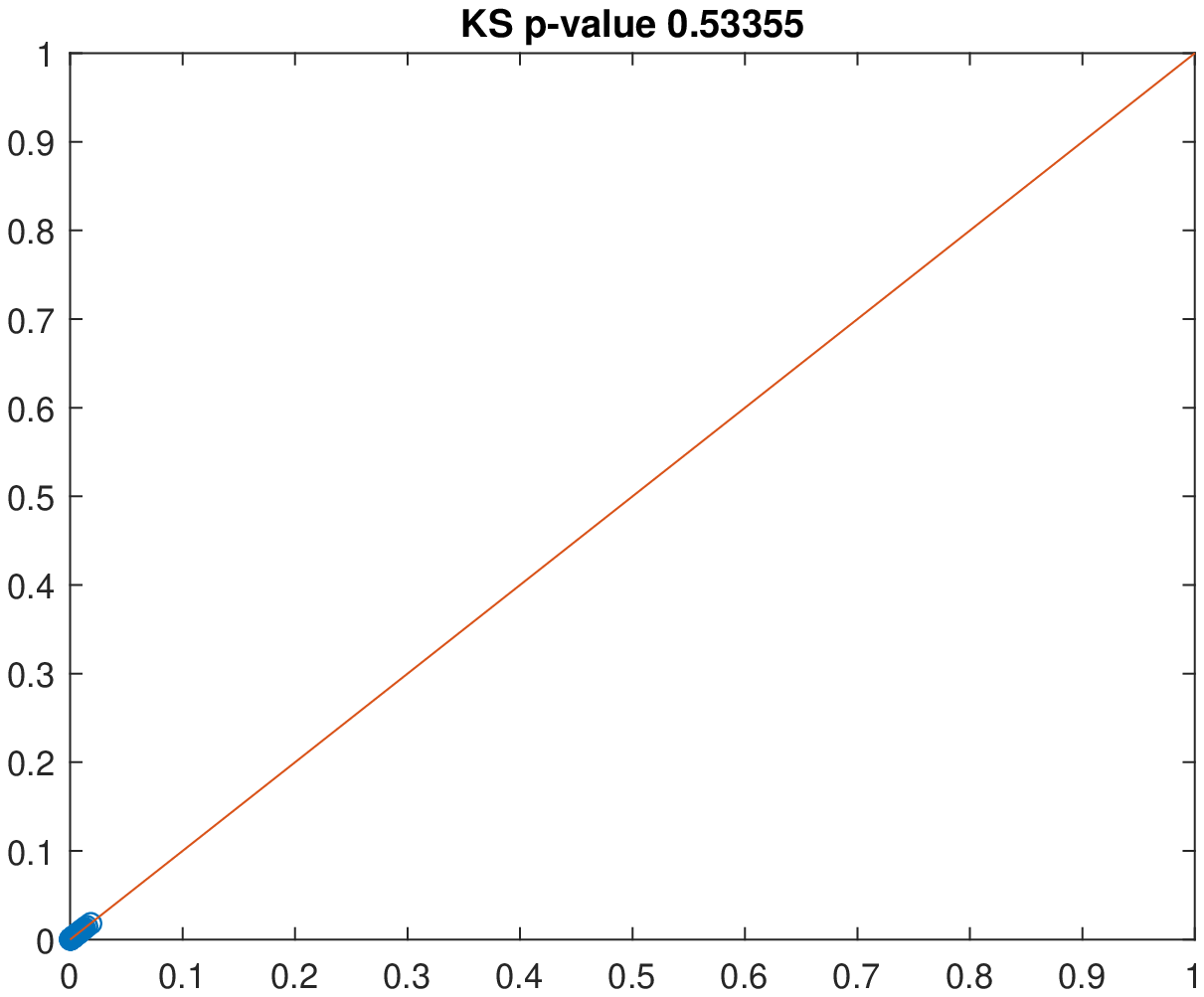} }} \hspace*{5mm}
	\subfloat
	\centering{{\includegraphics[width=.28\linewidth]{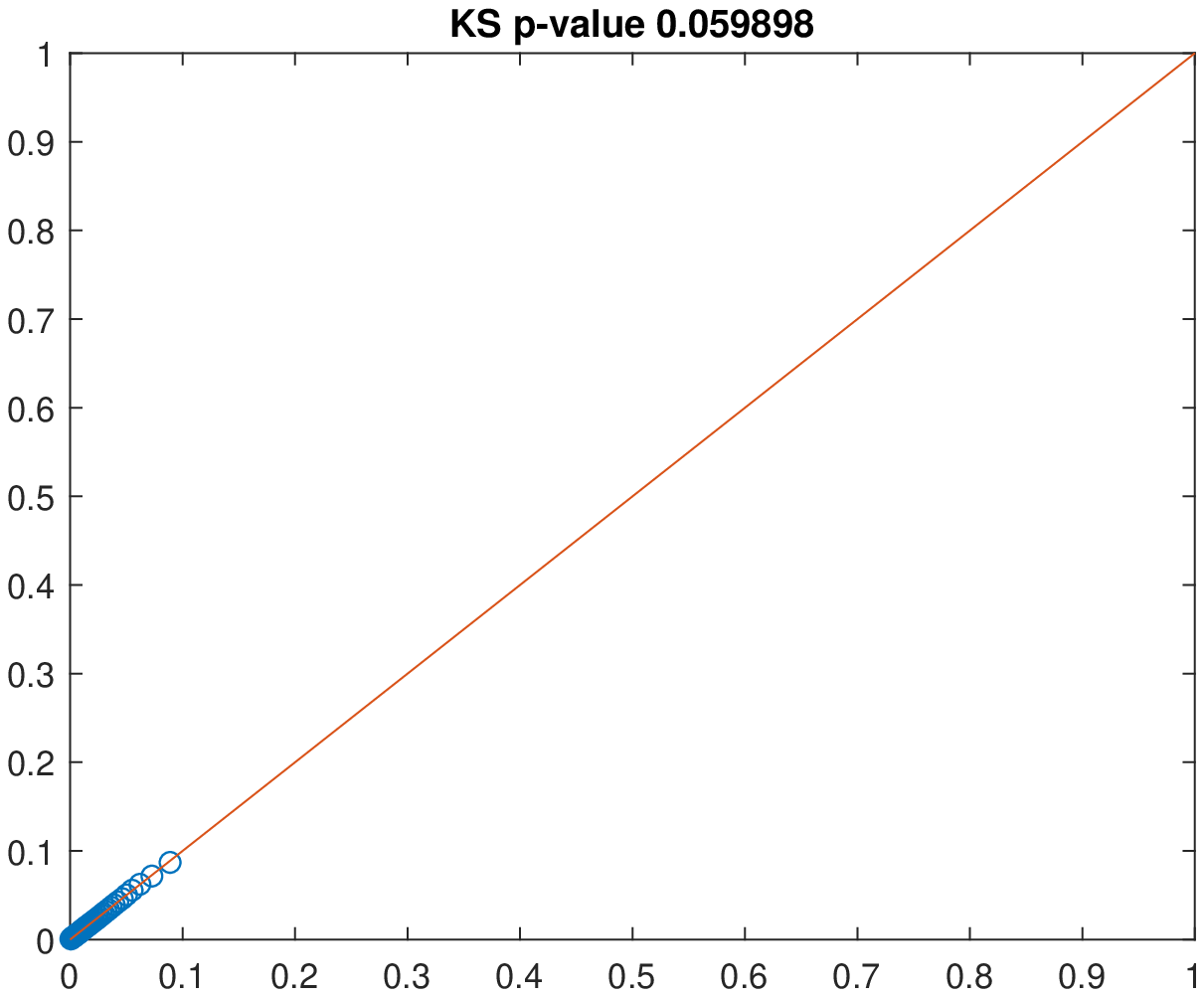} }} \hspace*{5mm}
	\subfloat
	\centering{{\includegraphics[width=.28\linewidth]{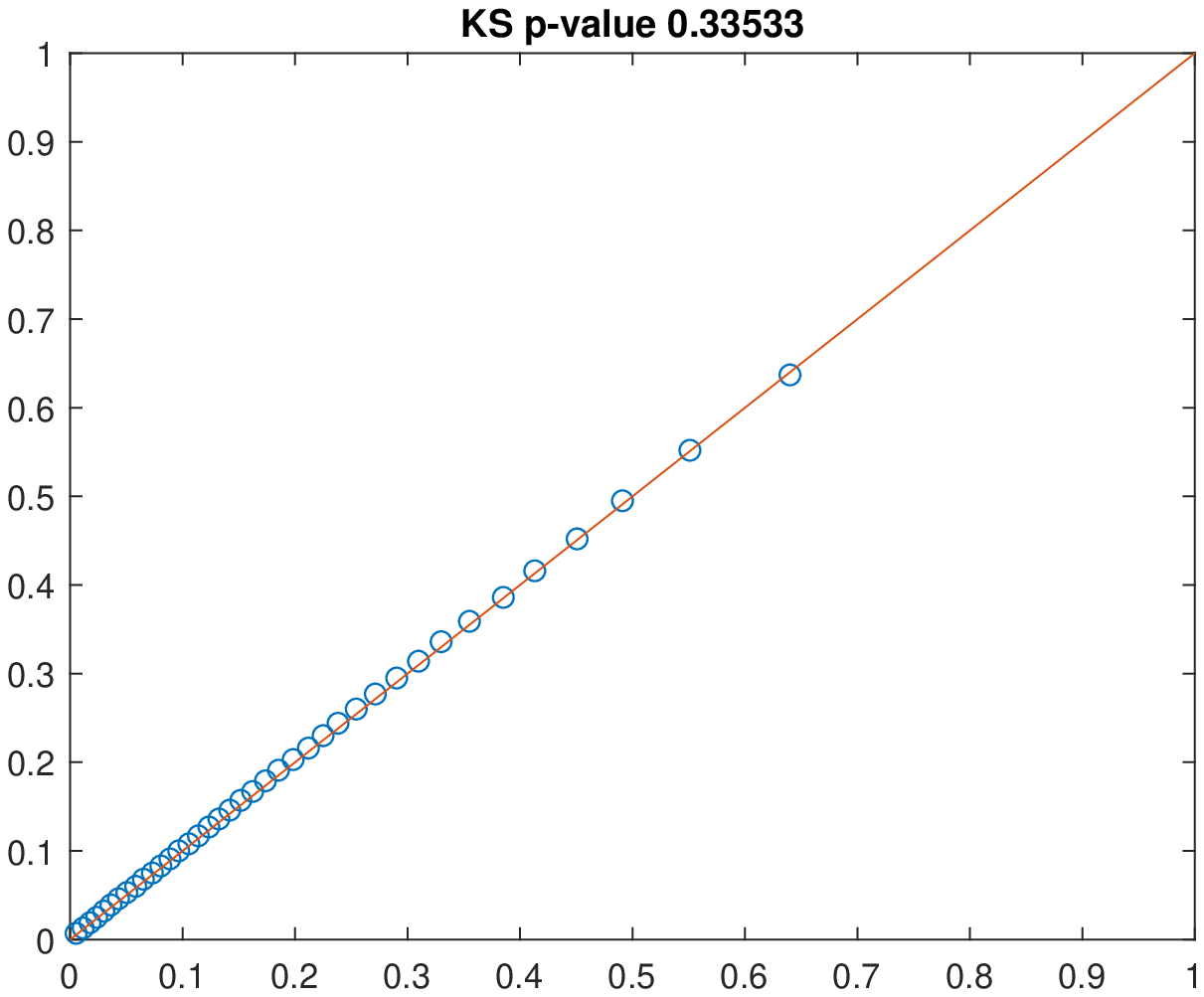}}}
	\caption{(Top row): Histograms for 10,000 samples generated from the law of a Wright--Fisher diffusion conditioned on non-absorption, started at $x=0$ at time 0, sampled at times $t=0.01,0.05,0.5$ respectively. The truncated transition density is plotted in red. (Bottom row): QQ-plots for the corresponding samples with the $p$-value returned from the Kolmogorov--Smirnov test reported above the plot.}
\end{figure}

\begin{figure}[H]
	\centering
	\subfloat
	\centering{{\includegraphics[width=.28\linewidth]{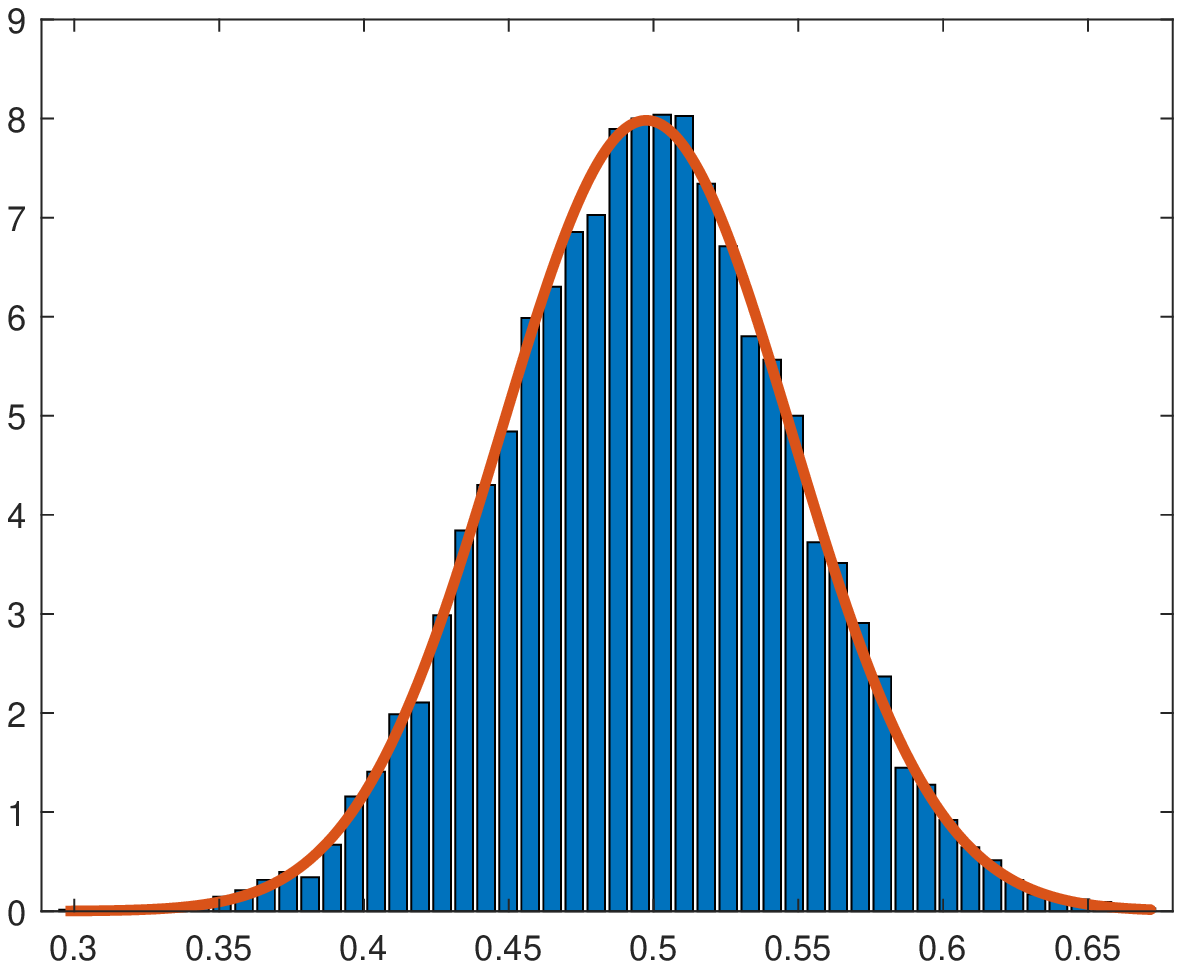} }} \hspace*{5mm}
	\subfloat
	\centering{{\includegraphics[width=.28\linewidth]{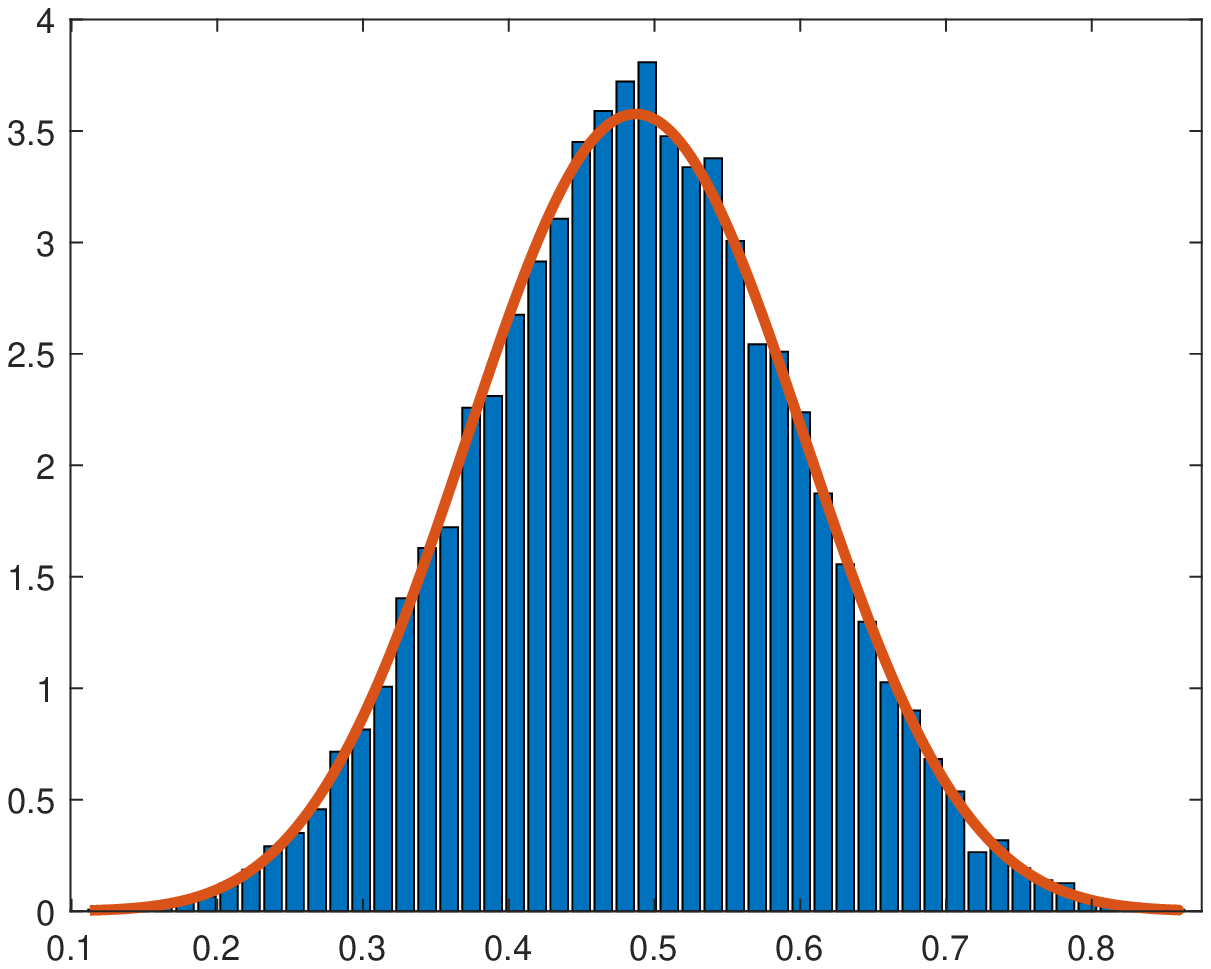} }} \hspace*{5mm}
	\subfloat
	\centering{{\includegraphics[width=.28\linewidth]{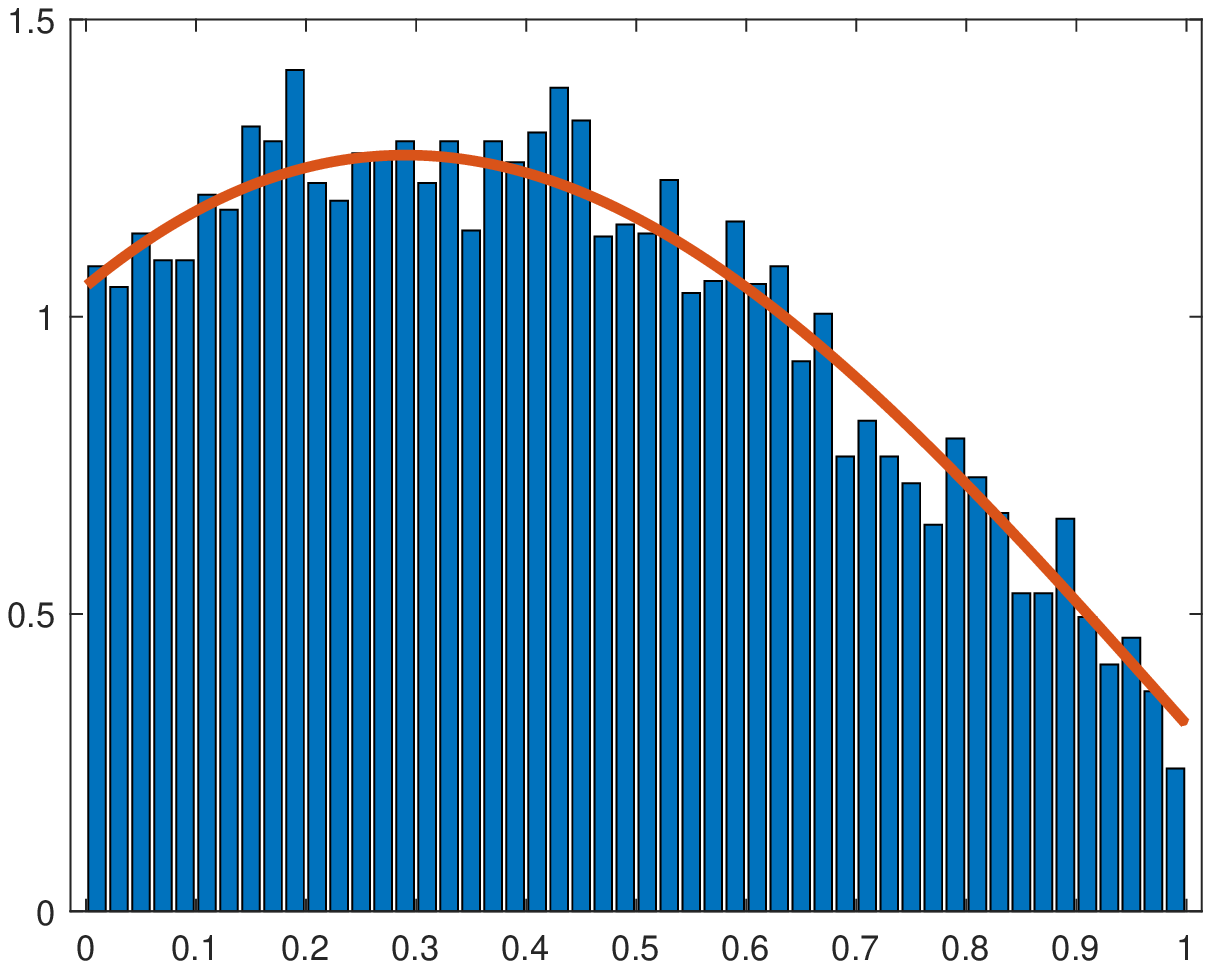}}}
	\\ \vspace*{5mm}
	\centering
	\subfloat
	\centering{{\includegraphics[width=.28\linewidth]{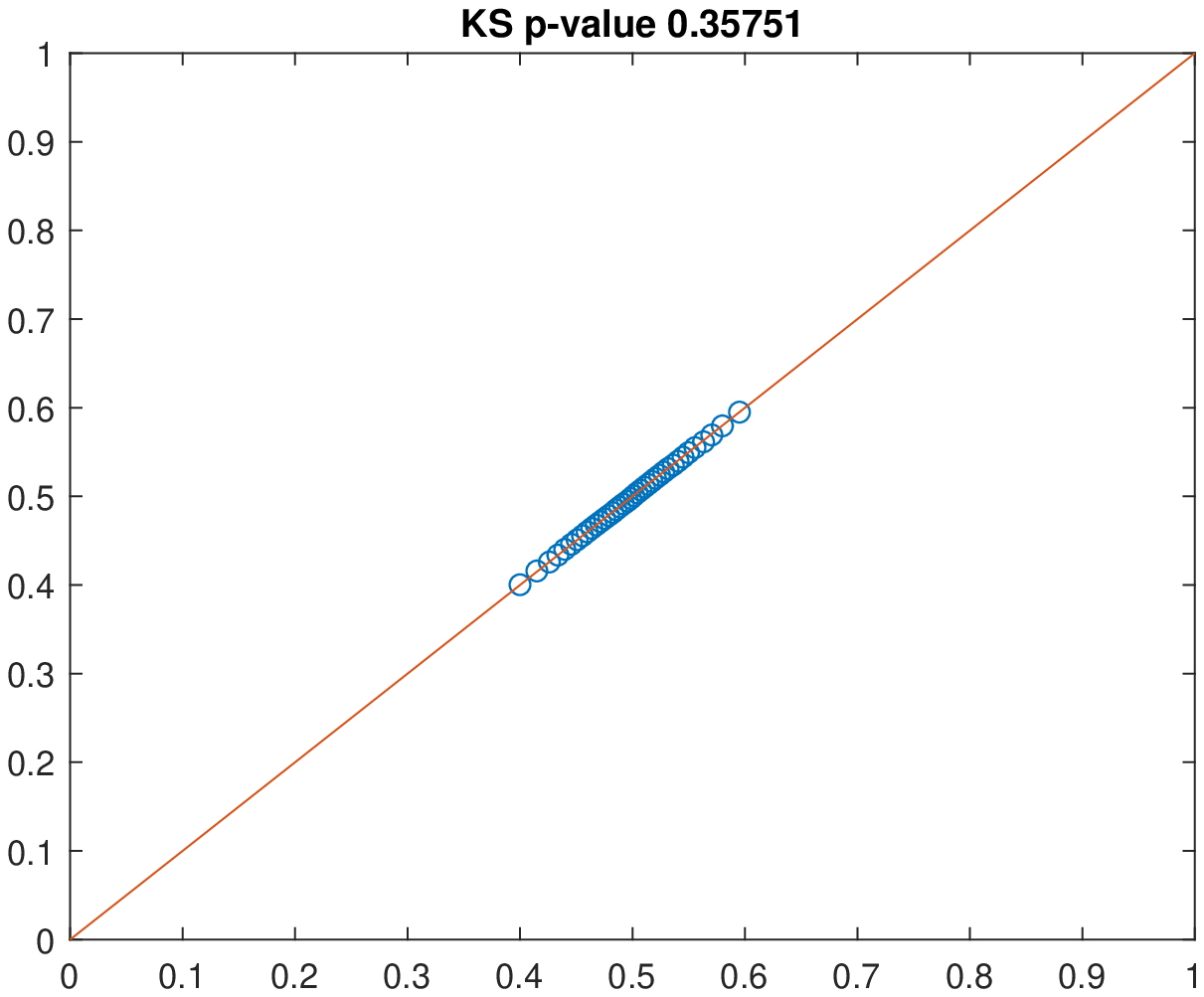} }} \hspace*{5mm}
	\subfloat
	\centering{{\includegraphics[width=.28\linewidth]{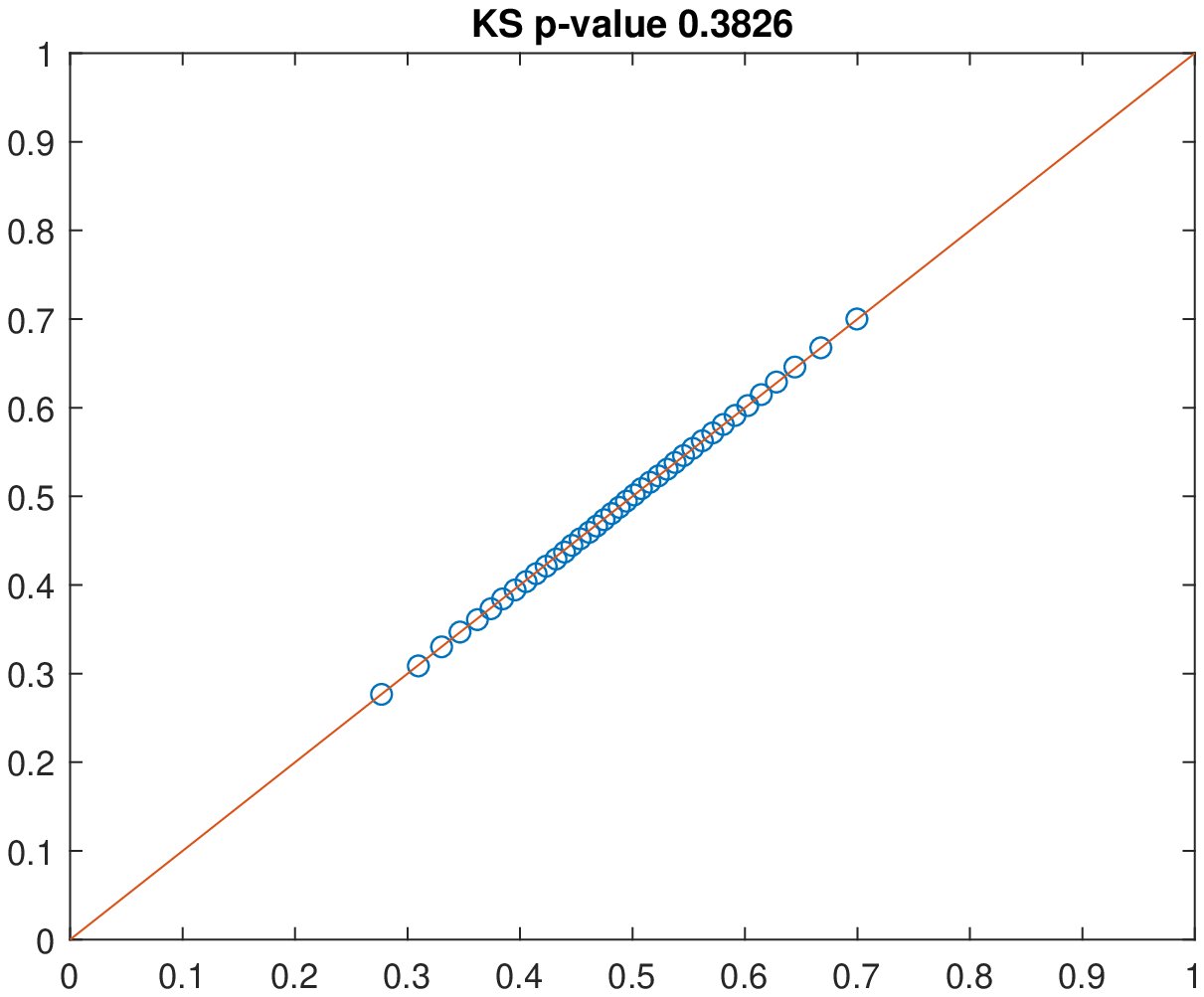} }} \hspace*{5mm}
	\subfloat
	\centering{{\includegraphics[width=.28\linewidth]{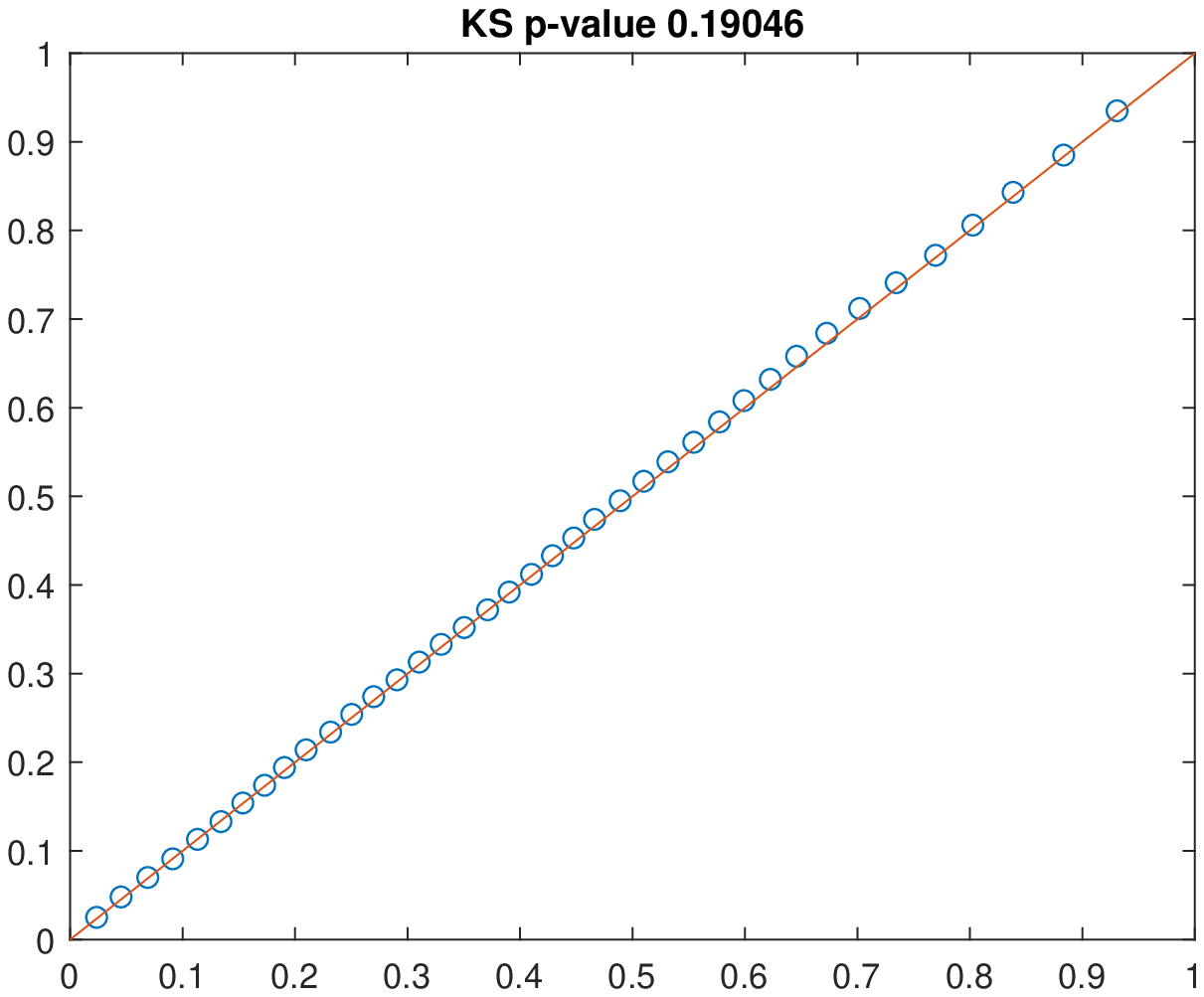}}}
	\caption{(Top row): Histograms for 10,000 samples generated from the law of a Wright--Fisher diffusion conditioned on non-absorption, started at $x=0.5$ at time 0, sampled at times $t=0.01,0.05,0.5$ respectively. The truncated transition density is plotted in red. (Bottom row): QQ-plots for the corresponding samples with the $p$-value returned from the Kolmogorov--Smirnov test reported above the plot.}
\end{figure}

\begin{figure}[H]
	\centering
	\subfloat
	\centering{{\includegraphics[width=.28\linewidth]{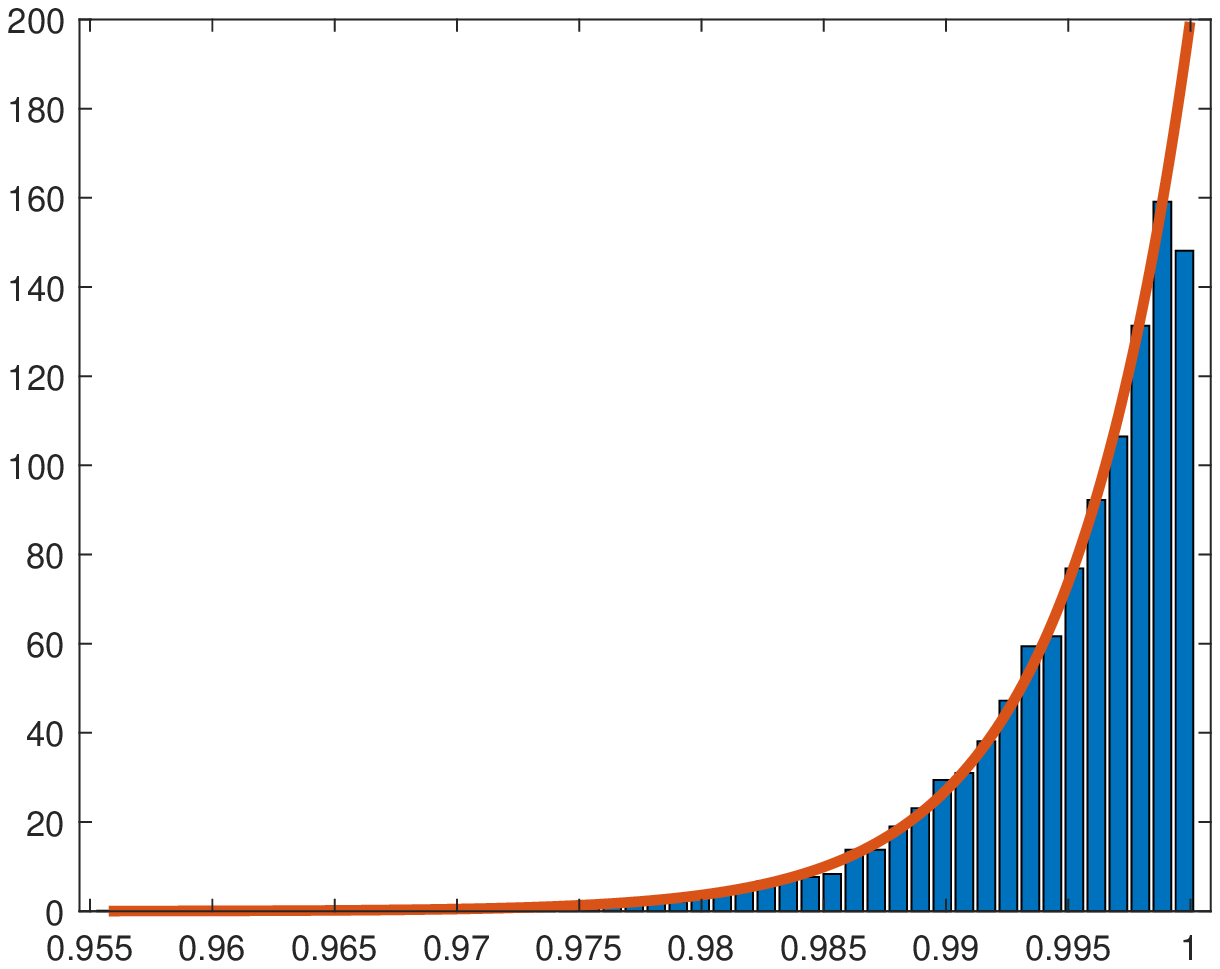} }} \hspace*{5mm}
	\subfloat
	\centering{{\includegraphics[width=.28\linewidth]{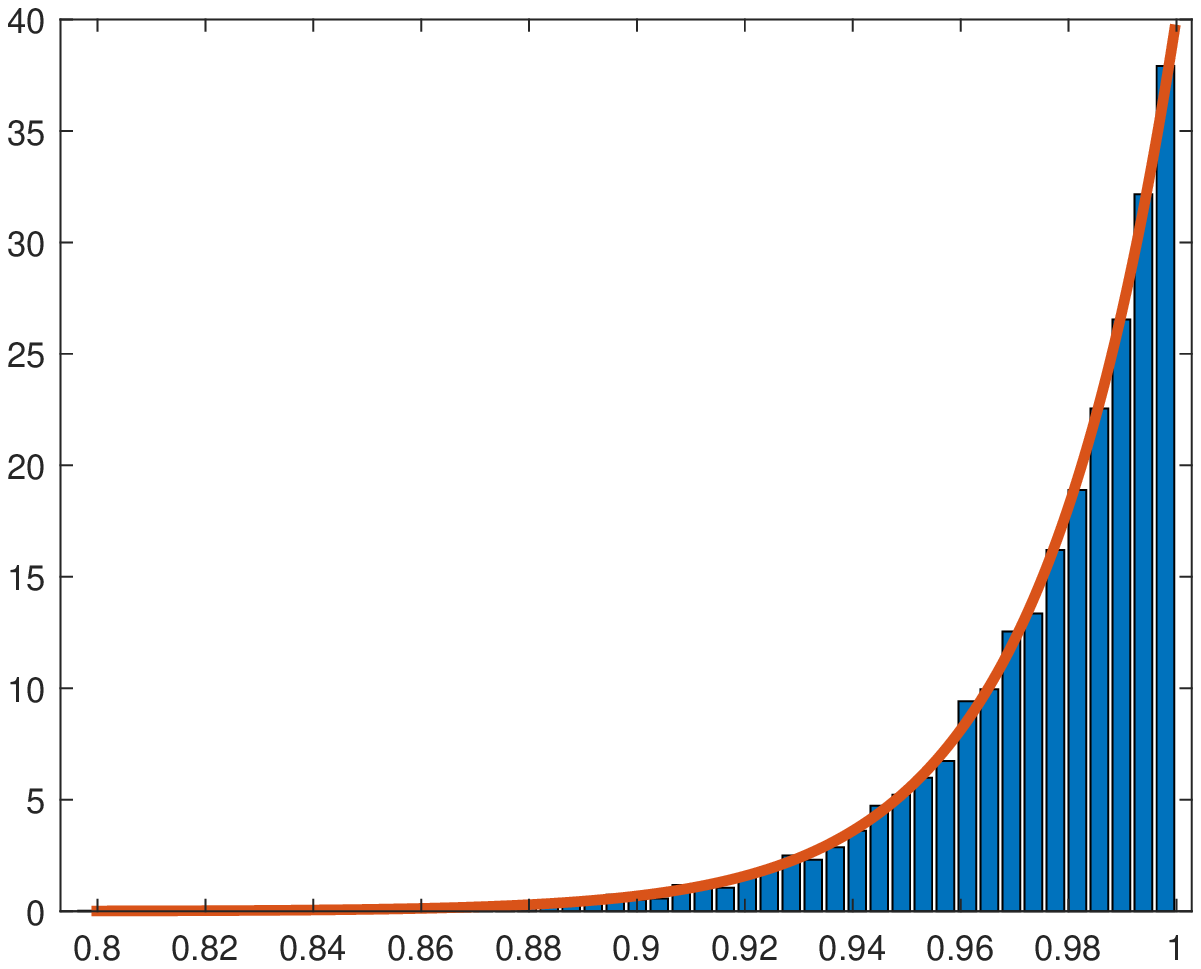} }} \hspace*{5mm}
	\subfloat
	\centering{{\includegraphics[width=.28\linewidth]{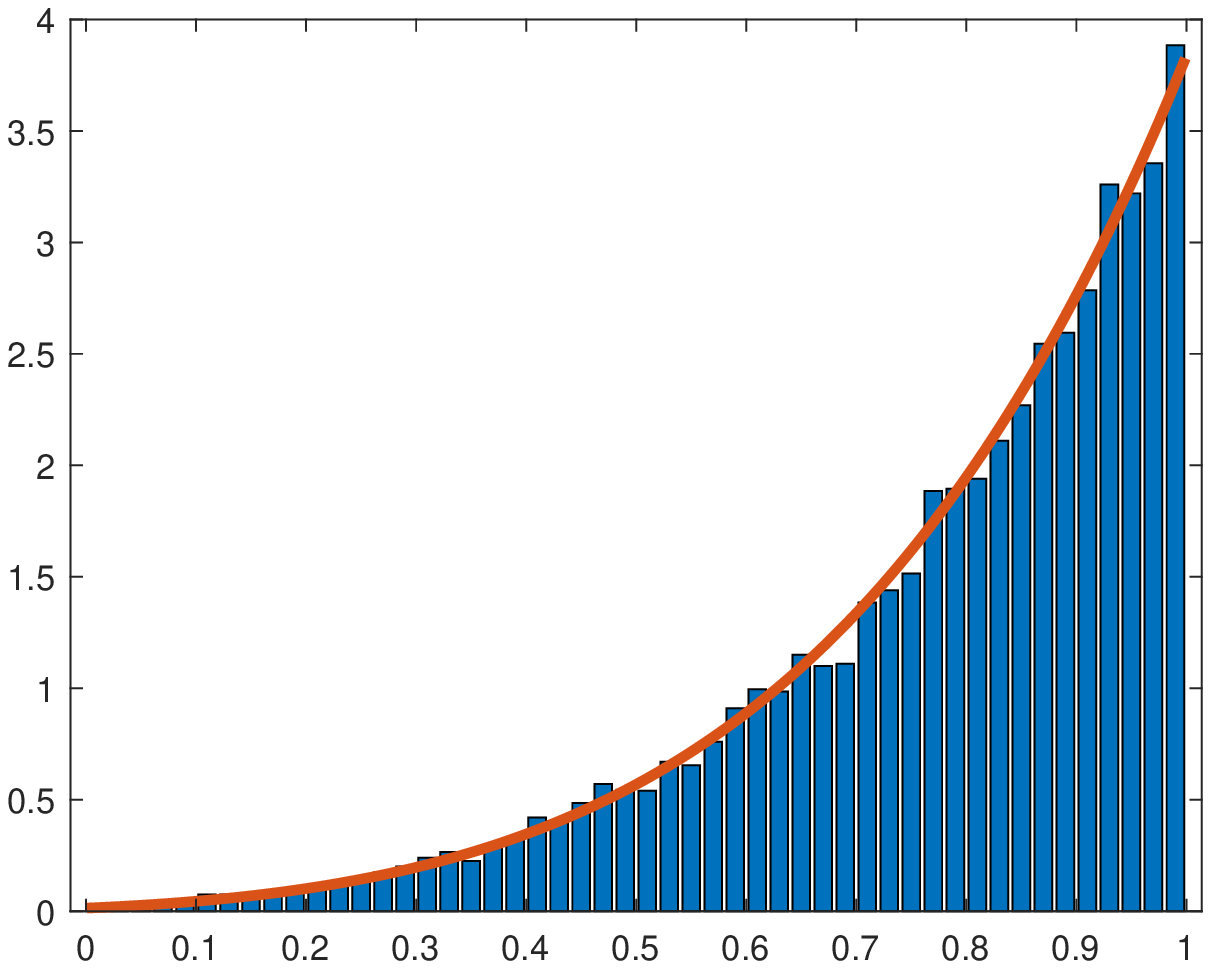}}}
	\\ \vspace*{5mm}
	\centering
	\subfloat
	\centering{{\includegraphics[width=.28\linewidth]{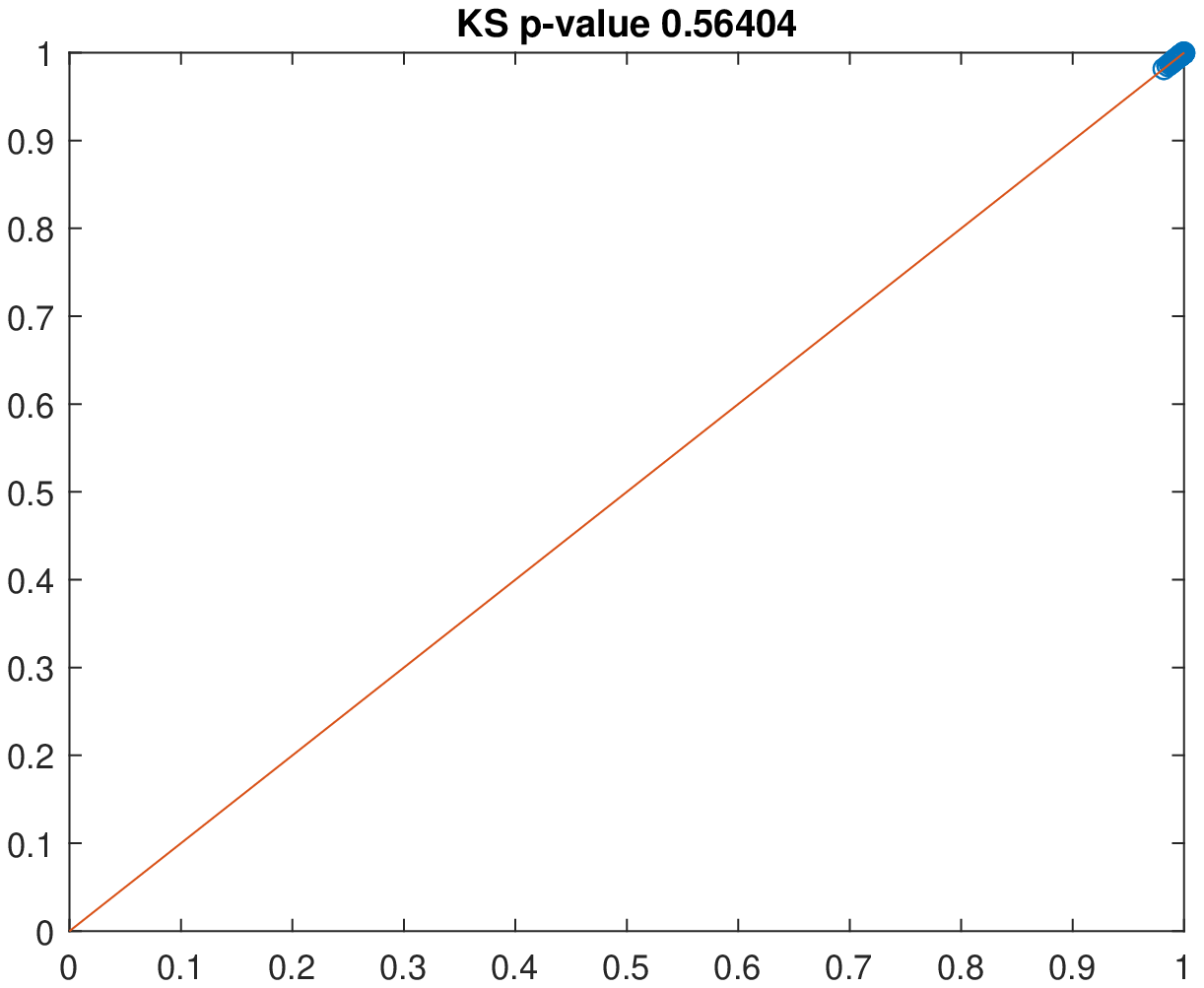} }} \hspace*{5mm}
	\subfloat
	\centering{{\includegraphics[width=.28\linewidth]{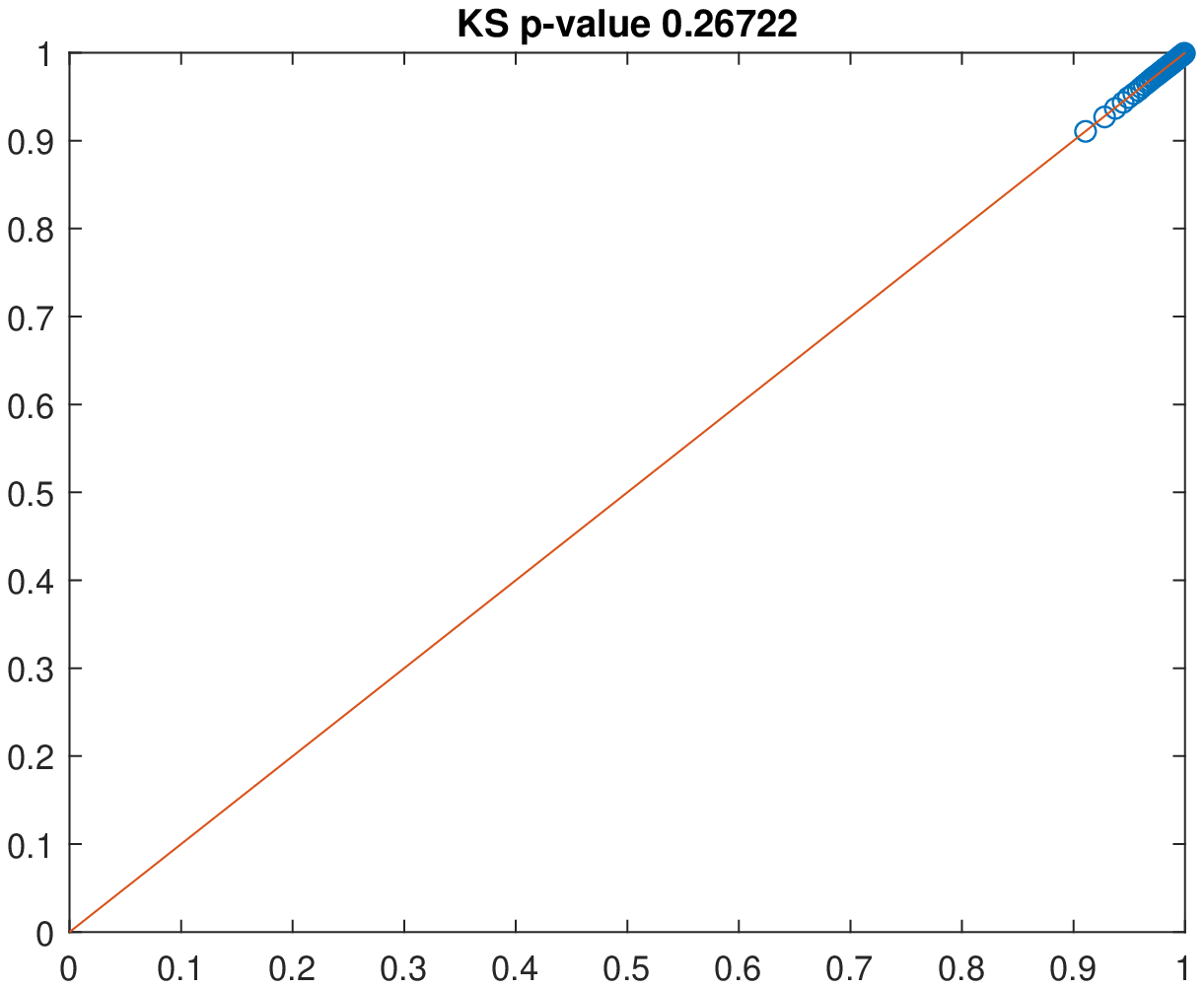} }} \hspace*{5mm}
	\subfloat
	\centering{{\includegraphics[width=.28\linewidth]{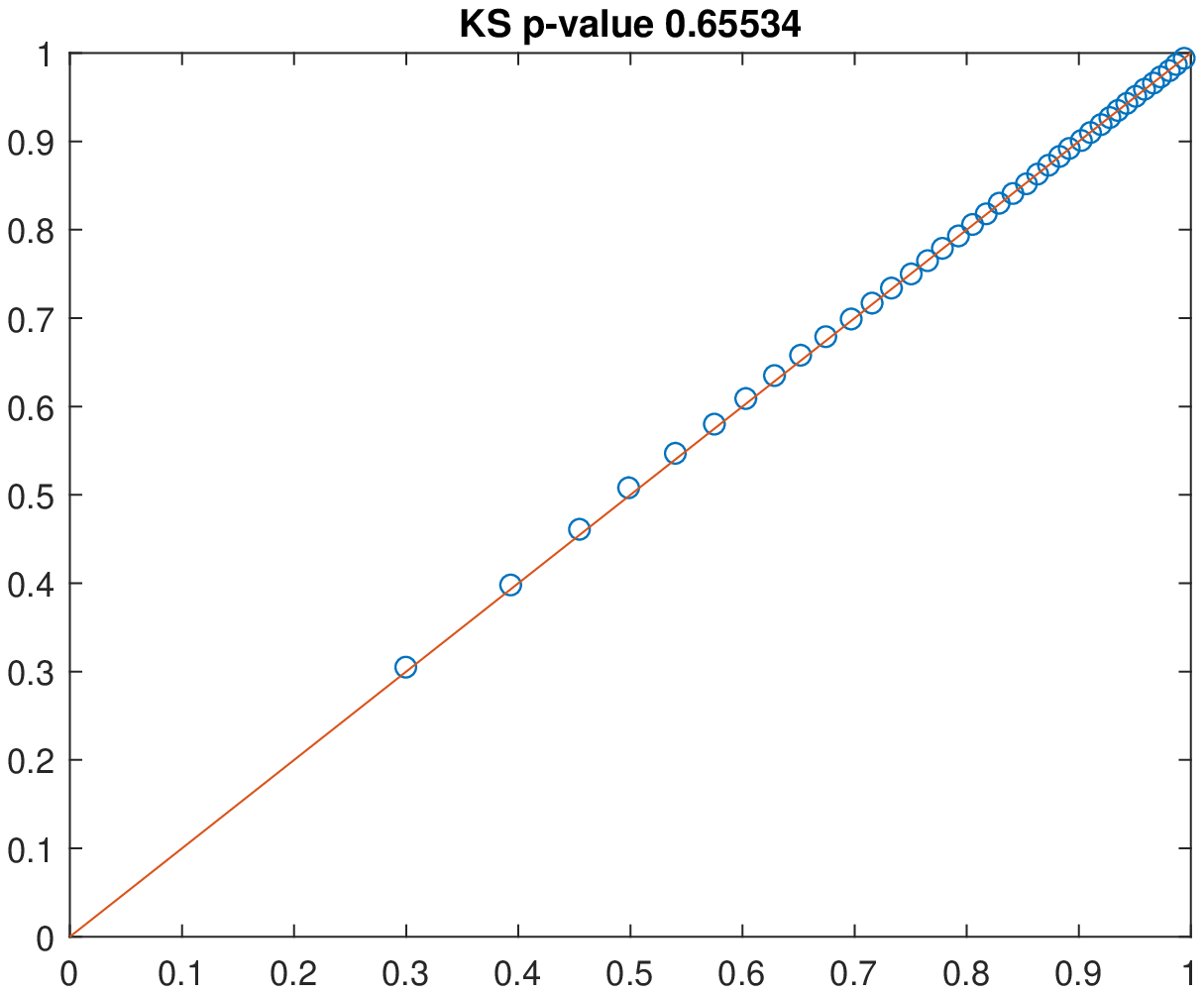}}}
	\caption{(Top row): Histograms for 10,000 samples generated from the law of a Wright--Fisher diffusion conditioned on non-absorption, started at $x=1$ at time 0, sampled at times $t=0.01,0.05,0.5$ respectively. The truncated transition density is plotted in red. (Bottom row): QQ-plots for the corresponding samples with the $p$-value returned from the Kolmogorov--Smirnov test reported above the plot.}
\end{figure}

\subsection{Unconditioned diffusions}
\noindent In the case when the diffusion was allowed to be absorbed at the boundaries, simulations for the following cases were obtained: start points $x\in\{0.25,0.5,0.75\}$, sampling times $t \in \{0.05,0.25, 0.5\}$, and mutation parameter $\boldsymbol{\theta} = \boldsymbol{0}$.  We report the probability of being absorbed at either boundary in the table below, where $\widehat{\mathbb{P}}$ denotes the empirical estimate for this quantity whereas $\mathbb{P}$ is the theoretical value obtained by evaluating the truncation to the transition density at the boundary. All of the estimated probabilities match their theoretical counterparts, and further both the QQ-plots and Kolmogorov--Smirnov tests confirm that the generated draws are coming from the correct distribution.

\begin{table}[H]
	\begin{center}
		\begin{tabular}{c|c c|c c}
			$x=0.25$ & $\widehat{\mathbb{P}}[\textnormal{Absorbed at 0}]$ & $\mathbb{P}[\textnormal{Absorbed at 0}]$ & $\widehat{\mathbb{P}}[\textnormal{Absorbed at 1}]$ & $\mathbb{P}[\textnormal{Absorbed at 1}]$ \\ \hline
			$t = 0.05$ & 0 & 1.51641e-5 & 0 & 8.92526e-26 \\
			$t = 0.25$ & 0.1025 & 0.101181 & 1e-4 & 9.79038e-5 \\
			$t = 0.5$ & 0.2923 & 0.302098 & 0.0074 & 0.0077254
		\end{tabular}
	\end{center}
	\caption{Empirical ($\widehat{\mathbb{P}}$) and theoretical ($\mathbb{P}$) absorption probabilities for the diffusion started at $x=0.25$.}
	\label{AbsorptionProbabilitiesD1}
\end{table}

\begin{table}[H]
	\begin{center}
		\begin{tabular}{c|c c|c c}
			$x=0.5$ & $\widehat{\mathbb{P}}[\textnormal{Absorbed at 0}]$ & $\mathbb{P}[\textnormal{Absorbed at 0}]$ & $\widehat{\mathbb{P}}[\textnormal{Absorbed at 1}]$ & $\mathbb{P}[\textnormal{Absorbed at 1}]$ \\ \hline
			$t = 0.05$ & 0 & 9.81343e-9 & 0 & 9.81343e-9 \\
			$t = 0.25$ & 0.0066 & 0.00569842 & 0.0064 & 0.00569842 \\
			$t = 0.5$ & 0.0687 & 0.066694 & 0.065 & 0.066694
		\end{tabular}
	\end{center}
	\caption{Empirical ($\widehat{\mathbb{P}}$) and theoretical ($\mathbb{P}$) absorption probabilities for the diffusion started at $x=0.5$.}
	\label{AbsorptionProbabilitiesD2}
\end{table}

\begin{table}[H]
	\begin{center}
		\begin{tabular}{c|c c|c c}
			$x=0.75$ & $\widehat{\mathbb{P}}[\textnormal{Absorbed at 0}]$ & $\mathbb{P}[\textnormal{Absorbed at 0}]$ & $\widehat{\mathbb{P}}[\textnormal{Absorbed at 1}]$ & $\mathbb{P}[\textnormal{Absorbed at 1}]$ \\ \hline
			$t = 0.05$ & 0 & 8.92526e-26 & 0 & 1.51641e-5 \\
			$t = 0.25$ & 1e-4 & 9.79038e-5 & 0.0986 & 0.101181 \\
			$t = 0.5$ & 0.0099 & 0.0077254 & 0.2979 & 0.302098
		\end{tabular}
	\end{center}
	\caption{Empirical ($\widehat{\mathbb{P}}$) and theoretical ($\mathbb{P}$) absorption probabilities for the diffusion started at $x=0.75$.}
	\label{AbsorptionProbabilitiesD3}
\end{table}

\begin{figure}[H]
	\centering
	\subfloat
	\centering{{\includegraphics[width=.28\linewidth]{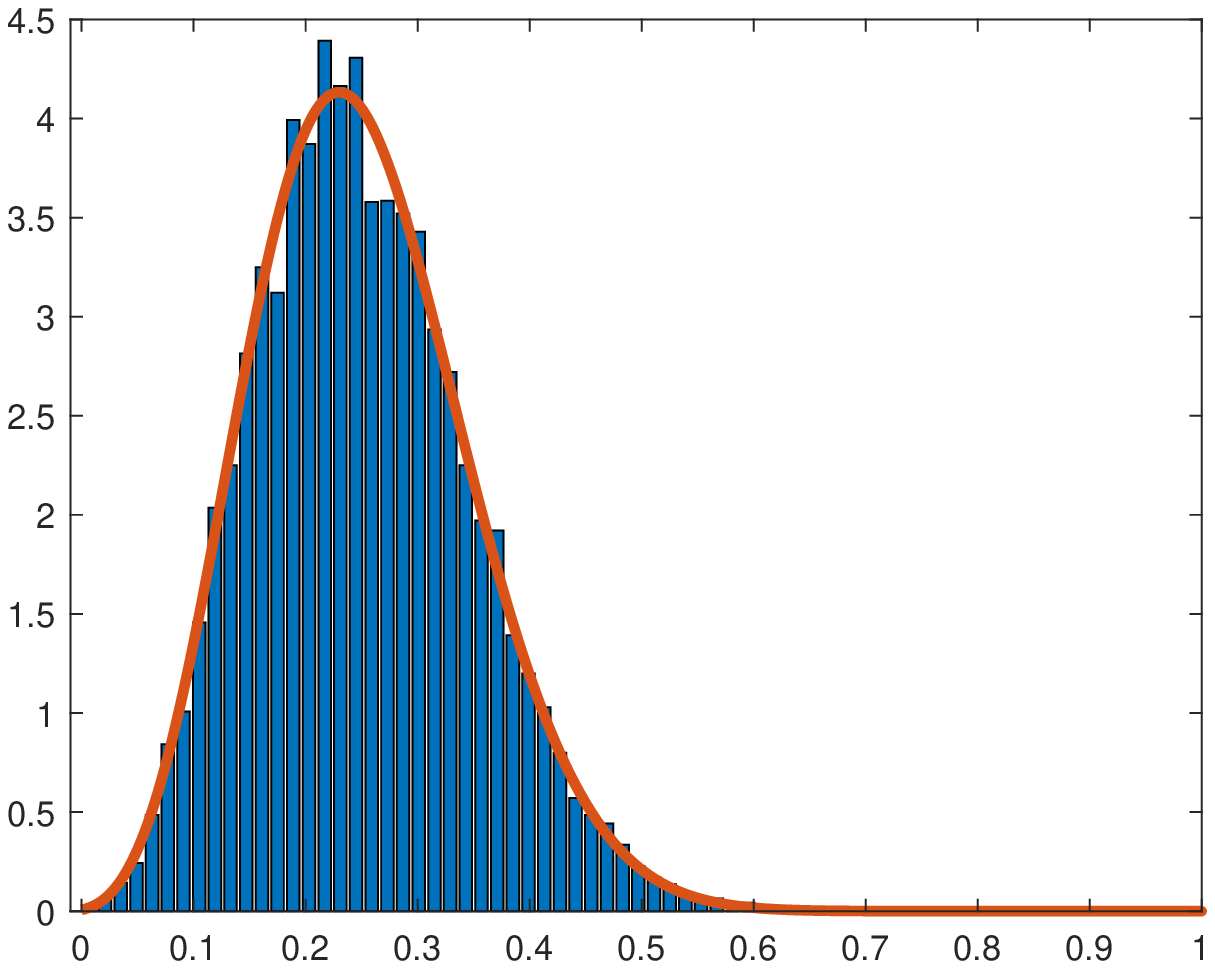} }} \hspace*{5mm}
	\subfloat
	\centering{{\includegraphics[width=.28\linewidth]{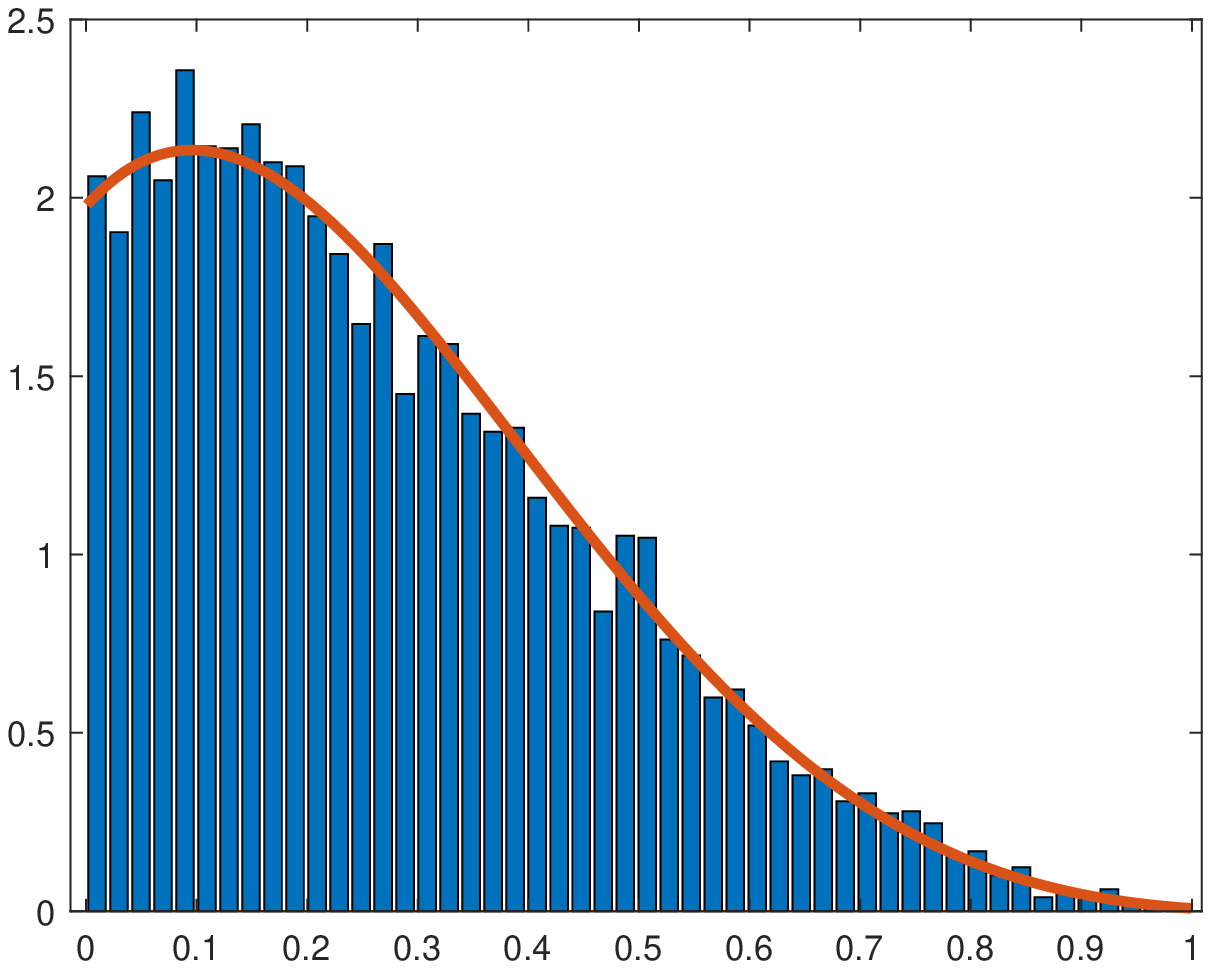} }} \hspace*{5mm}
	\subfloat
	\centering{{\includegraphics[width=.28\linewidth]{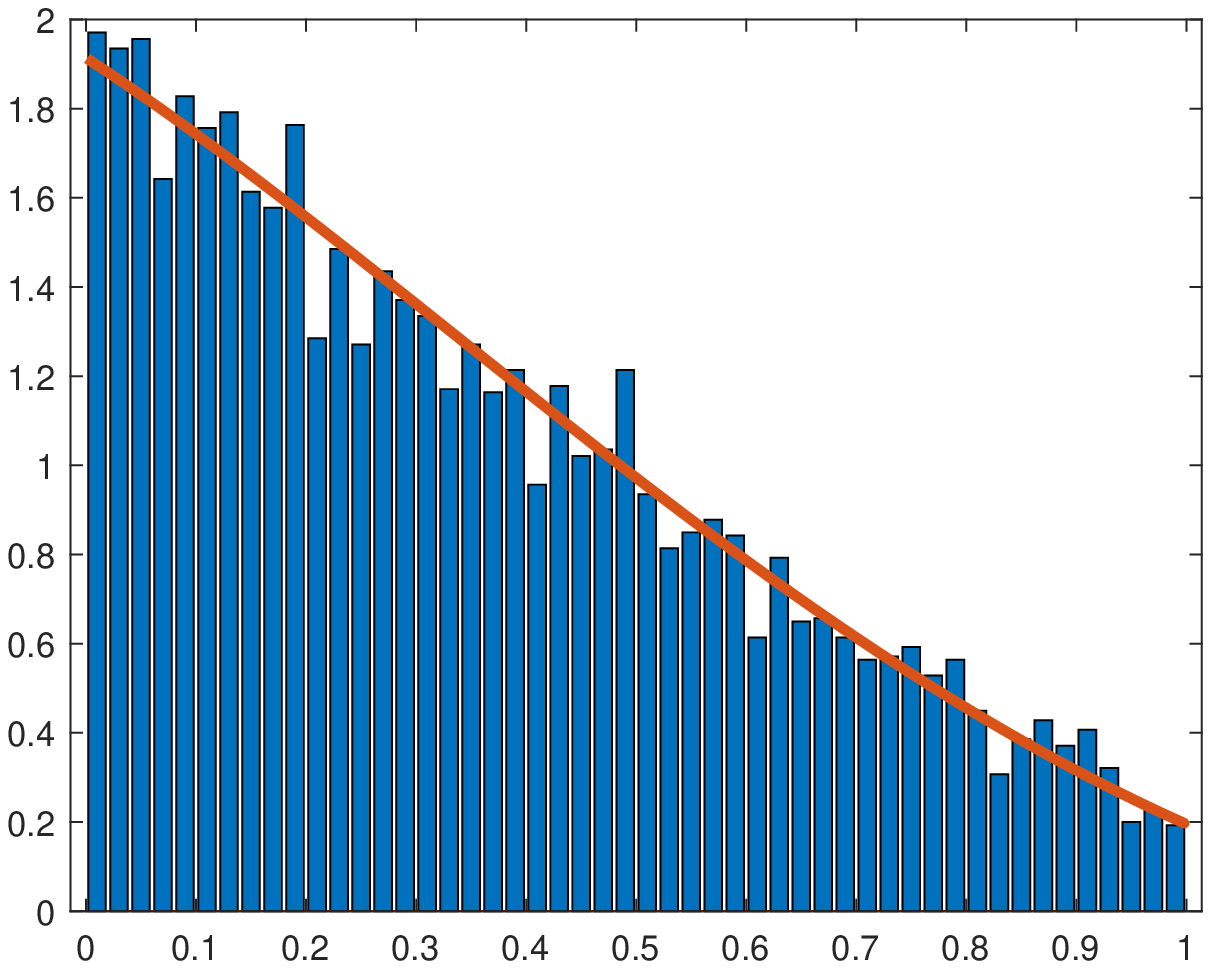}}}
	\\ \vspace*{5mm}
	\centering
	\subfloat
	\centering{{\includegraphics[width=.28\linewidth]{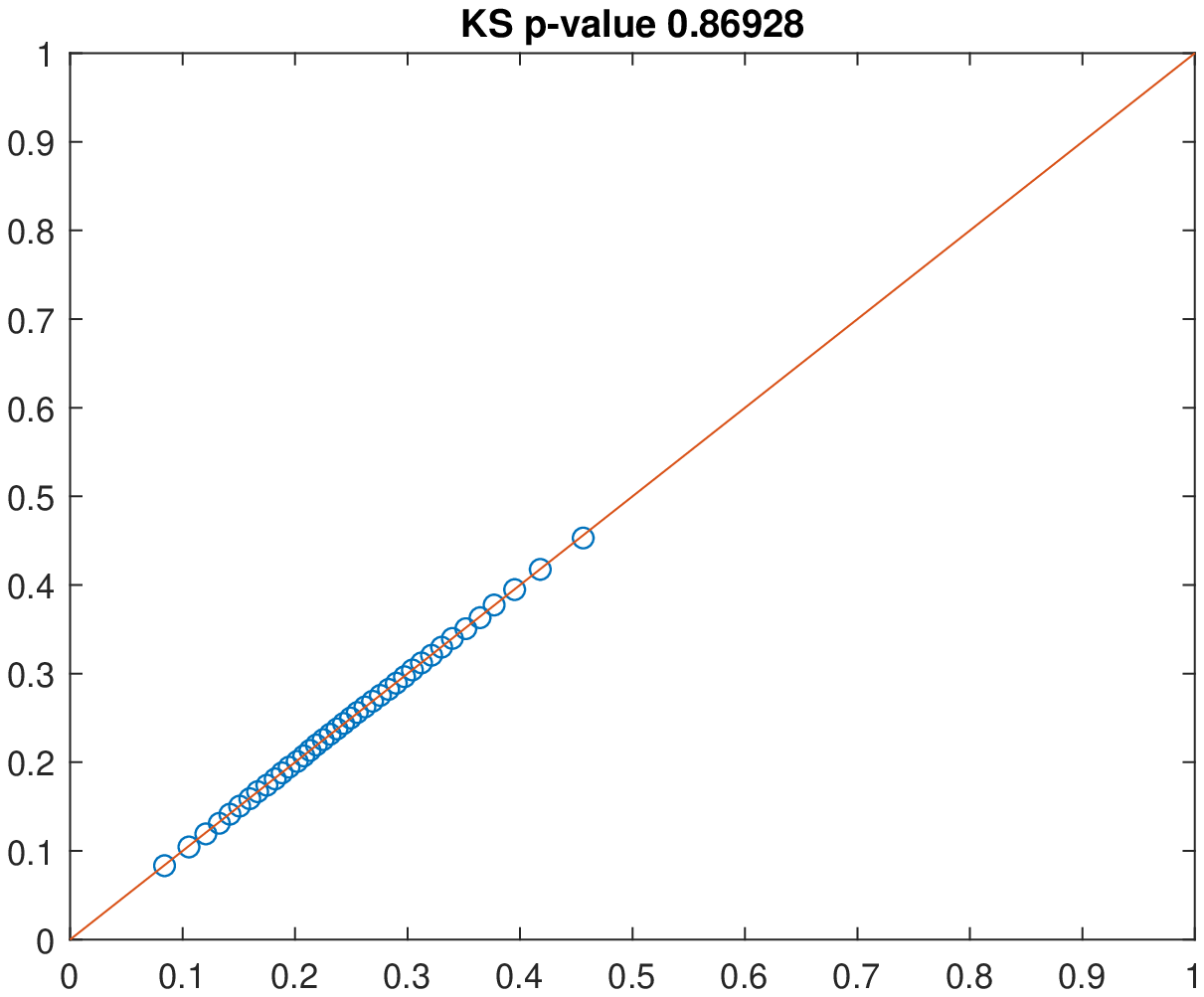} }} \hspace*{5mm}
	\subfloat
	\centering{{\includegraphics[width=.28\linewidth]{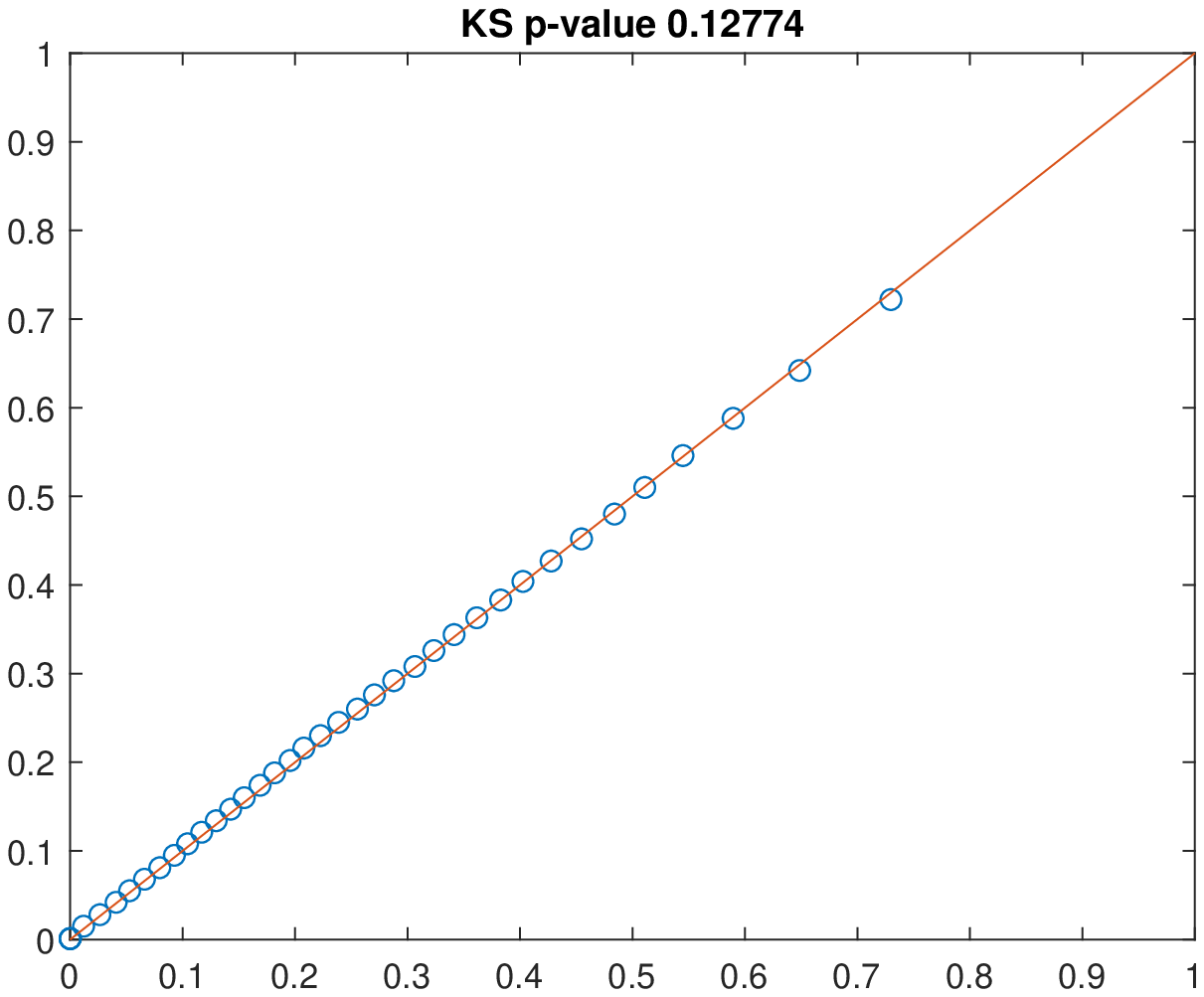} }} \hspace*{5mm}
	\subfloat
	\centering{{\includegraphics[width=.28\linewidth]{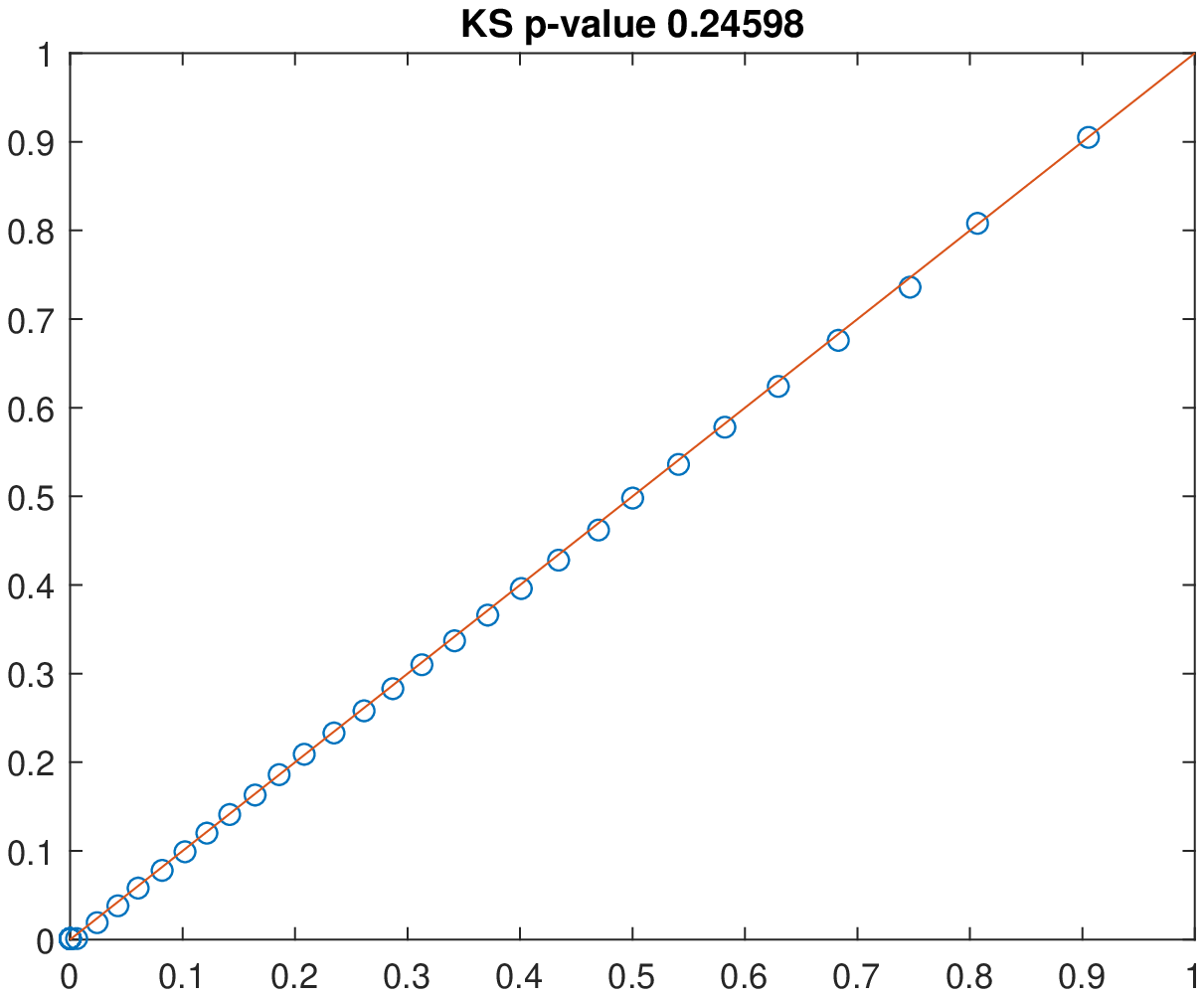}}}
	\caption{(Top row): Histograms for 10,000 samples generated from the law of a Wright--Fisher diffusion started at $x=0.25$ at time 0, sampled at times $t=0.05,0.25,0.5$ respectively, with the process allowed to be absorbed at the boundaries. Note that samples equal to 0 or 1 are not included in the above histograms, but their relative frequency can be found from the empirical probabilities found in Table \ref{AbsorptionProbabilitiesD1}. The truncated transition density is plotted in red. (Bottom row): QQ-plots for the corresponding samples with the $p$-value returned from the Kolmogorov--Smirnov test reported above the plot.}
\end{figure}

\begin{figure}[H]
	\centering
	\subfloat
	\centering{{\includegraphics[width=.28\linewidth]{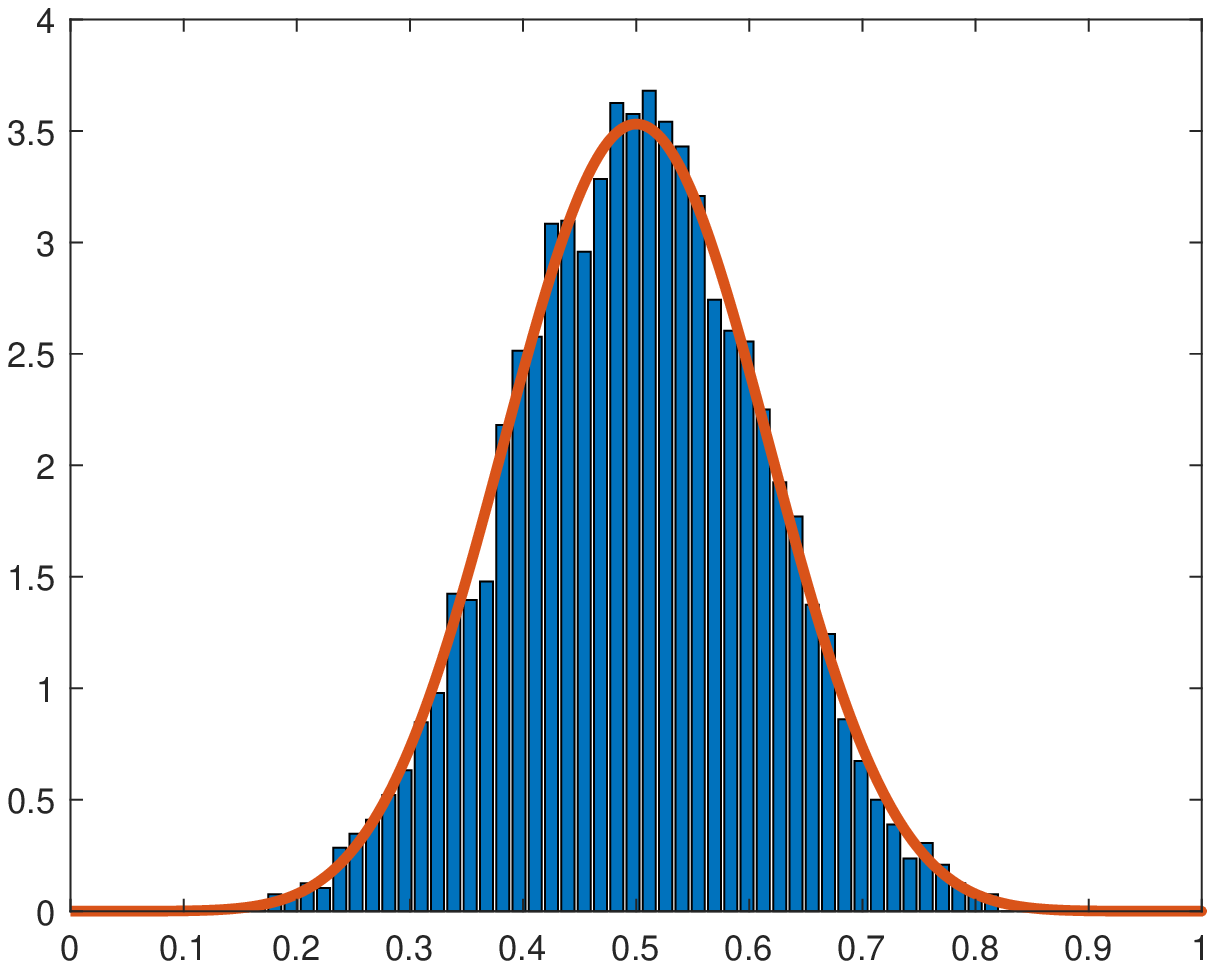} }} \hspace*{5mm}
	\subfloat
	\centering{{\includegraphics[width=.28\linewidth]{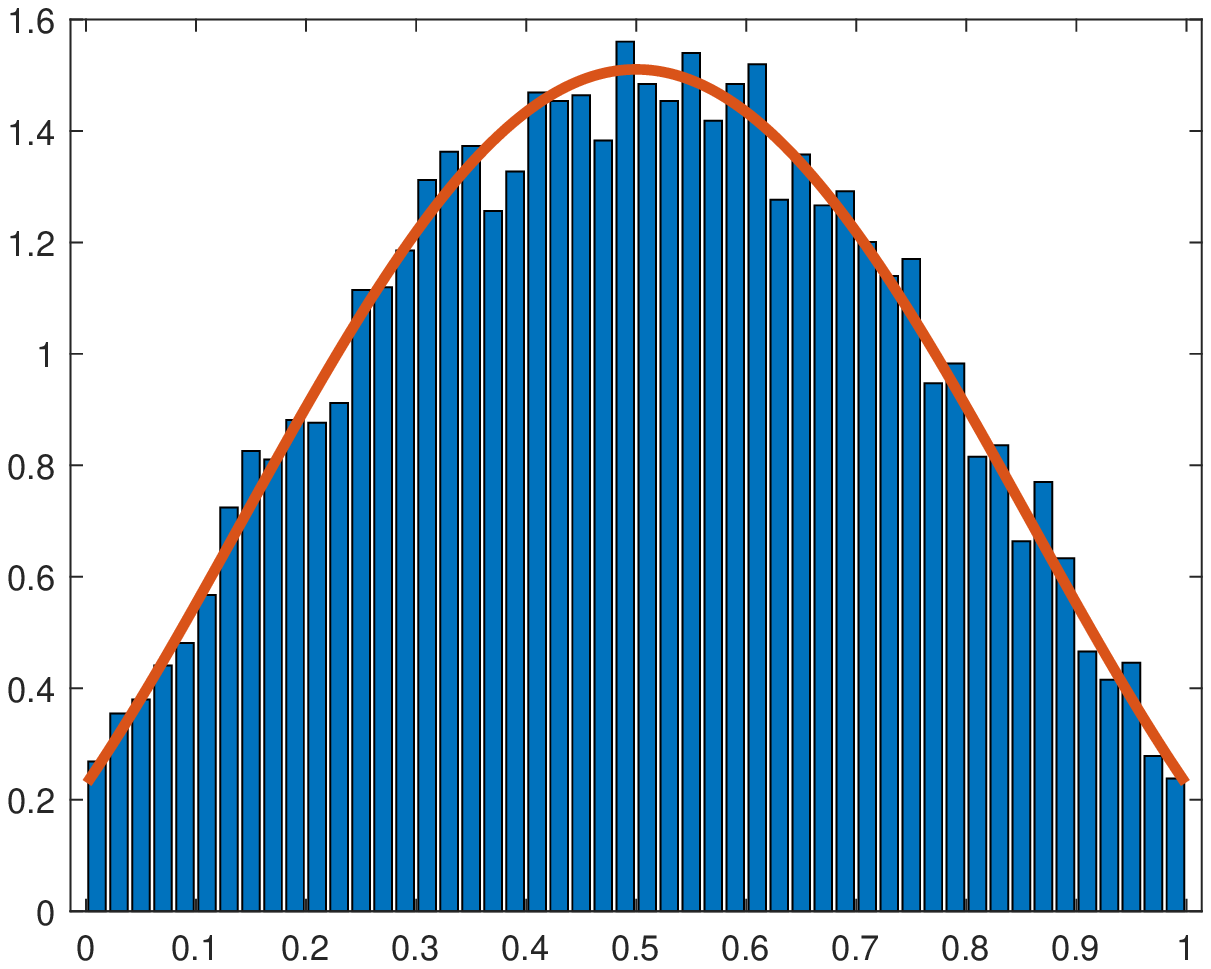} }} \hspace*{5mm}
	\subfloat
	\centering{{\includegraphics[width=.28\linewidth]{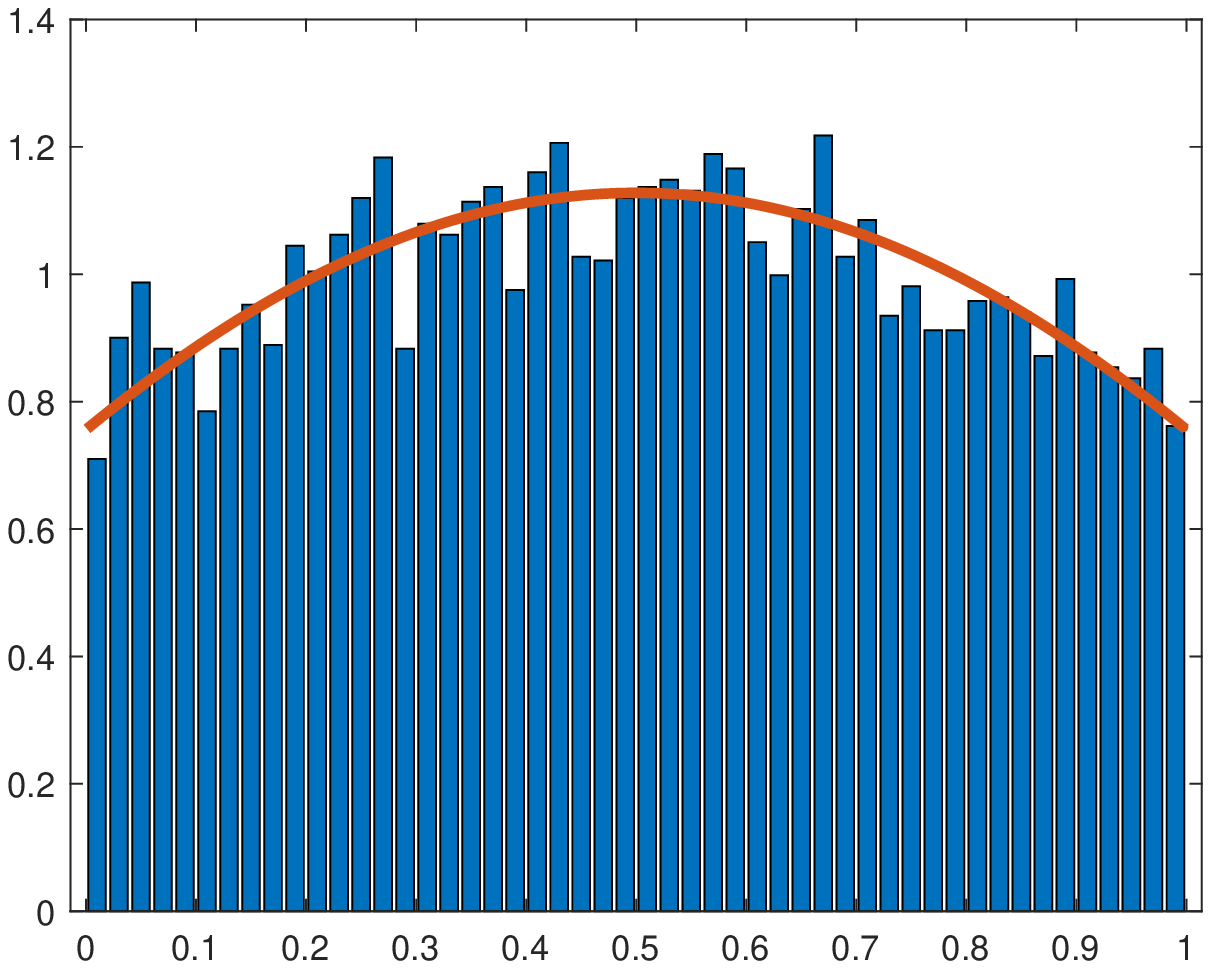}}}
	\\ \vspace*{5mm}
	\centering
	\subfloat
	\centering{{\includegraphics[width=.28\linewidth]{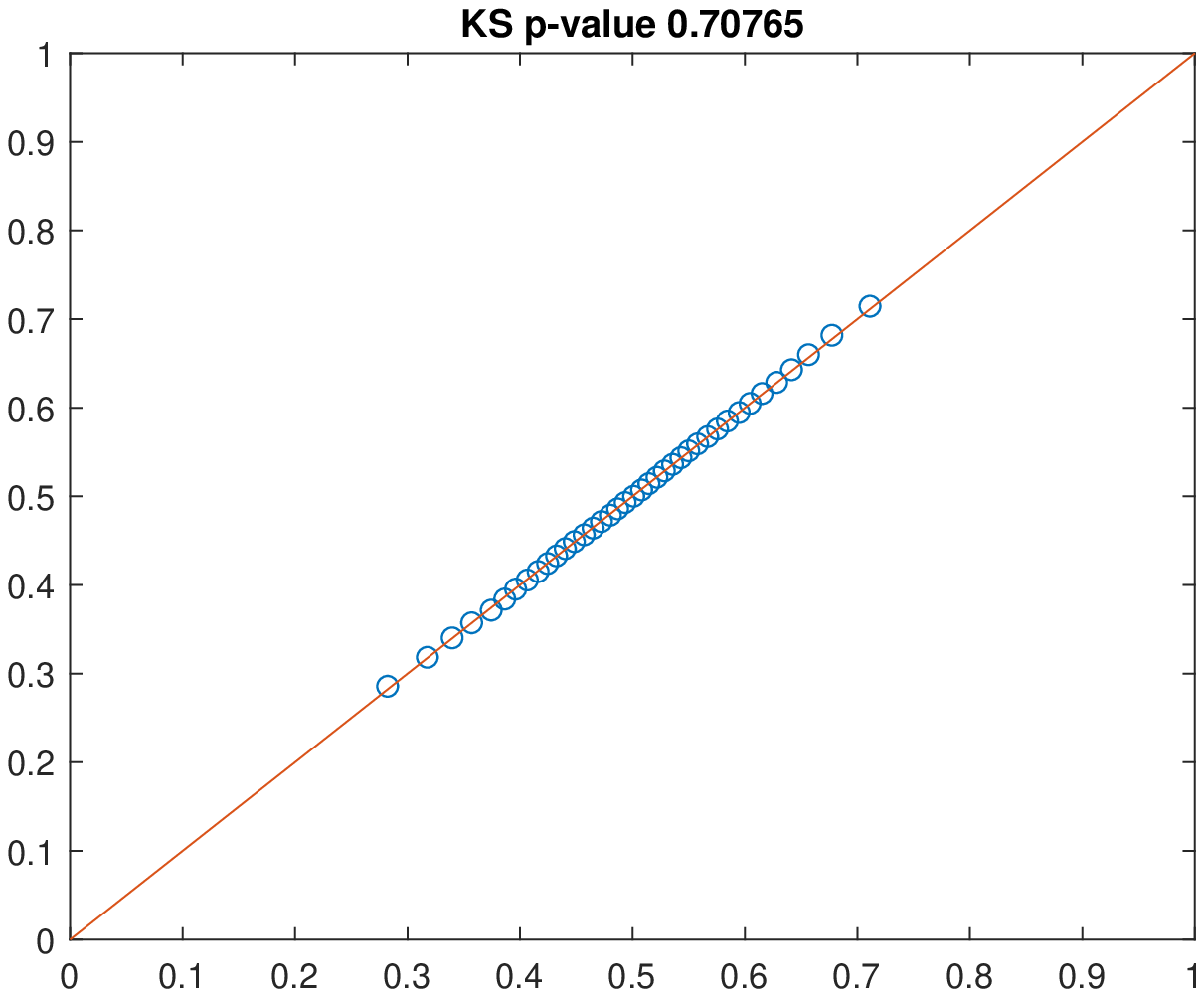} }} \hspace*{5mm}
	\subfloat
	\centering{{\includegraphics[width=.28\linewidth]{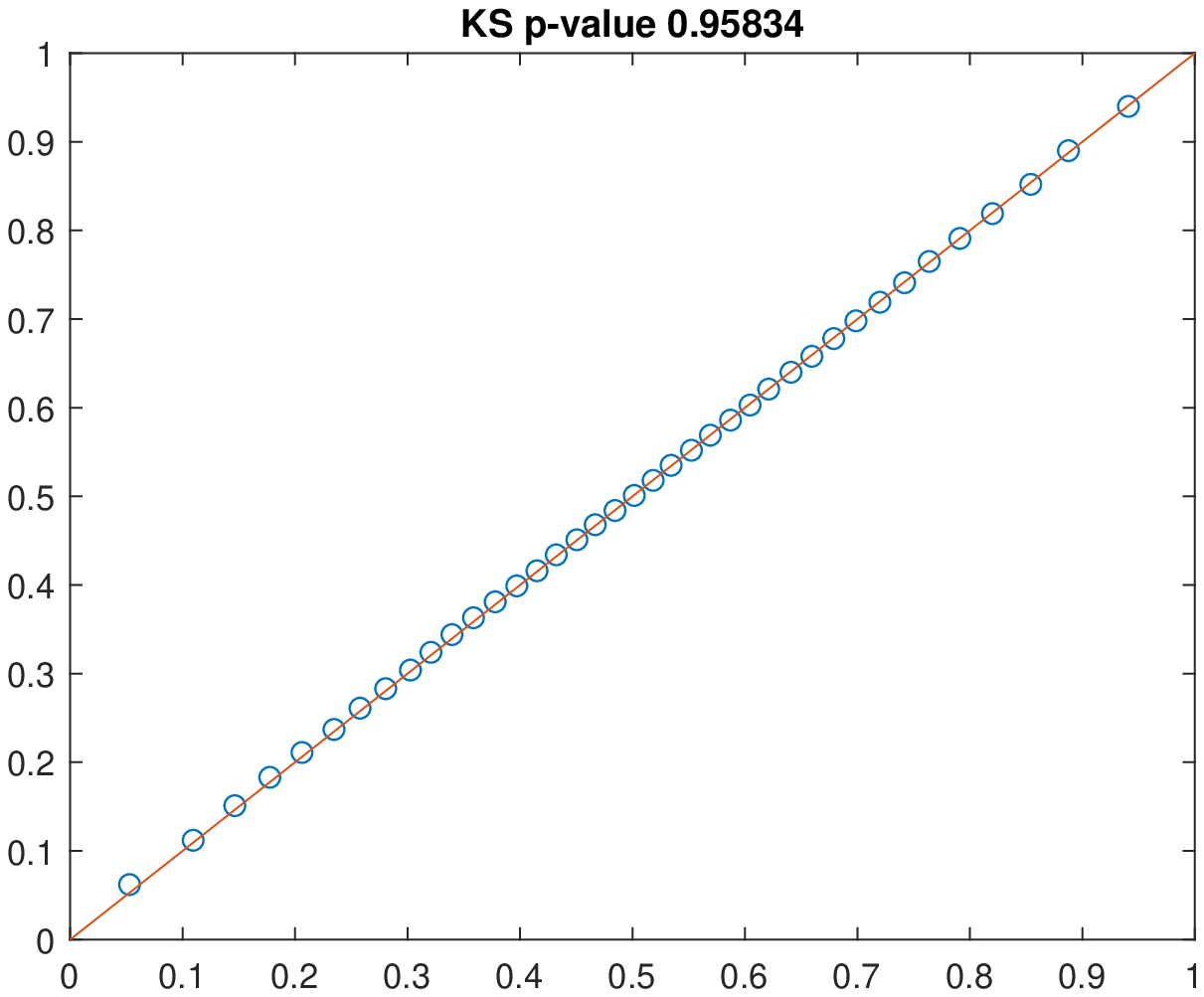} }} \hspace*{5mm}
	\subfloat
	\centering{{\includegraphics[width=.28\linewidth]{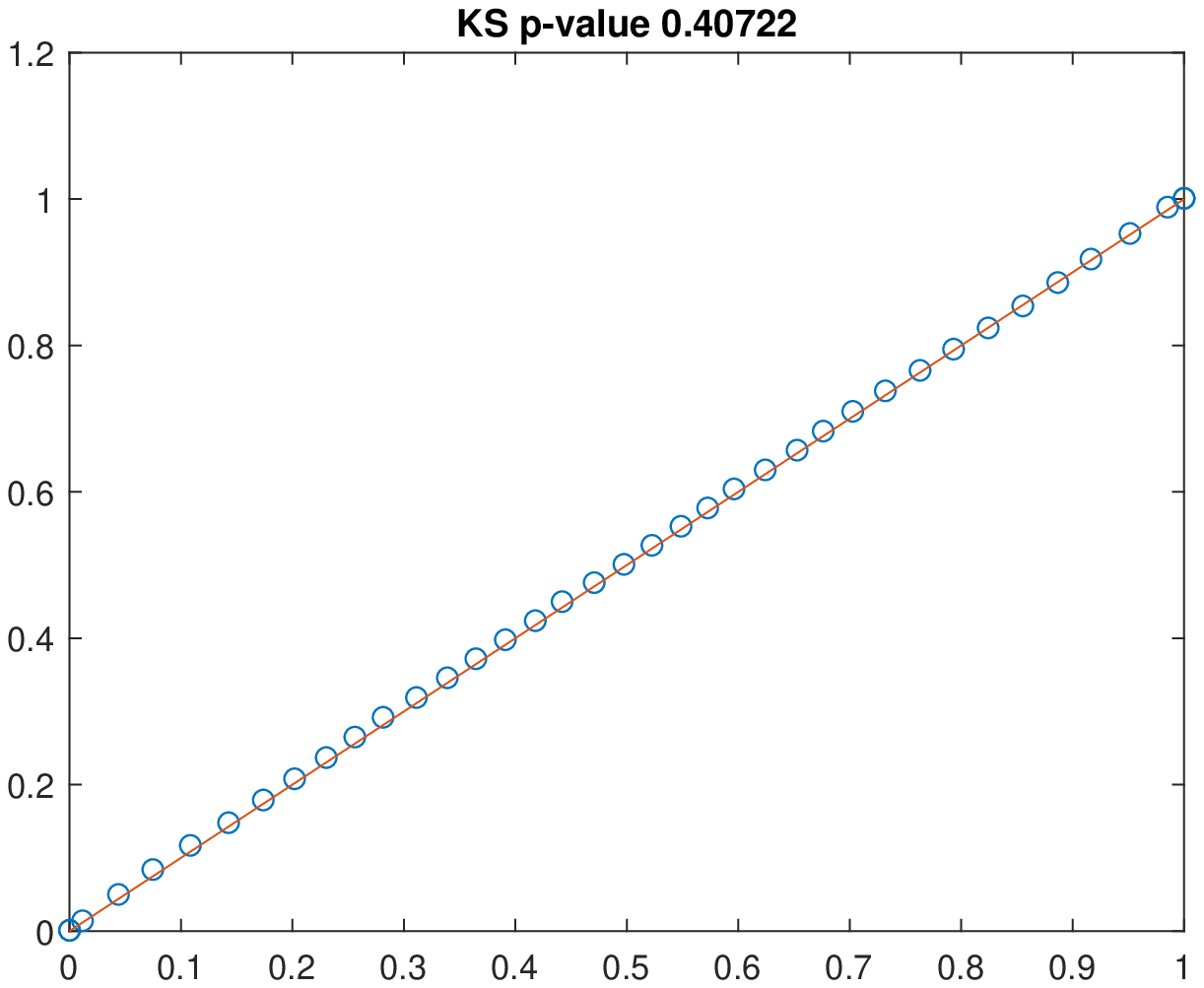}}}
	\caption{(Top row): Histograms for 10,000 samples generated from the law of a Wright--Fisher diffusion started at $x=0.5$ at time 0, sampled at times $t=0.05,0.25,0.5$ respectively, with the process allowed to be absorbed at the boundaries. Note that samples equal to 0 or 1 are not included in the above histograms, but their relative frequency can be found from the empirical probabilities found in Table \ref{AbsorptionProbabilitiesD2}. The truncated transition density is plotted in red. (Bottom row): QQ-plots for the corresponding samples with the $p$-value returned from the Kolmogorov--Smirnov test reported above the plot.}
\end{figure}

\begin{figure}[H]
	\centering
	\subfloat
	\centering{{\includegraphics[width=.28\linewidth]{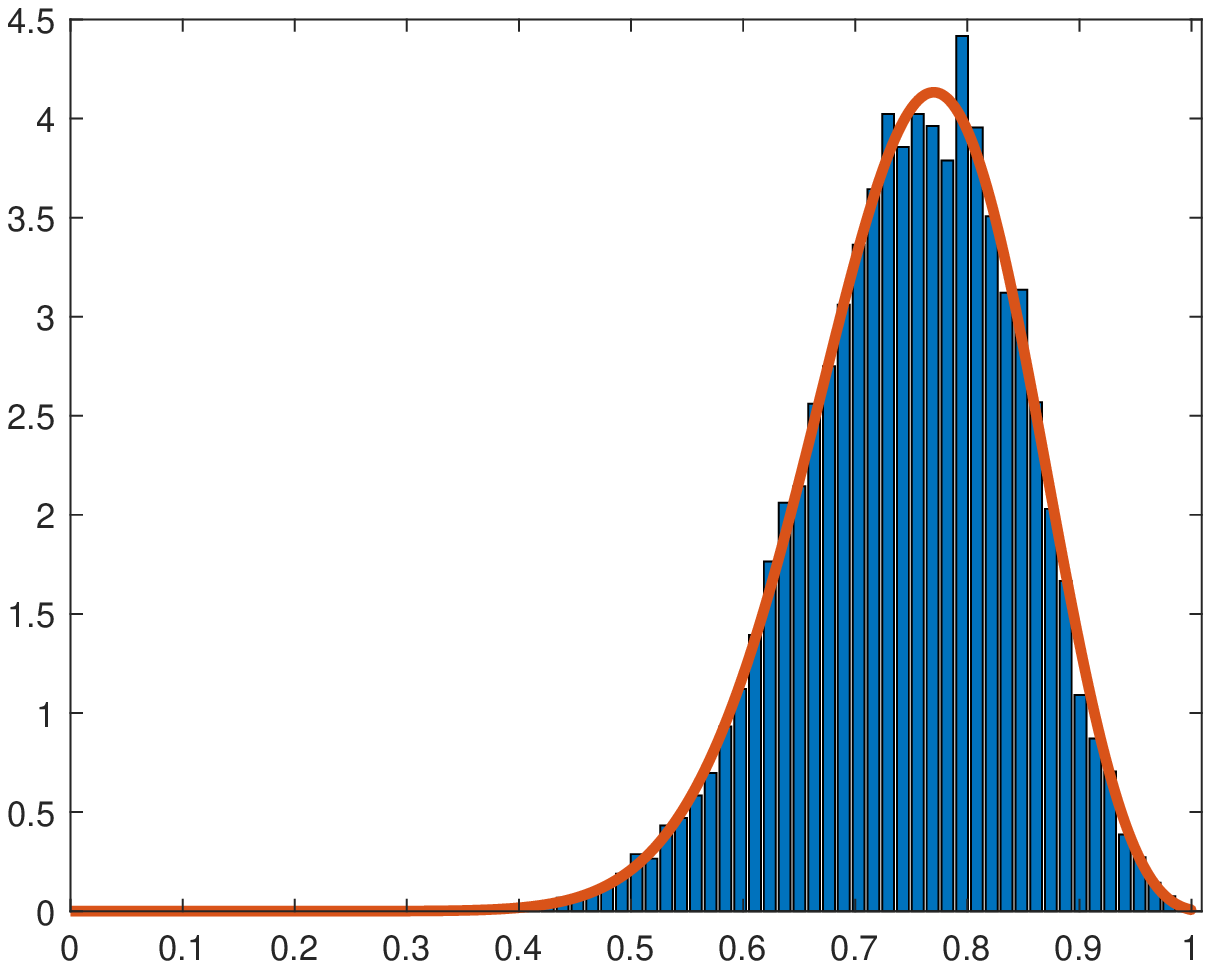} }} \hspace*{5mm}
	\subfloat
	\centering{{\includegraphics[width=.28\linewidth]{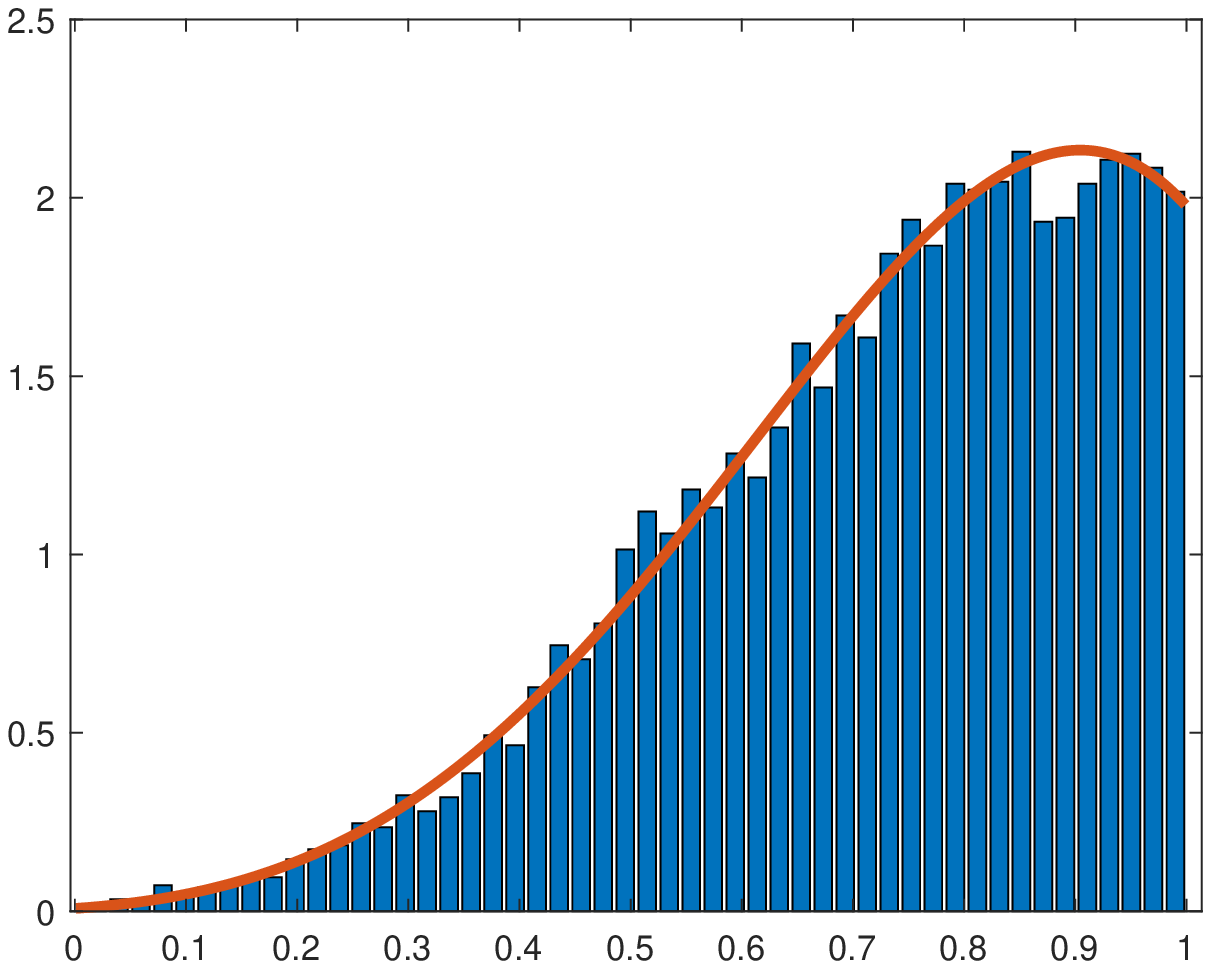} }} \hspace*{5mm}
	\subfloat
	\centering{{\includegraphics[width=.28\linewidth]{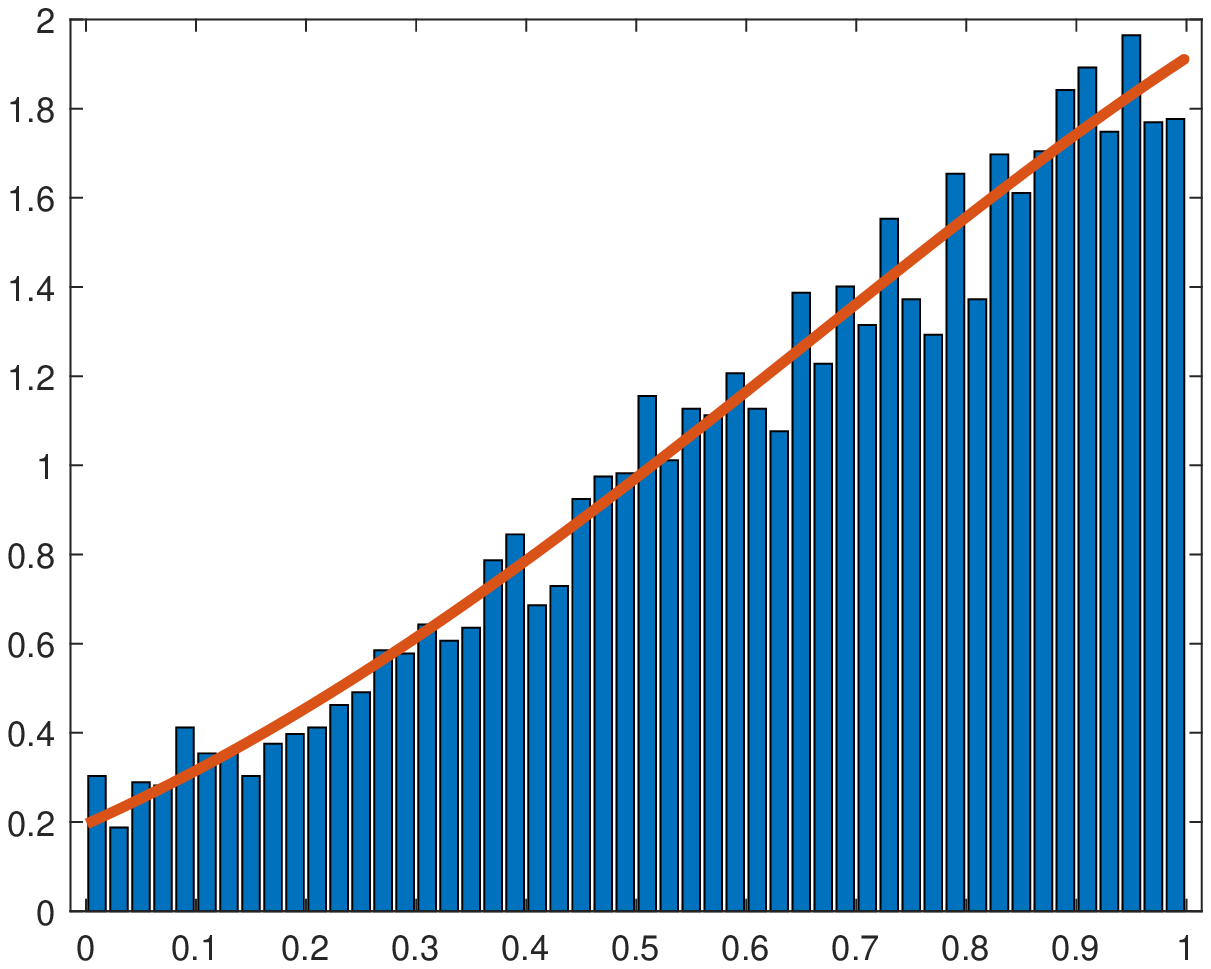}}}
	\\ \vspace*{5mm}
	\centering
	\subfloat
	\centering{{\includegraphics[width=.28\linewidth]{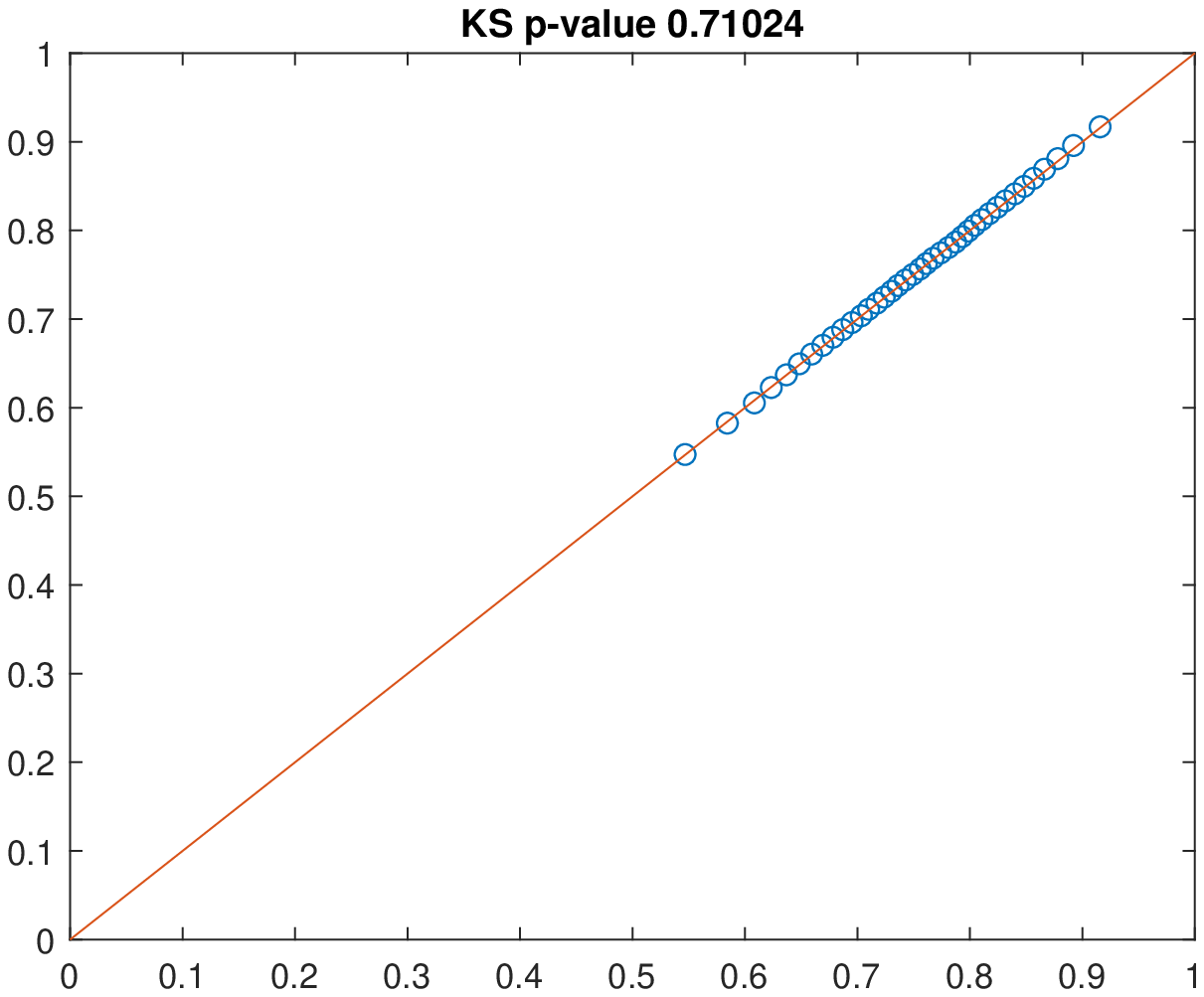} }} \hspace*{5mm}
	\subfloat
	\centering{{\includegraphics[width=.28\linewidth]{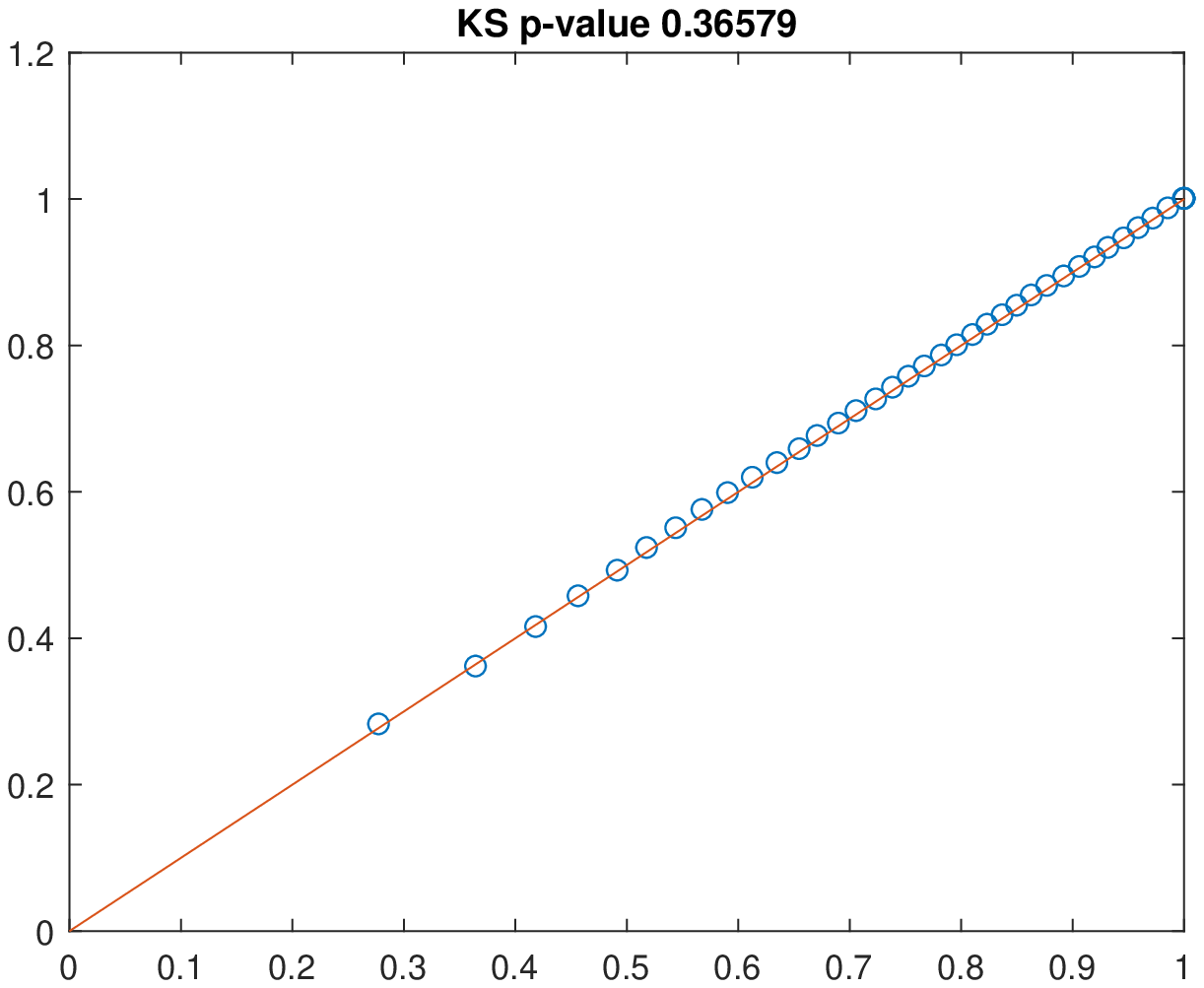} }} \hspace*{5mm}
	\subfloat
	\centering{{\includegraphics[width=.28\linewidth]{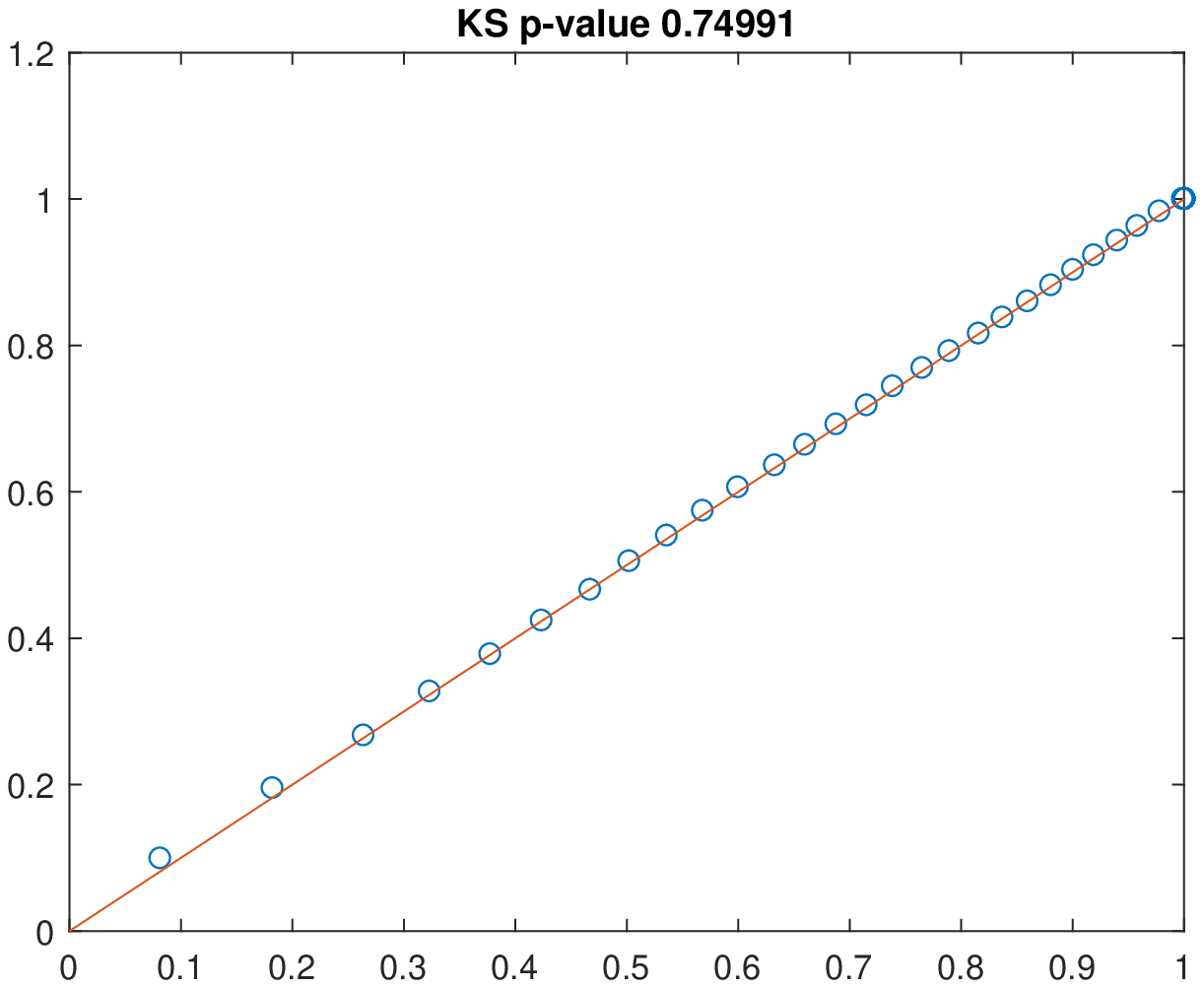}}}
	\caption{(Top row): Histograms for 10,000 samples generated from the law of a Wright--Fisher diffusion started at $x=0.75$ at time 0, sampled at times $t=0.05,0.25,0.5$ respectively, with the process allowed to be absorbed at the boundaries. Note that samples equal to 0 or 1 are not included in the above histograms, but their relative frequency can be found from the empirical probabilities found in Table \ref{AbsorptionProbabilitiesD3}. The truncated transition density is plotted in red. (Bottom row): QQ-plots for the corresponding samples with the $p$-value returned from the Kolmogorov--Smirnov test reported above the plot.}
\end{figure}

\subsection{Diffusion bridges conditioned on non-absorption}

\noindent To validate the diffusion bridge simulation, we chose to simulate draws from the following three diffusion bridges:

\begin{table}[H]
	\begin{center}
		\begin{tabular}{c|c c c c}
			& $(t_0,x_0)$ & $(t_1,x_1)$ & $(t_2,x_2)$ & $(t_3,x_3)$ \\ \hline
			Bridge 1 & (0,0) & (0.05,0.1) & (0.1,0.25) & \\ \hline 
			Bridge 2 & (0.2,0.1) & (0.3,0.3) & (0.4,0.4) & (0.5,0.5) \\ \hline
			Bridge 3 & (0,1) & (0.5,0.95) & &  
		\end{tabular}
		\caption{The left and right endpoints for the three different bridges simulated, where $(t_0,x_0)$ denotes the bridge's start time $t_0$ and start point $x_0$, $(t_1,x_1)$ denotes the second observation time and point for the diffusion bridge and so on.}
		\label{ConditionedBridges}
	\end{center}
\end{table}

\noindent We further considered the following sampling times for each bridge:

\begin{table}[H]
	\begin{center}
		\begin{tabular}{c|c c c}
			& $s_1$ & $s_2$ & $s_3$ \\ \hline
			Bridge 1 & 0.025 & 0.065 & 0.085 \\ \hline
			Bridge 2 & 0.25 & 0.35 & 0.45 \\ \hline
			Bridge 3 & 0.1 & 0.25 & 0.3 
		\end{tabular}
	\end{center}
	\caption{Sampling times for the three different diffusion bridges considered.}
	\label{ConditionedBridgesSamplingTimes}
\end{table}

\noindent The output generated is plotted below, starting with bridge 1, and the sampling times $s_i$ increasing from left to right. Again all the output strongly indicates that the method is returning draws from the desired target distribution.

\begin{figure}[H]
	\centering
	\subfloat
	\centering{{\includegraphics[width=.28\linewidth]{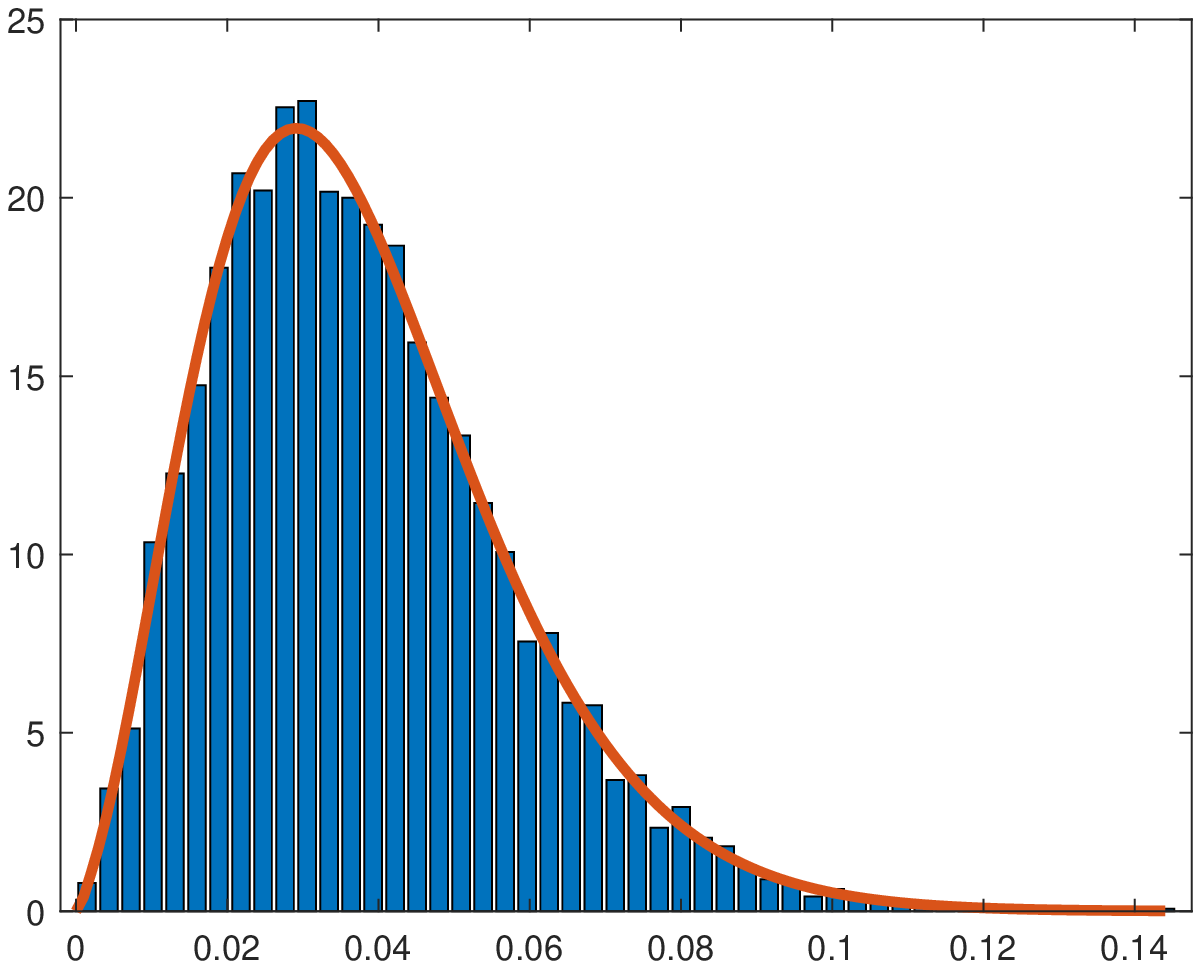} }} \hspace*{5mm}
	\subfloat
	\centering{{\includegraphics[width=.28\linewidth]{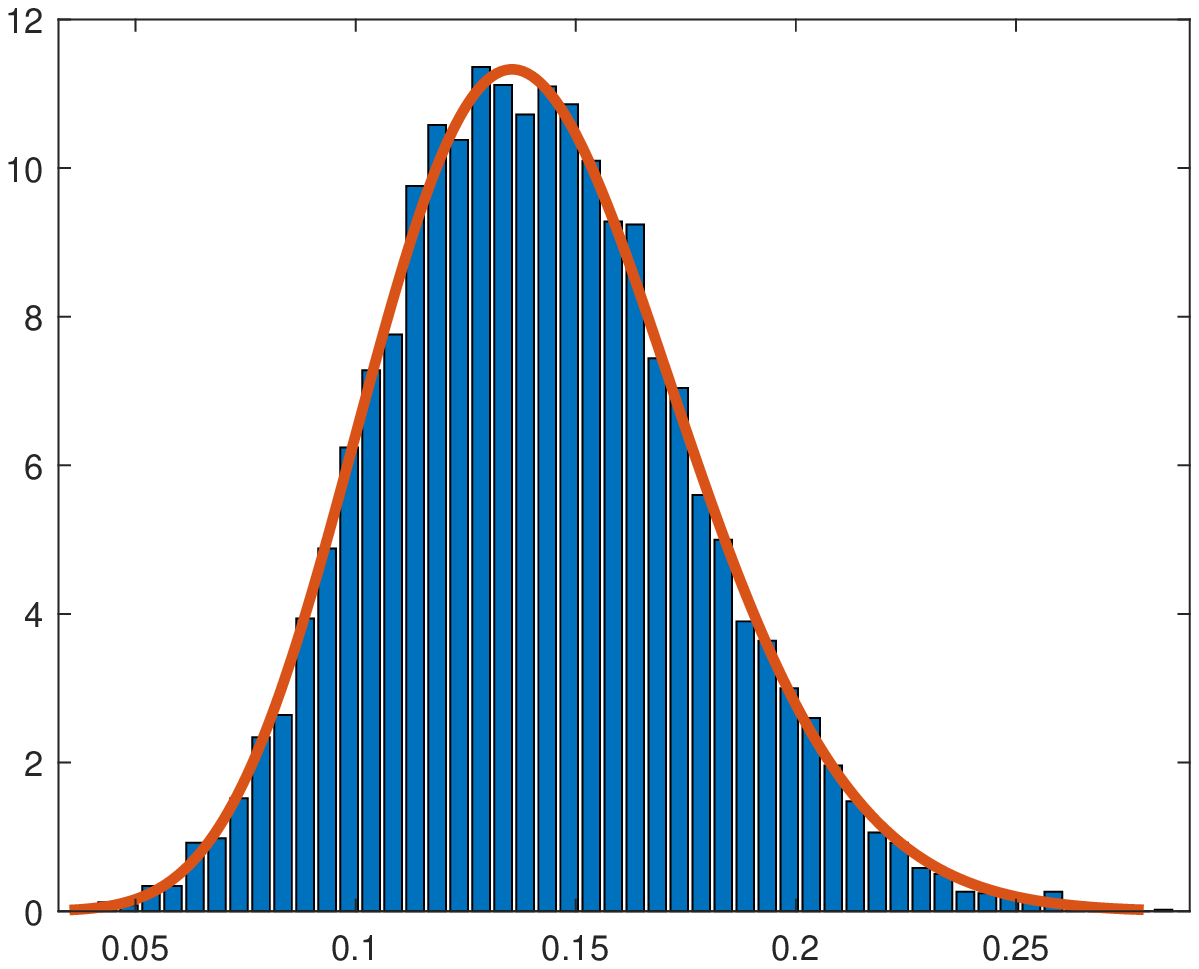} }} \hspace*{5mm}
	\subfloat
	\centering{{\includegraphics[width=.28\linewidth]{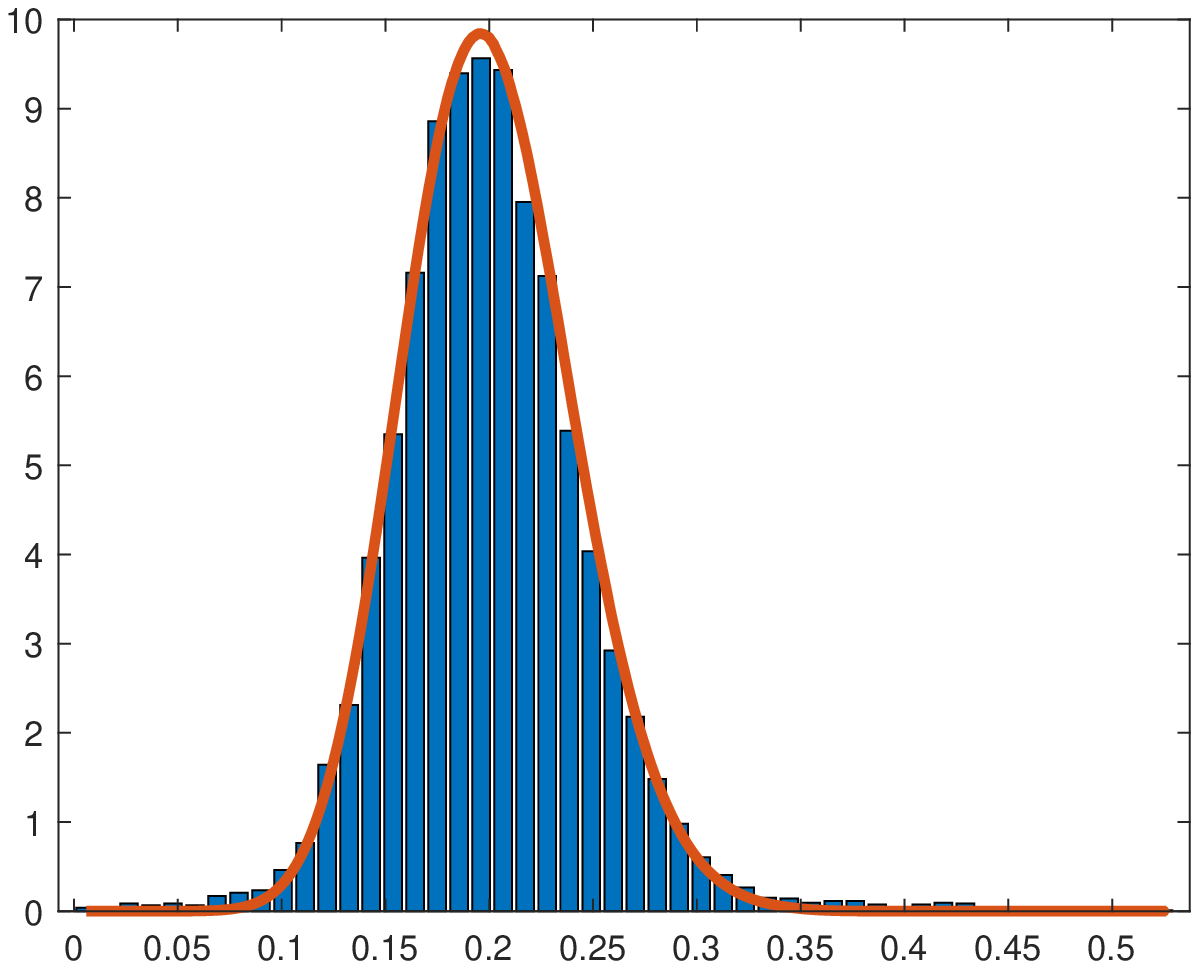}}}
	\\ \vspace*{5mm}
	\centering
	\subfloat
	\centering{{\includegraphics[width=.28\linewidth]{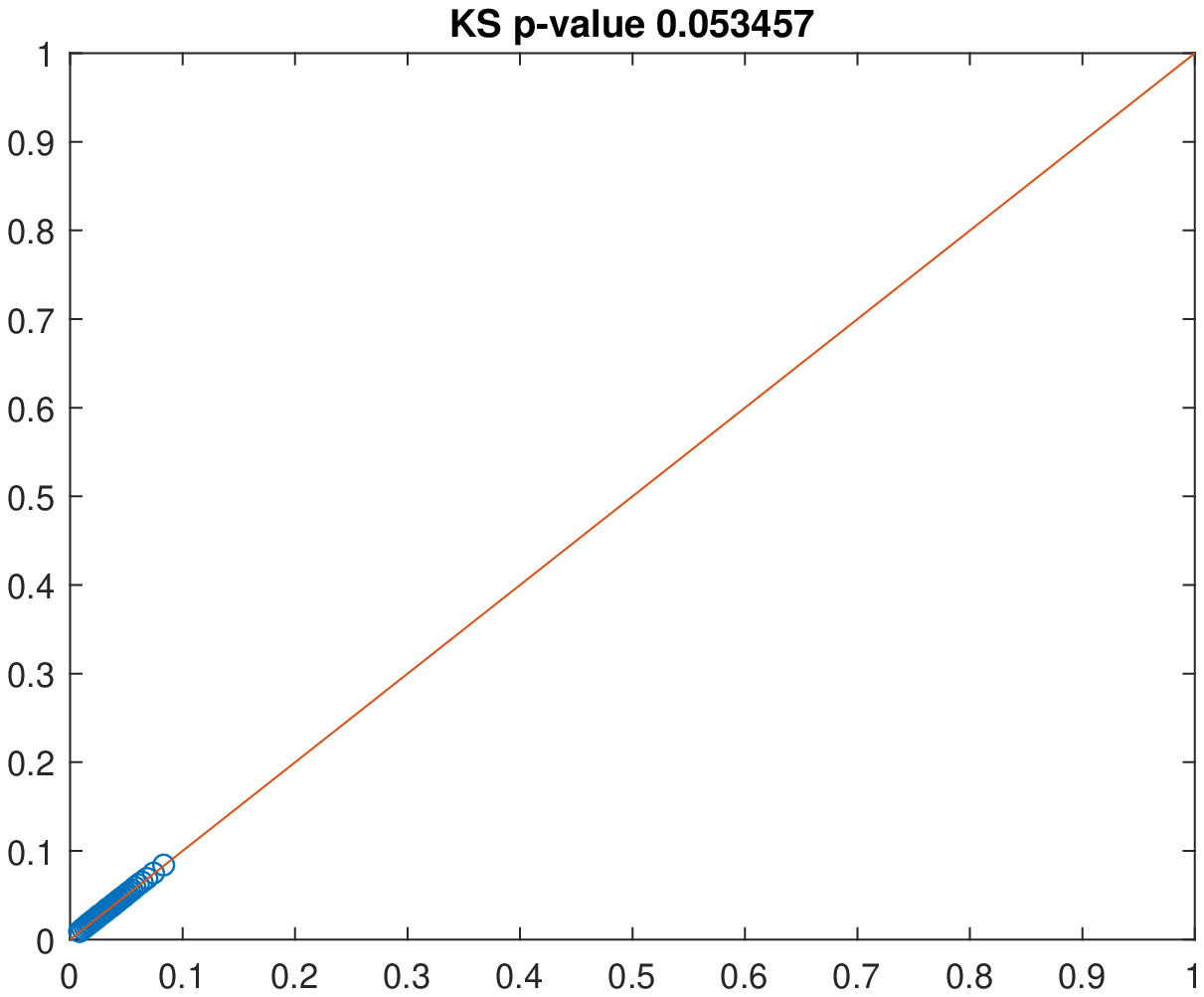} }} \hspace*{5mm}
	\subfloat
	\centering{{\includegraphics[width=.28\linewidth]{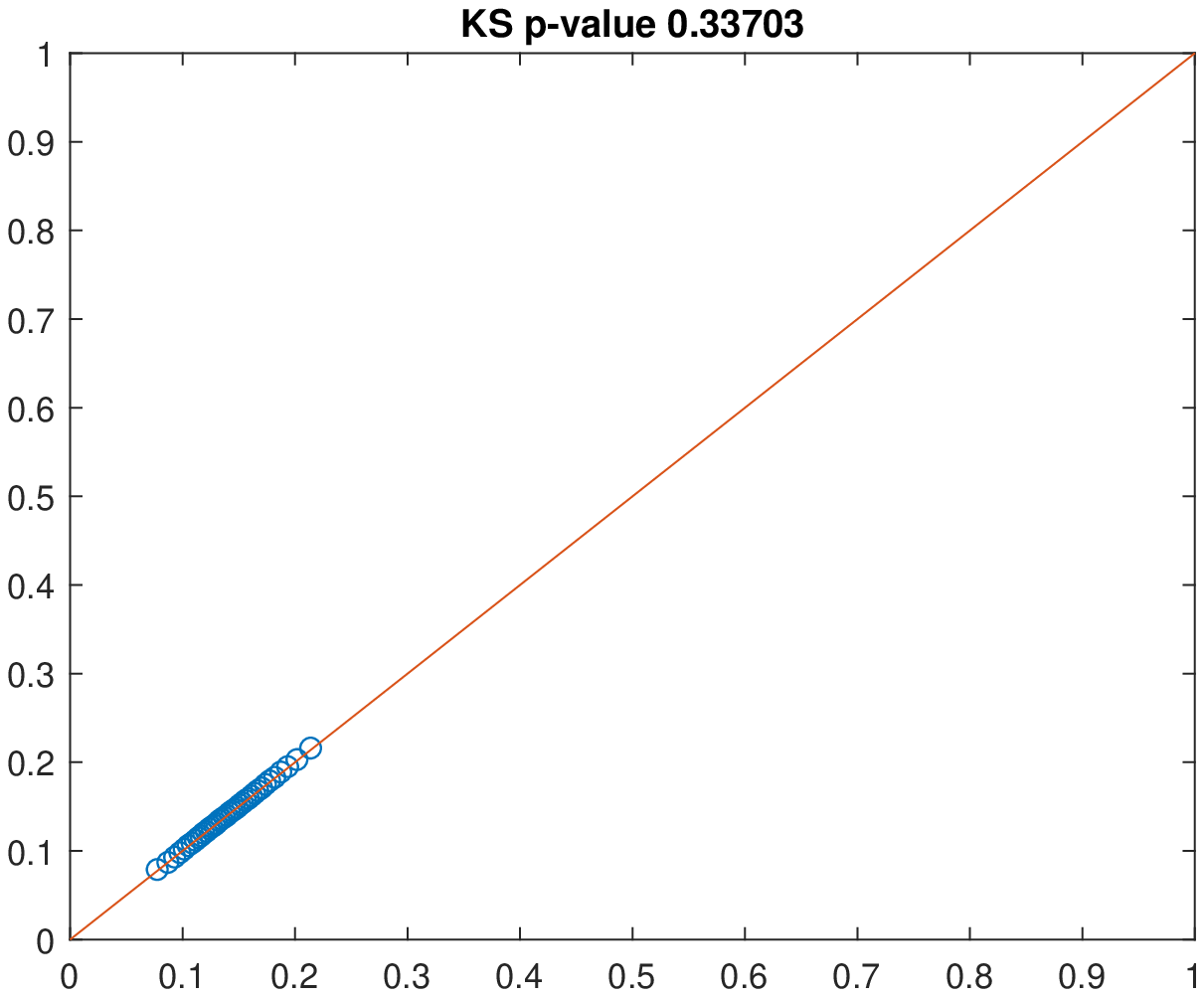} }} \hspace*{5mm}
	\subfloat
	\centering{{\includegraphics[width=.28\linewidth]{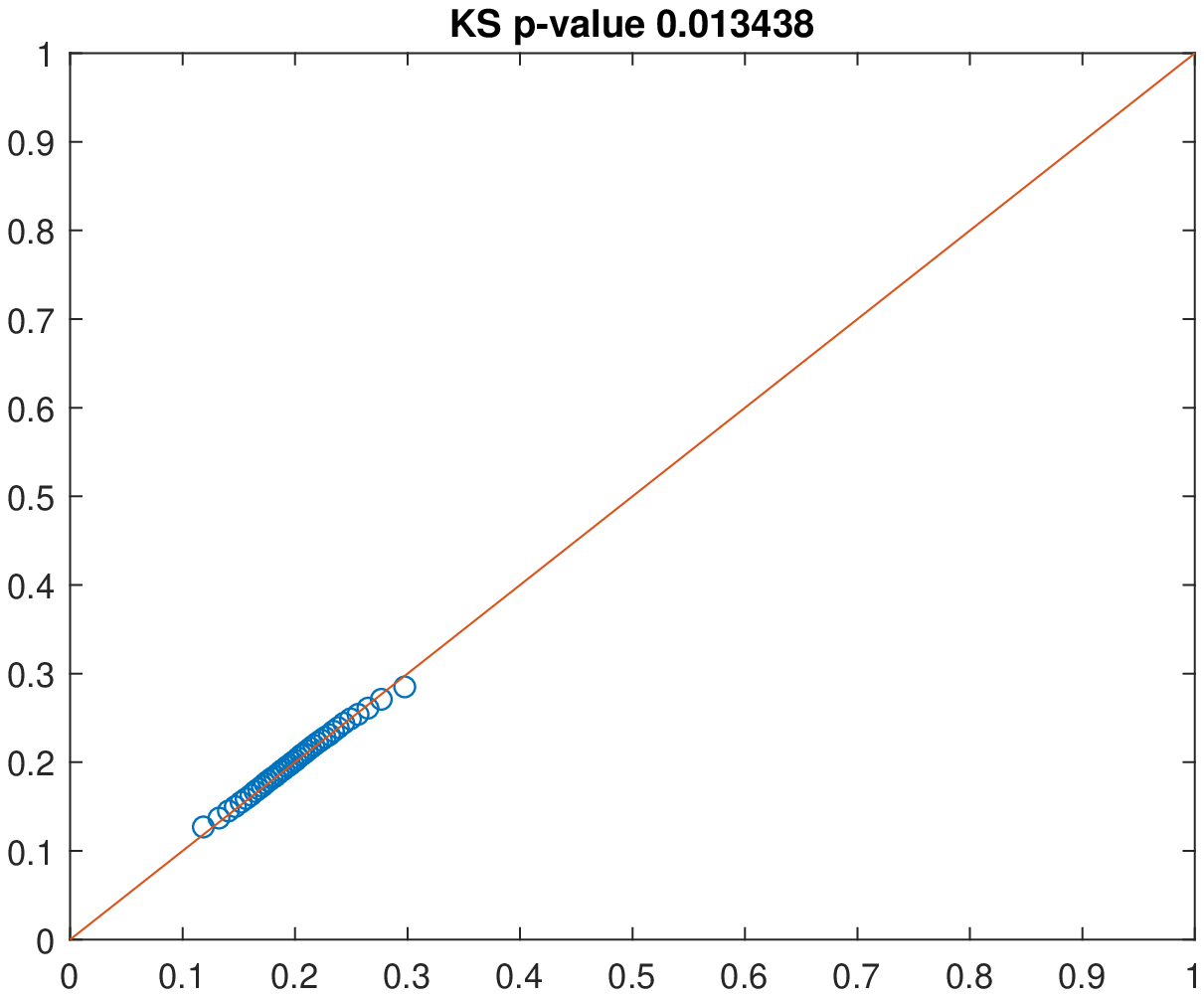}}}
	\caption{(Top row): Histograms for 10,000 samples generated from the law of the Wright--Fisher diffusion bridge `Bridge 1' in Table \ref{ConditionedBridges} above, sampled at the times given by the corresponding row in Table \ref{ConditionedBridgesSamplingTimes}. The truncated transition density is plotted in red. (Bottom row): QQ-plots for the corresponding samples with the $p$-value returned from the Kolmogorov--Smirnov test reported above the plot.}
\end{figure}

\begin{figure}[H]
	\centering
	\subfloat
	\centering{{\includegraphics[width=.28\linewidth]{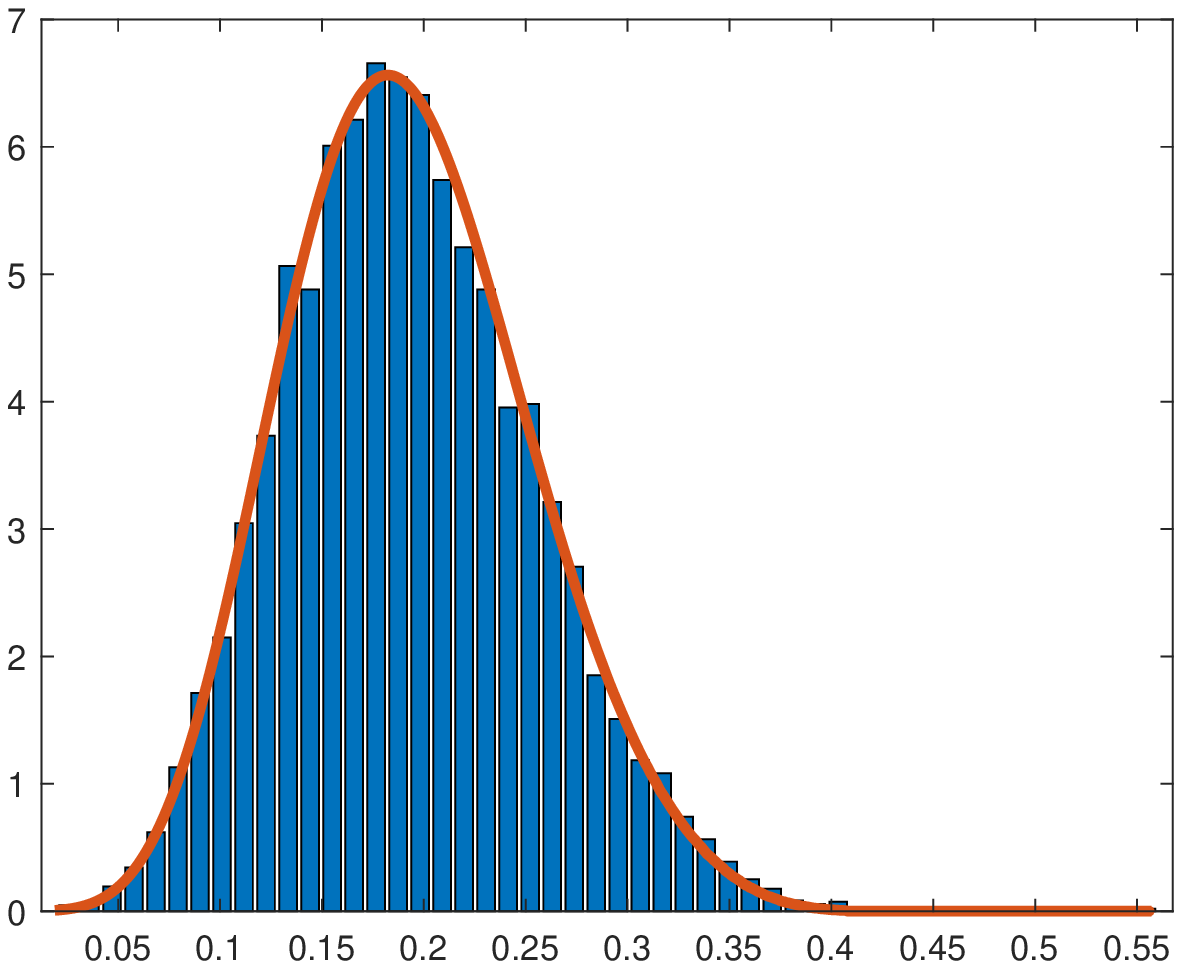} }} \hspace*{5mm}
	\subfloat
	\centering{{\includegraphics[width=.28\linewidth]{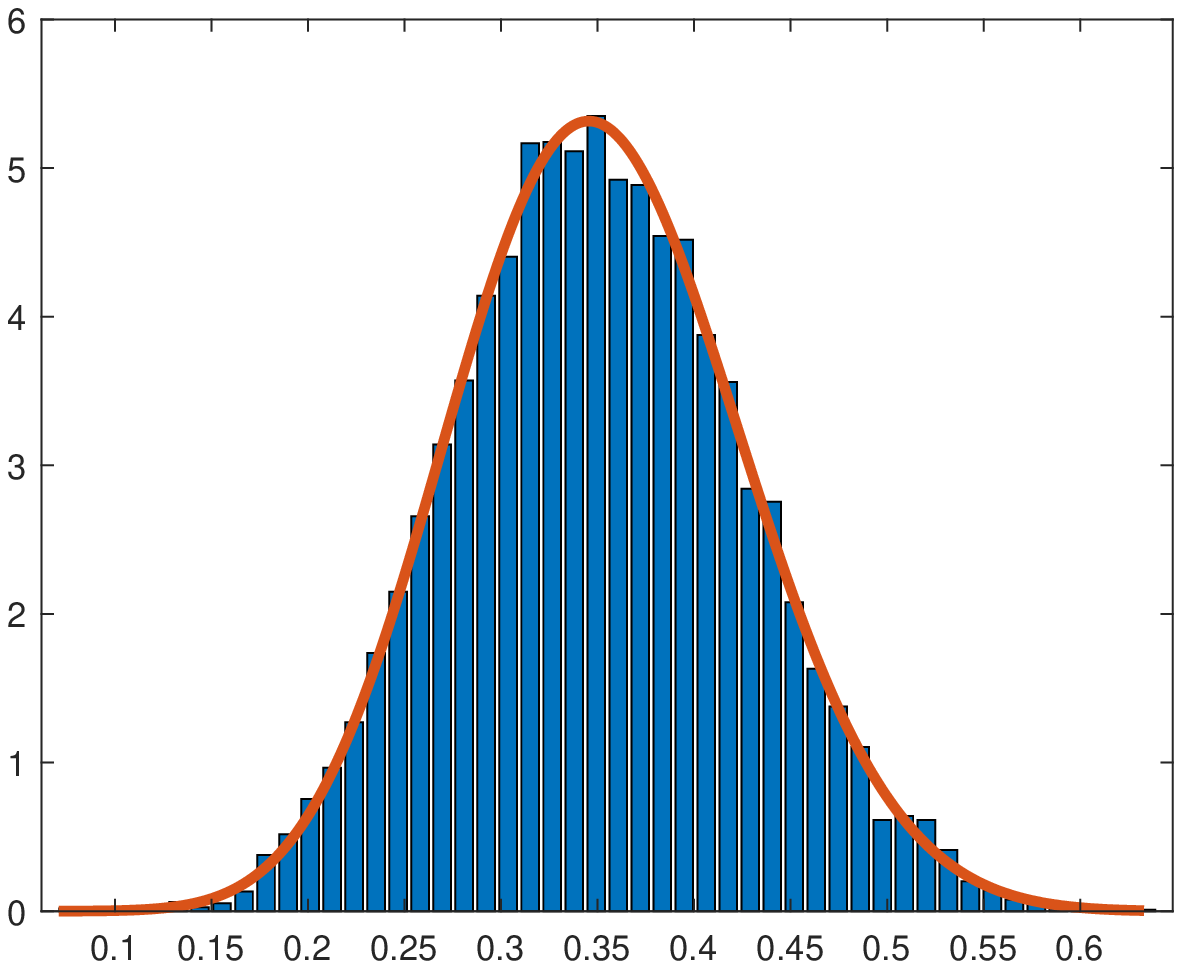} }} \hspace*{5mm}
	\subfloat
	\centering{{\includegraphics[width=.28\linewidth]{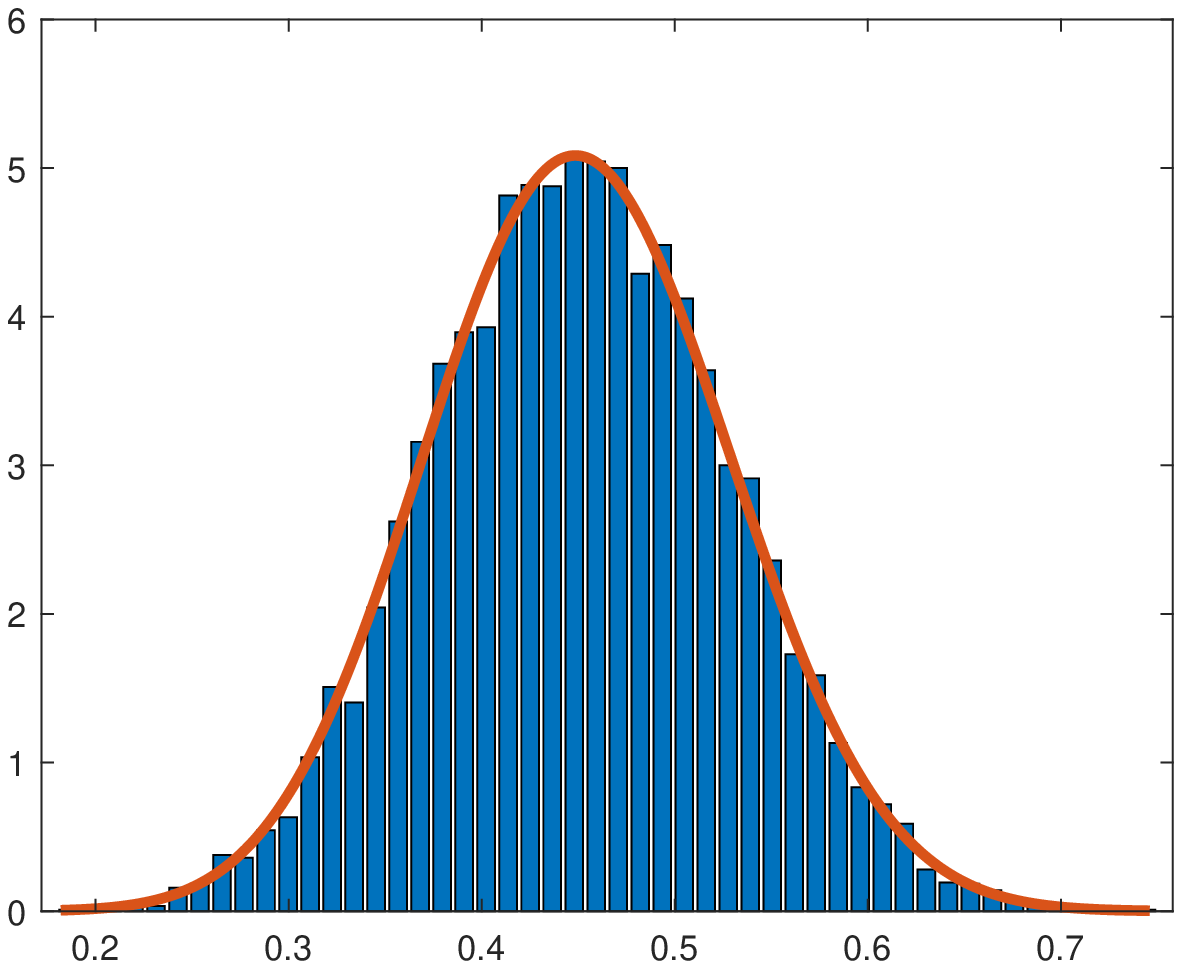}}}
	\\ \vspace*{5mm}
	\centering
	\subfloat
	\centering{{\includegraphics[width=.28\linewidth]{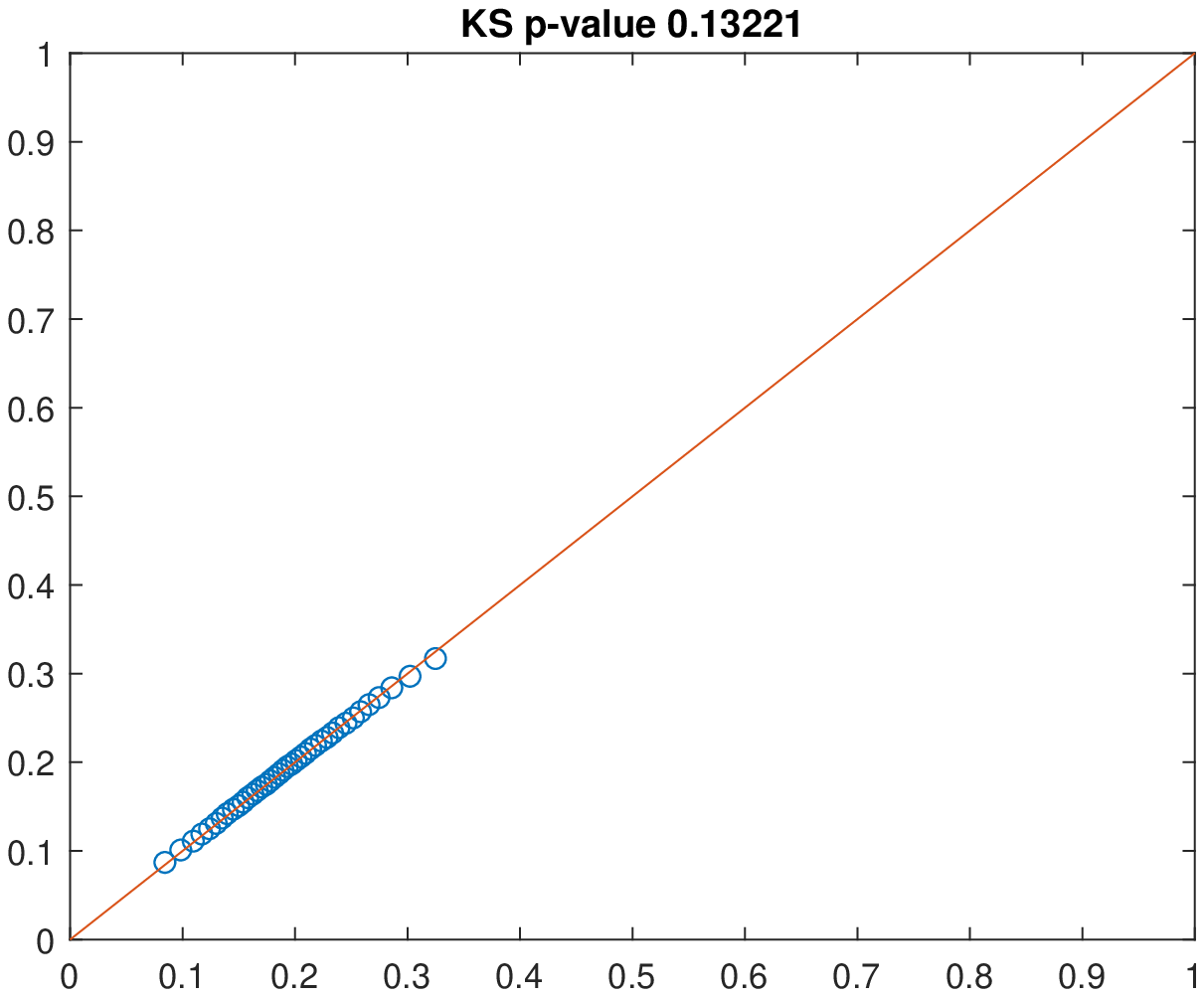} }} \hspace*{5mm}
	\subfloat
	\centering{{\includegraphics[width=.28\linewidth]{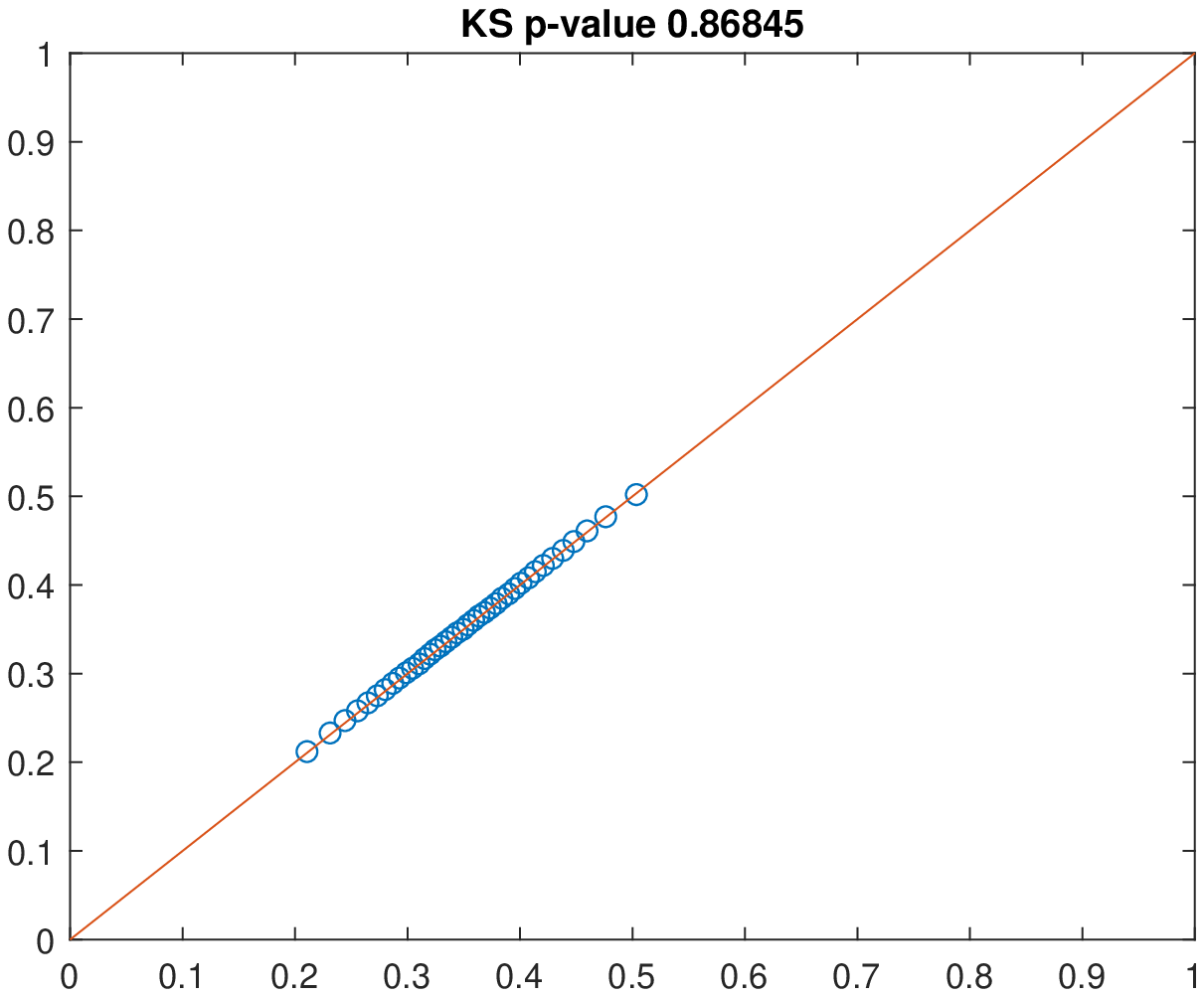} }} \hspace*{5mm}
	\subfloat
	\centering{{\includegraphics[width=.28\linewidth]{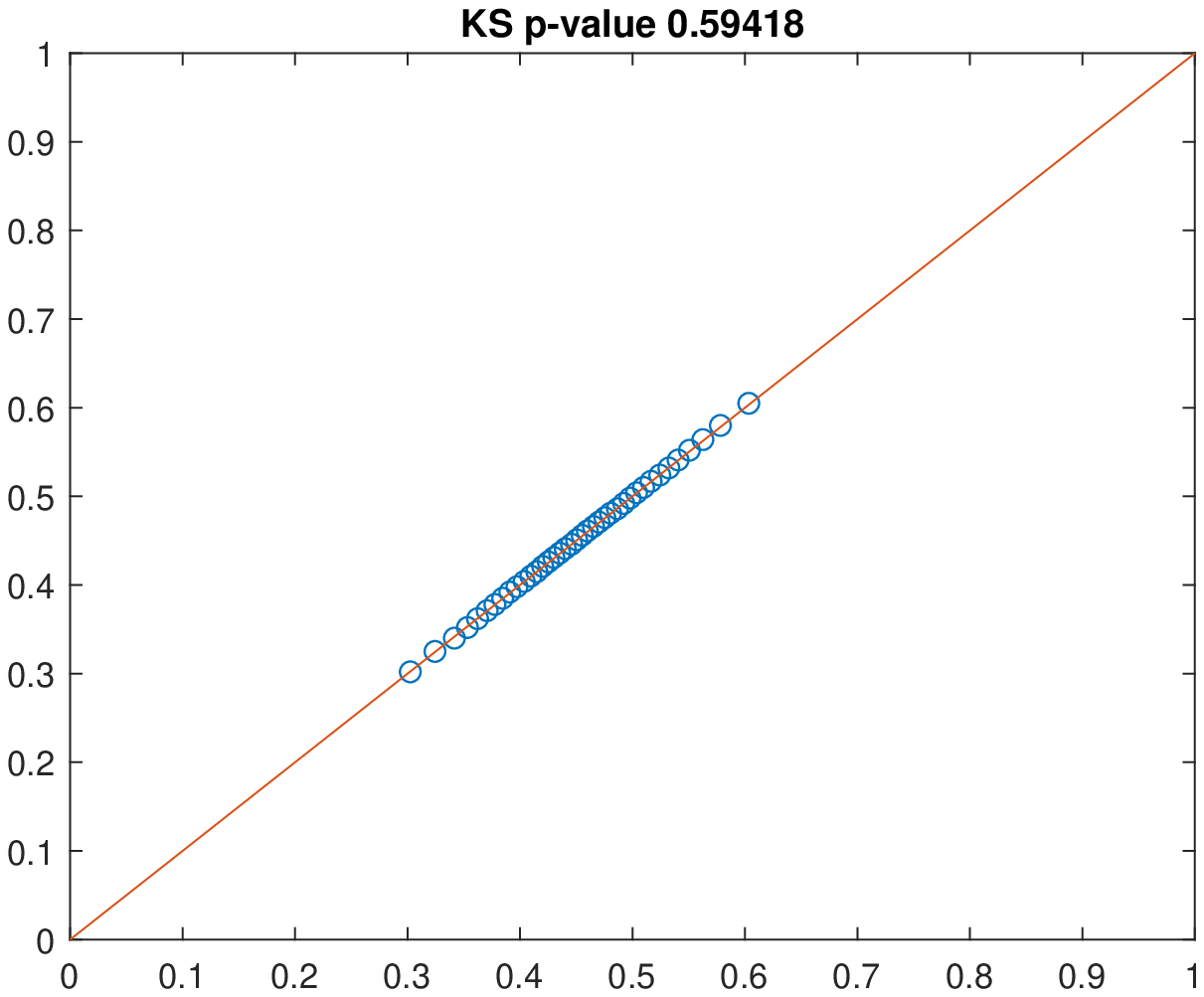}}}
	\caption{(Top row): Histograms for 10,000 samples generated from the law of the Wright--Fisher diffusion bridge `Bridge 2' as given in Table \ref{ConditionedBridges} above, sampled at the times given by the corresponding row in Table \ref{ConditionedBridgesSamplingTimes}. The truncated transition density is plotted in red. (Bottom row): QQ-plots for the corresponding samples with the $p$-value returned from the Kolmogorov--Smirnov test reported above the plot.}
\end{figure}

\begin{figure}[H]
	\centering
	\subfloat
	\centering{{\includegraphics[width=.28\linewidth]{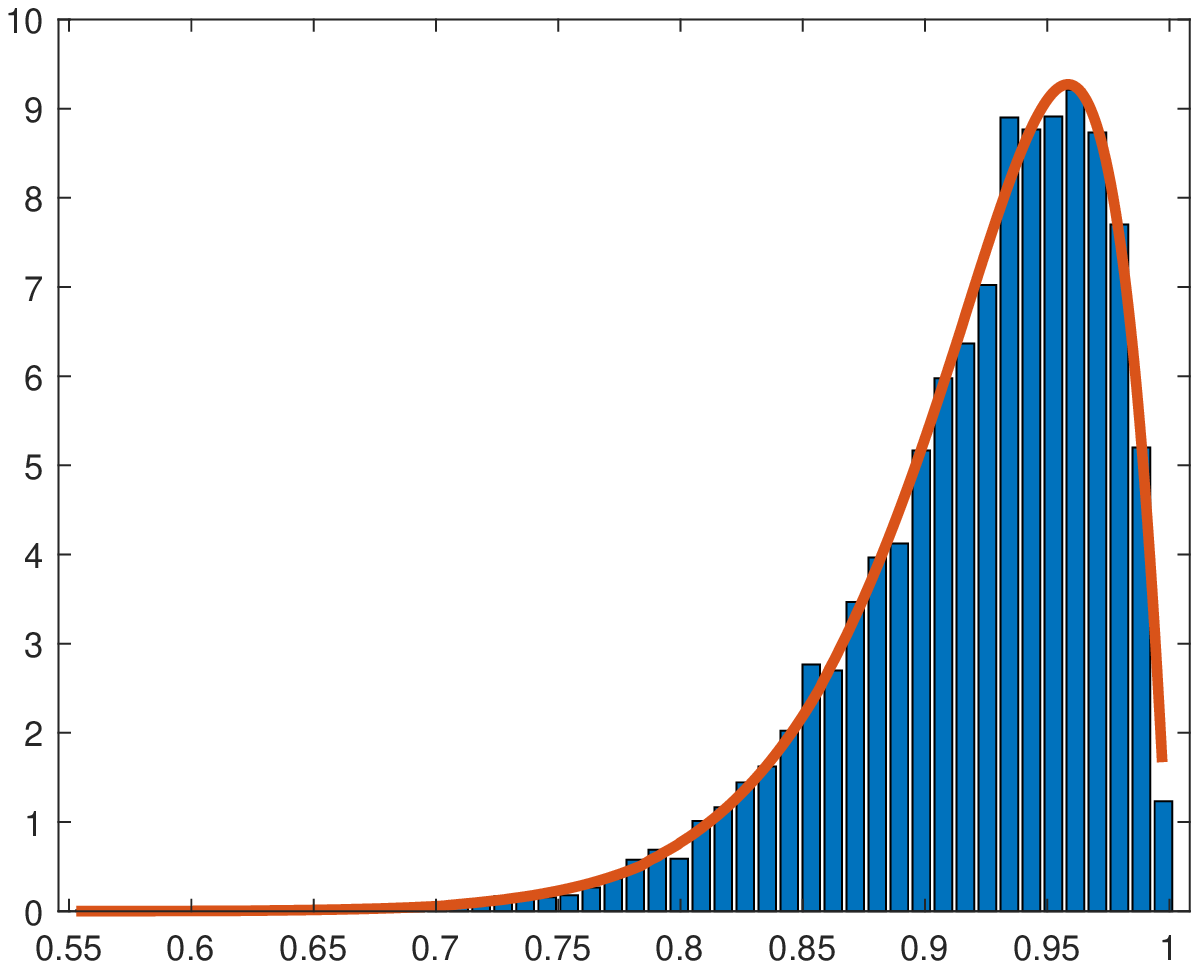} }} \hspace*{5mm}
	\subfloat
	\centering{{\includegraphics[width=.28\linewidth]{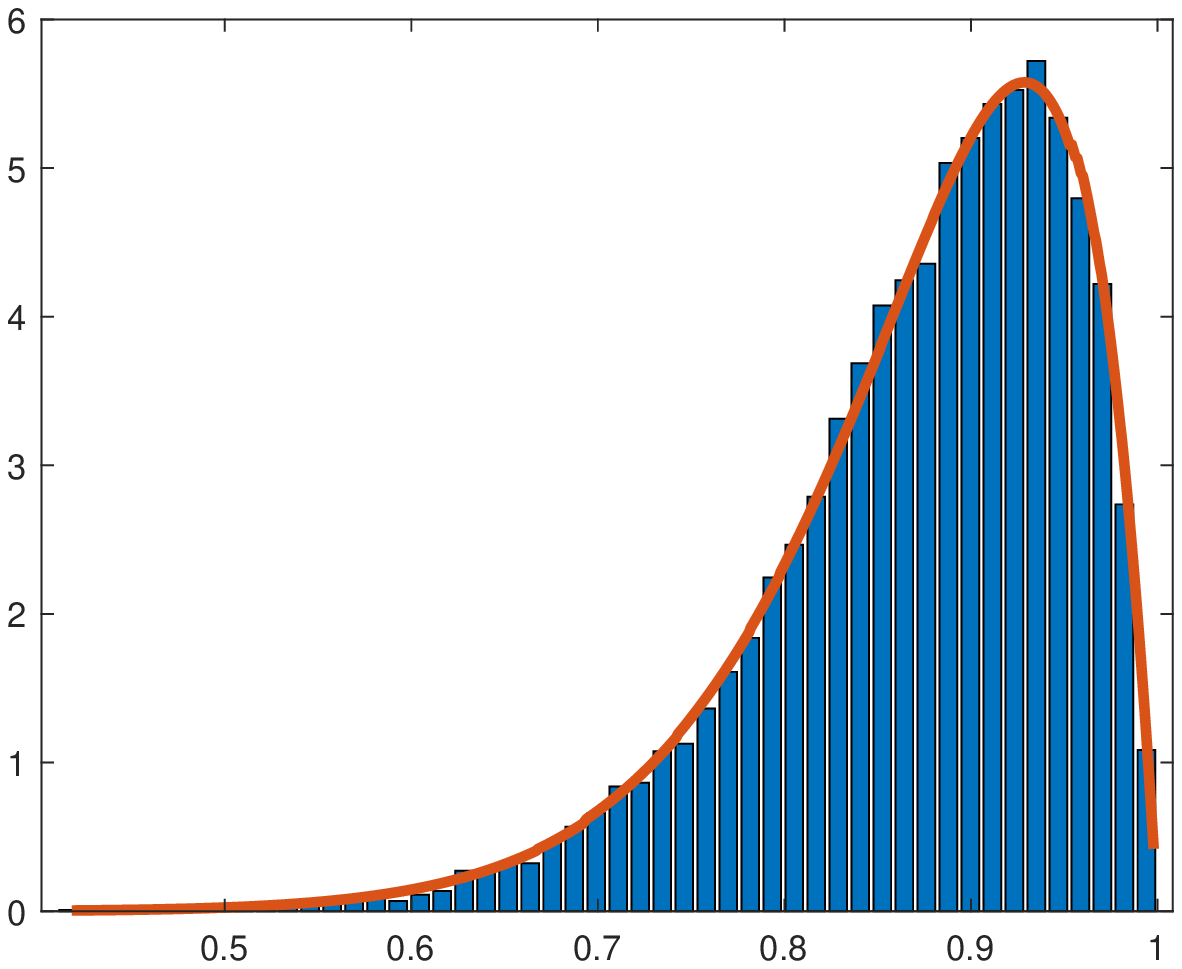} }} \hspace*{5mm}
	\subfloat
	\centering{{\includegraphics[width=.28\linewidth]{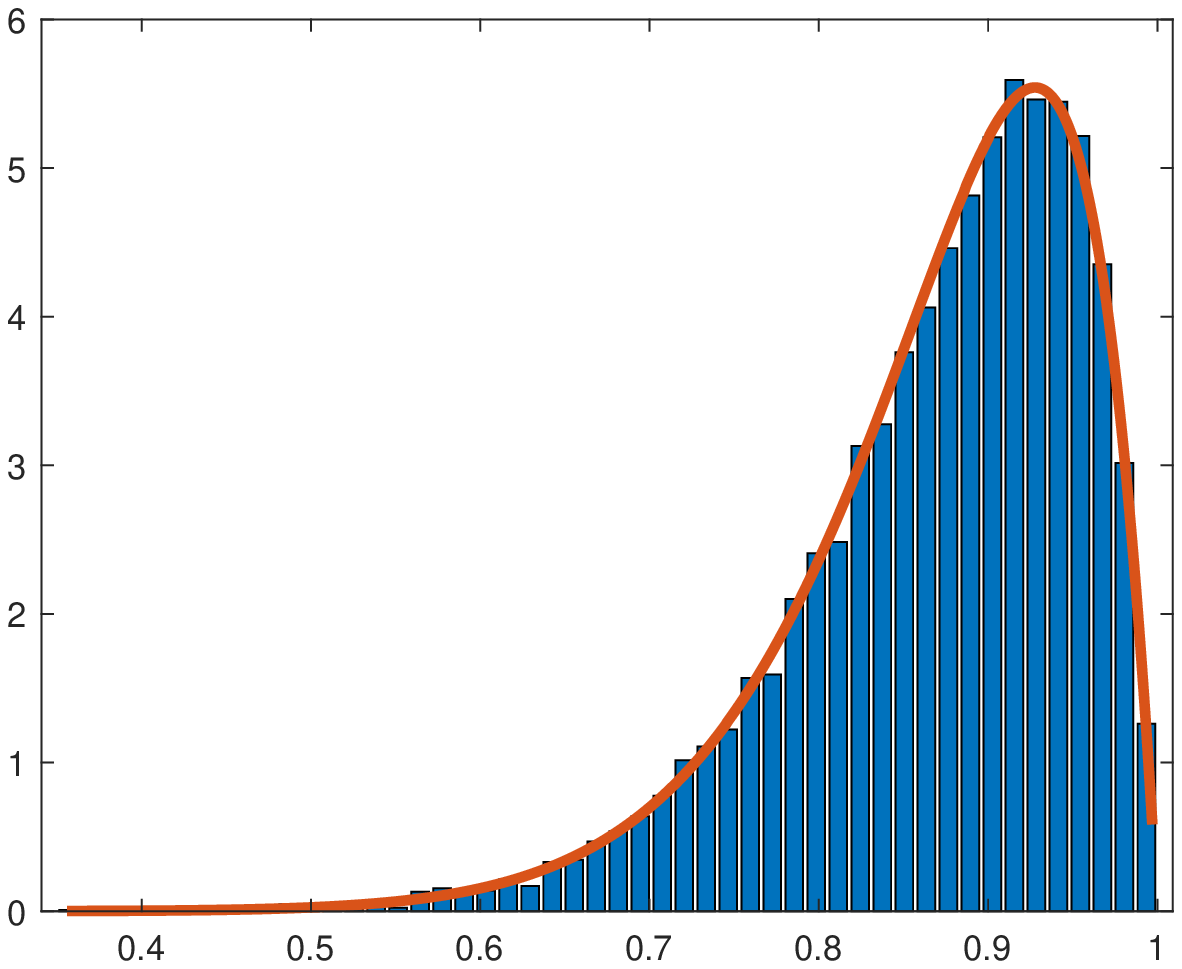}}}
	\\ \vspace*{5mm}
	\centering
	\subfloat
	\centering{{\includegraphics[width=.28\linewidth]{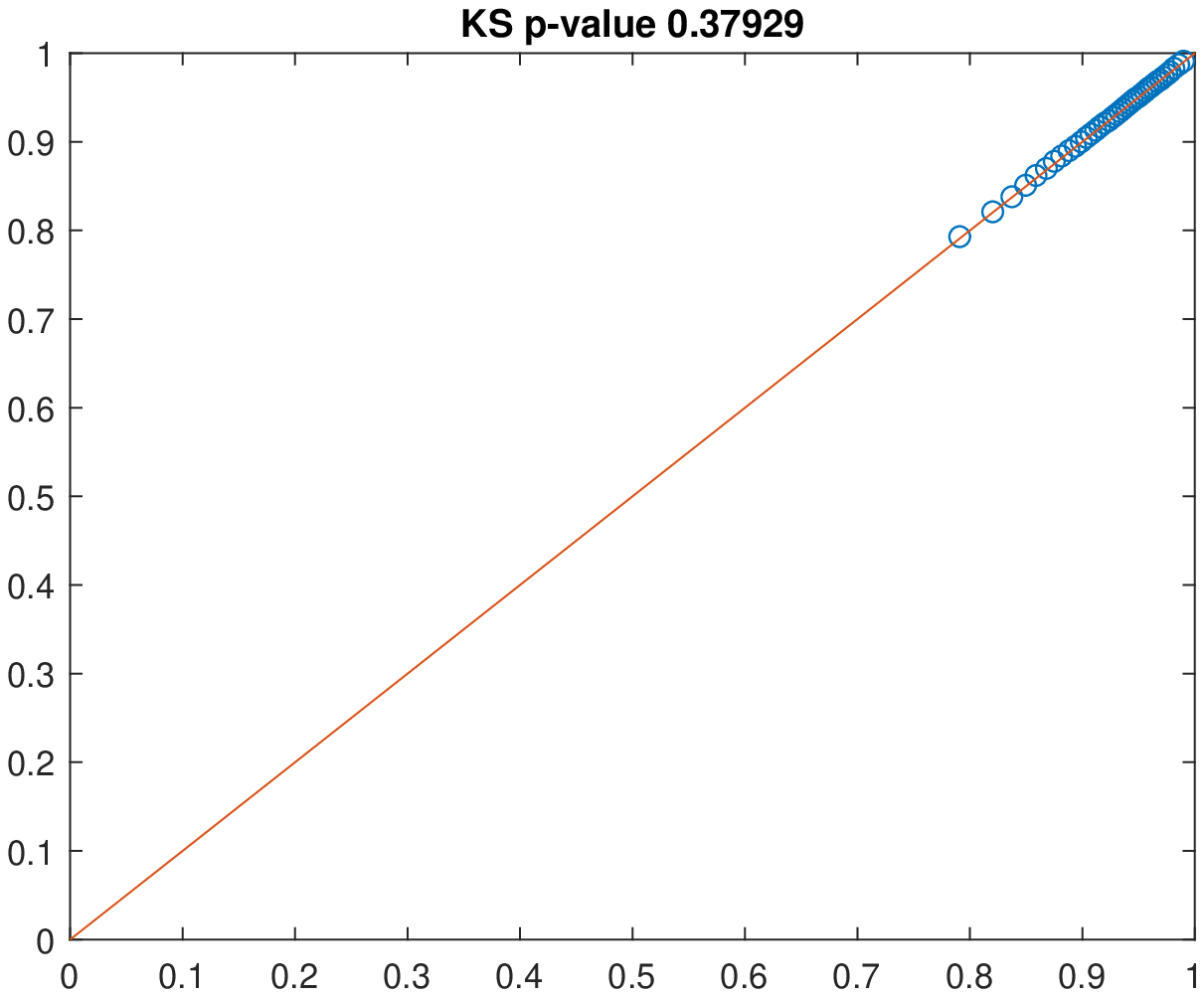} }} \hspace*{5mm}
	\subfloat
	\centering{{\includegraphics[width=.28\linewidth]{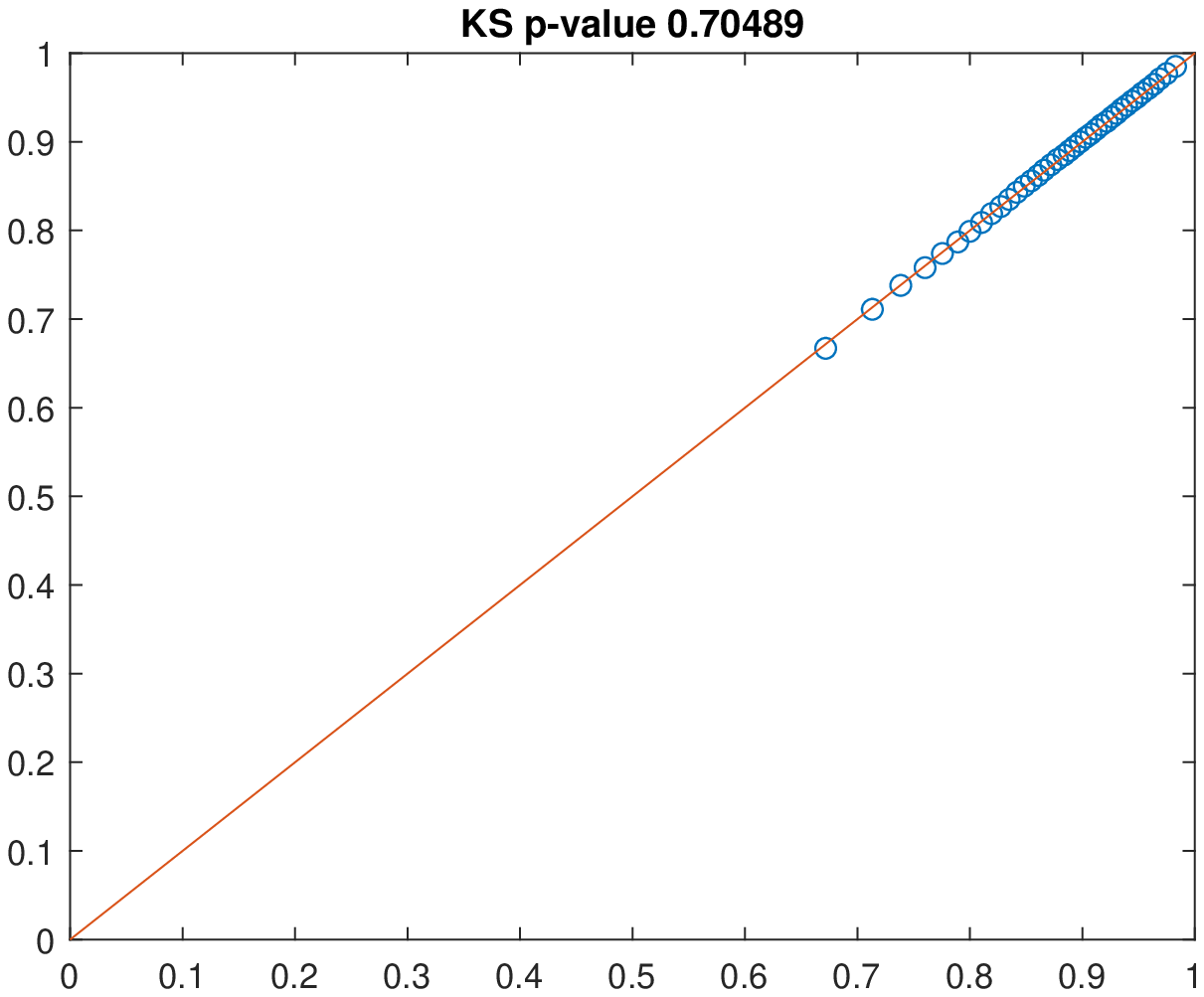} }} \hspace*{5mm}
	\subfloat
	\centering{{\includegraphics[width=.28\linewidth]{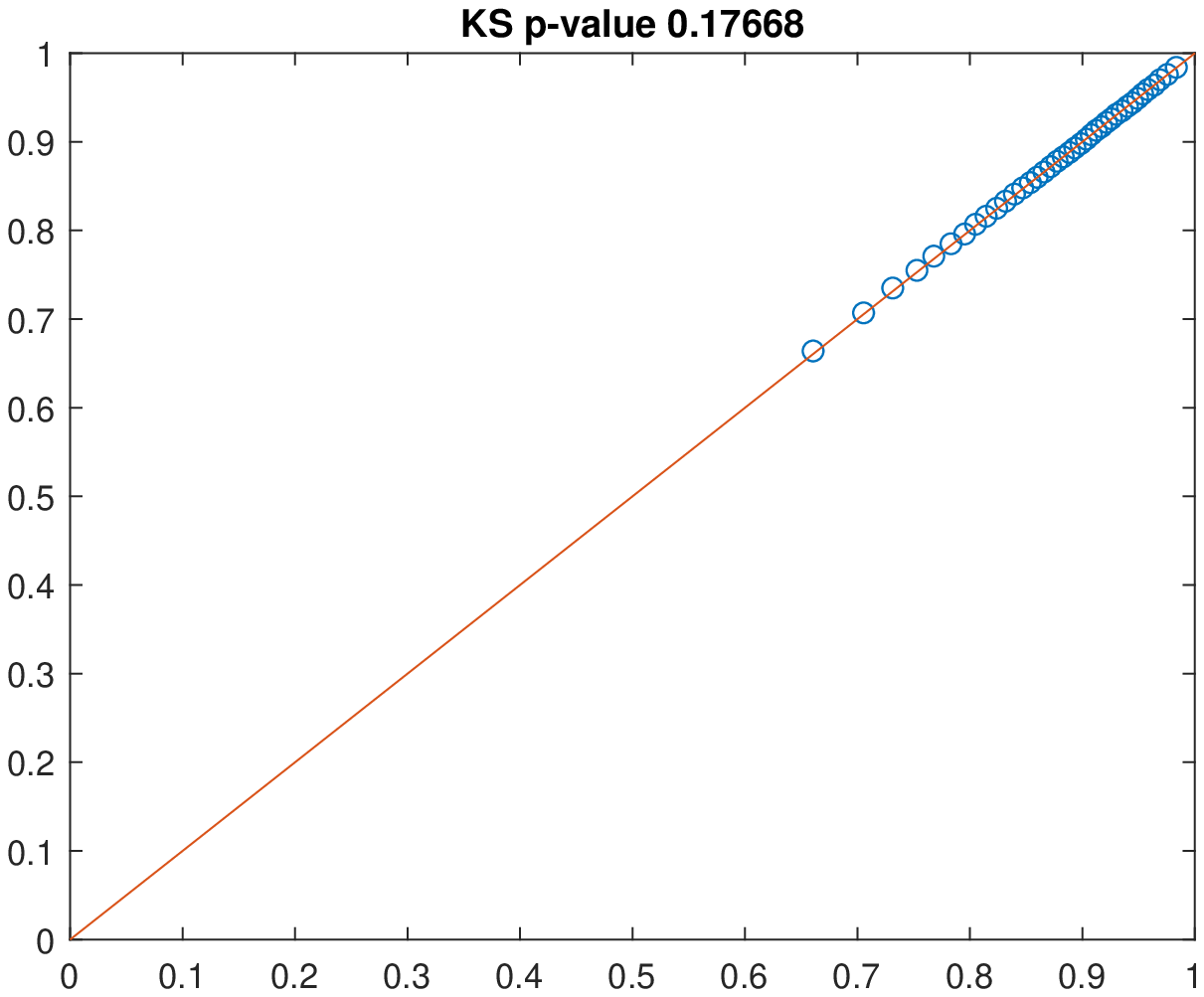}}}
	\caption{(Top row): Histograms for 10,000 samples generated from the law of the Wright--Fisher diffusion bridge `Bridge 3' as given in Table \ref{ConditionedBridges} above, sampled at the times given by the corresponding row in Table \ref{ConditionedBridgesSamplingTimes}. The truncated transition density is plotted in red. (Bottom row): QQ-plots for the corresponding samples with the $p$-value returned from the Kolmogorov--Smirnov test reported above the plot.}
\end{figure}

\subsection{Unconditioned bridges}
\noindent When the diffusion bridge is allowed to be absorbed at the boundary and $\boldsymbol{\theta}=\boldsymbol{0}$, we need only consider the cases when $z\in\{0,1\}$. To this end we considered the following two setups:

\begin{table}[H]
	\begin{center}
		\begin{tabular}{c|c c}
			& $(t_0,x_0)$ & $(t_1,x_1)$ \\ \hline
			Bridge 1 & (0,0.25) & (0.3,1) \\ \hline 
			Bridge 2 & (0,0.5) & (0.5,0) 
		\end{tabular}
	\end{center}
	\caption{The left and right endpoints for the three different bridges simulated, where $(t_0,x_0)$ denotes the bridge's start time $t_0$ and start point $x_0$, $(t_1,x_1)$ denotes the second observation time and point for the diffusion bridge and so on.}
	\label{UnconditionedBridges}
\end{table}

\noindent We further considered the following sampling times:

\begin{table}[H]
	\begin{center}
		\begin{tabular}{c|c c c}
			& $s_1$ & $s_2$ & $s_3$ \\ \hline
			Bridge 1 & 0.05 & 0.15 & 0.25 \\ \hline
			Bridge 2 & 0.05 & 0.25 & 0.45
		\end{tabular}
	\end{center}
	\caption{Sampling times for the two different diffusion bridges considered.}
	\label{UnconditionedBridgesSamplingTimes}
\end{table}

\noindent As in the diffusion case, we report the probability of absorption at the boundary in the table below, where once more $\widehat{\mathbb{P}}$ denotes the empirical estimate for this quantity whereas $\mathbb{P}$ is the theoretical value obtained by evaluating the truncation to the transition density at the boundary.

\begin{table}[H]
	\begin{center}
		\begin{tabular}{c|c c}
			Bridge 1 & $\widehat{\mathbb{P}}[\textnormal{Absorbed at 1}]$ & $\mathbb{P}[\textnormal{Absorbed at 1}]$ \\ \hline
			$s = 0.05$ & 0 & 3.900485e-16 \\
			$s = 0.15$ & 7e-4 & 6.752749e-4 \\
			$s = 0.25$ & 0.2331 & 0.234209
		\end{tabular}
	\end{center}
	\caption{Empirical ($\widehat{\mathbb{P}}$) and theoretical ($\mathbb{P}$) absorption probabilities for the diffusion started at $x=0.25$ and ending at $z=1$.}
	\label{AbsorptionProbabilitiesB1}
\end{table}

\begin{table}[H]
	\begin{center}
		\begin{tabular}{c|c c}
			Bridge 2 & $\widehat{\mathbb{P}}[\textnormal{Absorbed at 0}]$ & $\mathbb{P}[\textnormal{Absorbed at 0}]$ \\ \hline
			$s = 0.05$ & 0 & 2.418920e-10 \\
			$s = 0.25$ & 0.0881 & 0.085472 \\
			$s = 0.45$ & 0.7634 & 0.765359
		\end{tabular}
	\end{center}
	\caption{Empirical ($\widehat{\mathbb{P}}$) and theoretical ($\mathbb{P}$) absorption probabilities for the diffusion started at $x=0.5$ and ending at $x=0$.}
	\label{AbsorptionProbabilitiesB2}
\end{table}

\noindent The output generated is plotted below, starting with bridge 1, and the sampling time $s$ increasing from left to right. All of the plots, tests and probabilities above confirm that we are drawing samples from the desired distribution.

\begin{figure}[H]
	\centering
	\subfloat
	\centering{{\includegraphics[width=.28\linewidth]{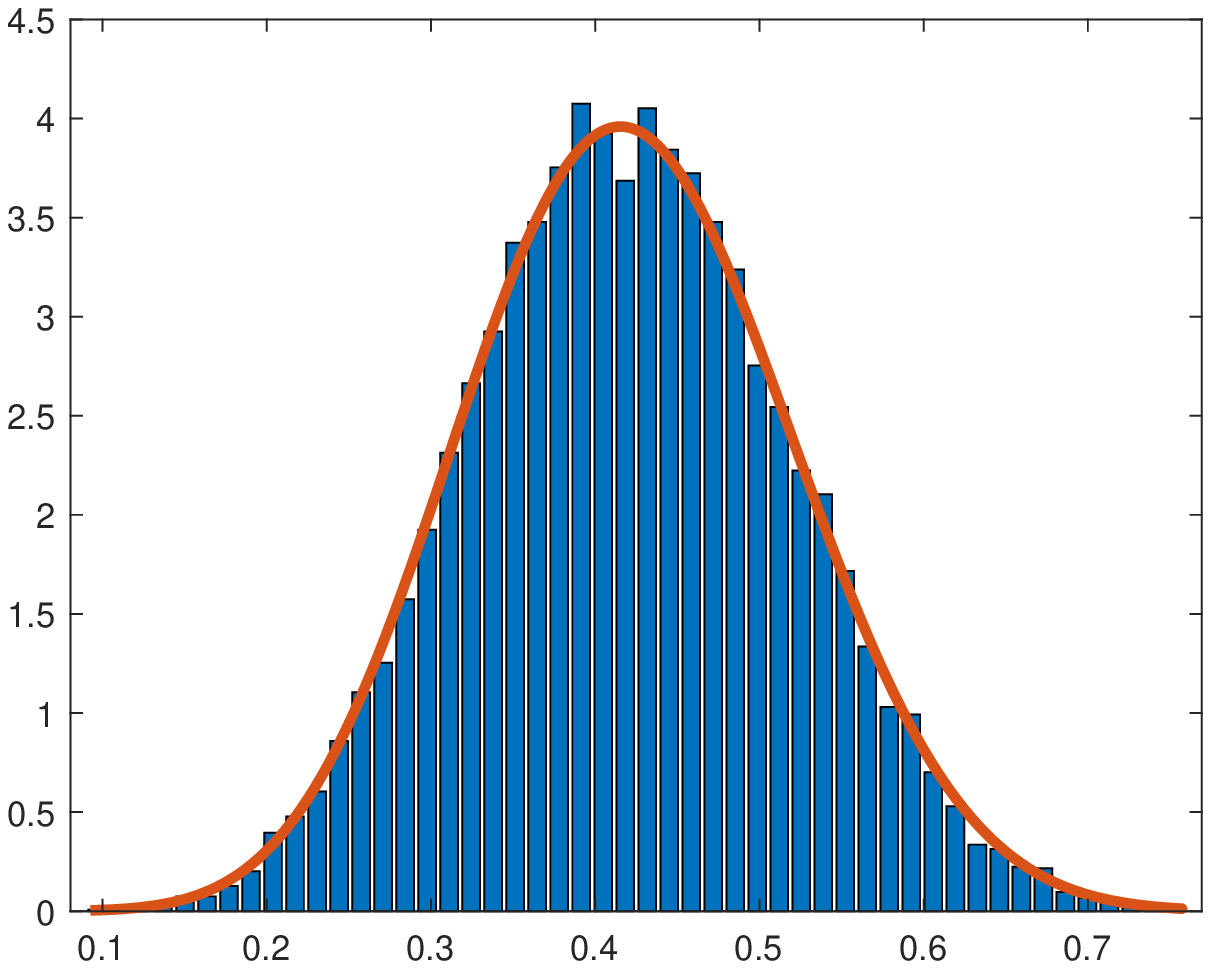} }} \hspace*{5mm}
	\subfloat
	\centering{{\includegraphics[width=.28\linewidth]{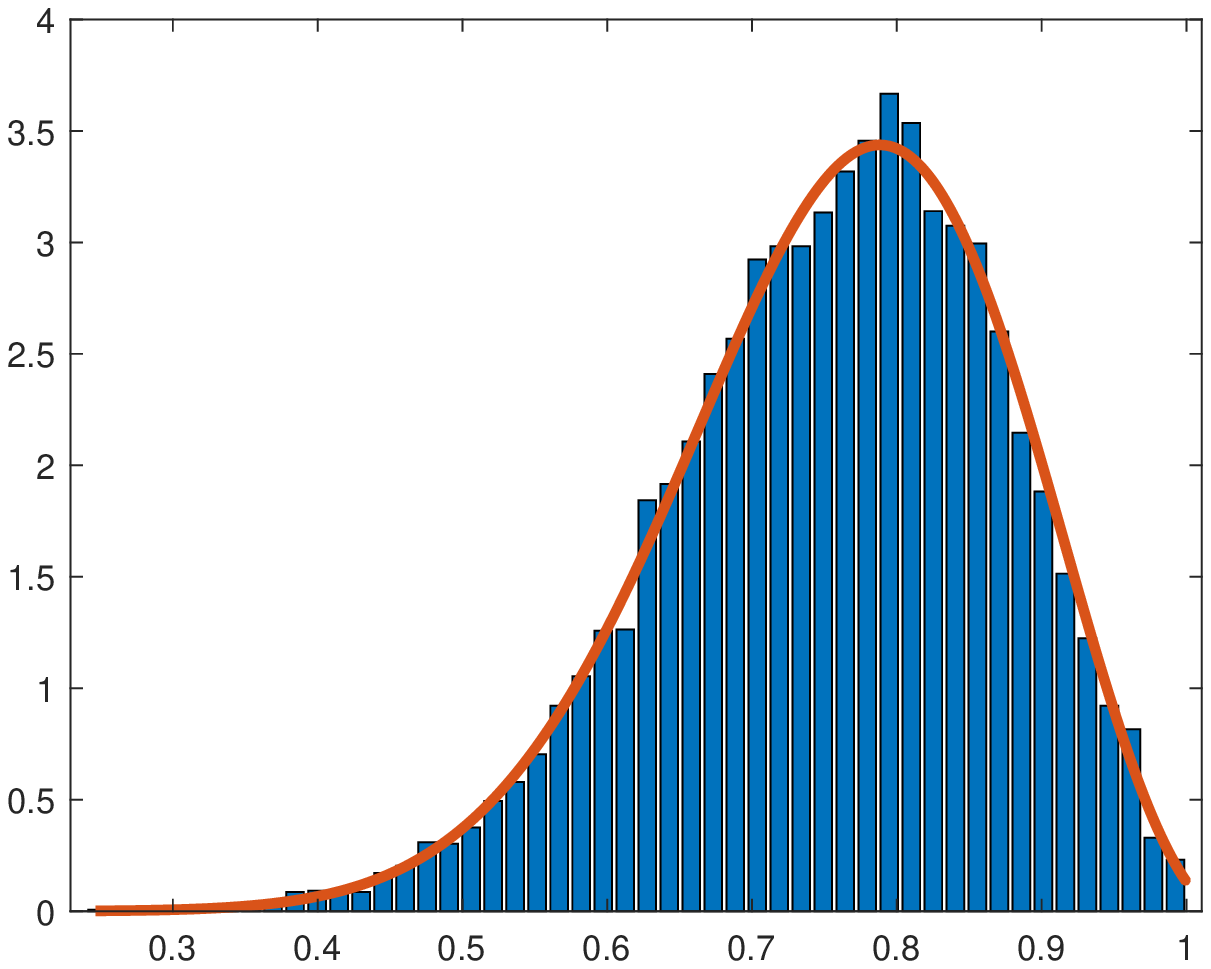} }} \hspace*{5mm}
	\subfloat
	\centering{{\includegraphics[width=.28\linewidth]{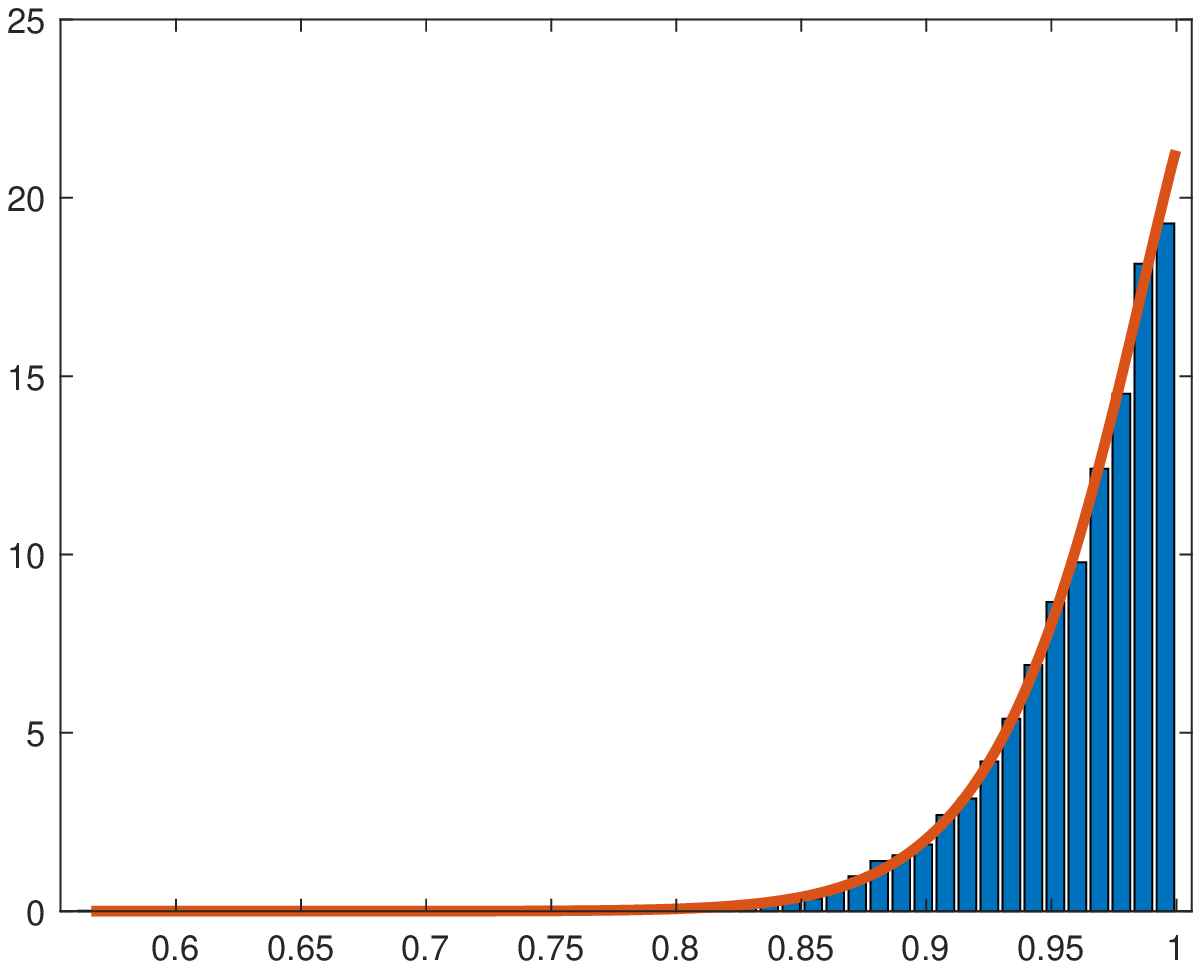}}}
	\\ \vspace*{5mm}
	\centering
	\subfloat
	\centering{{\includegraphics[width=.28\linewidth]{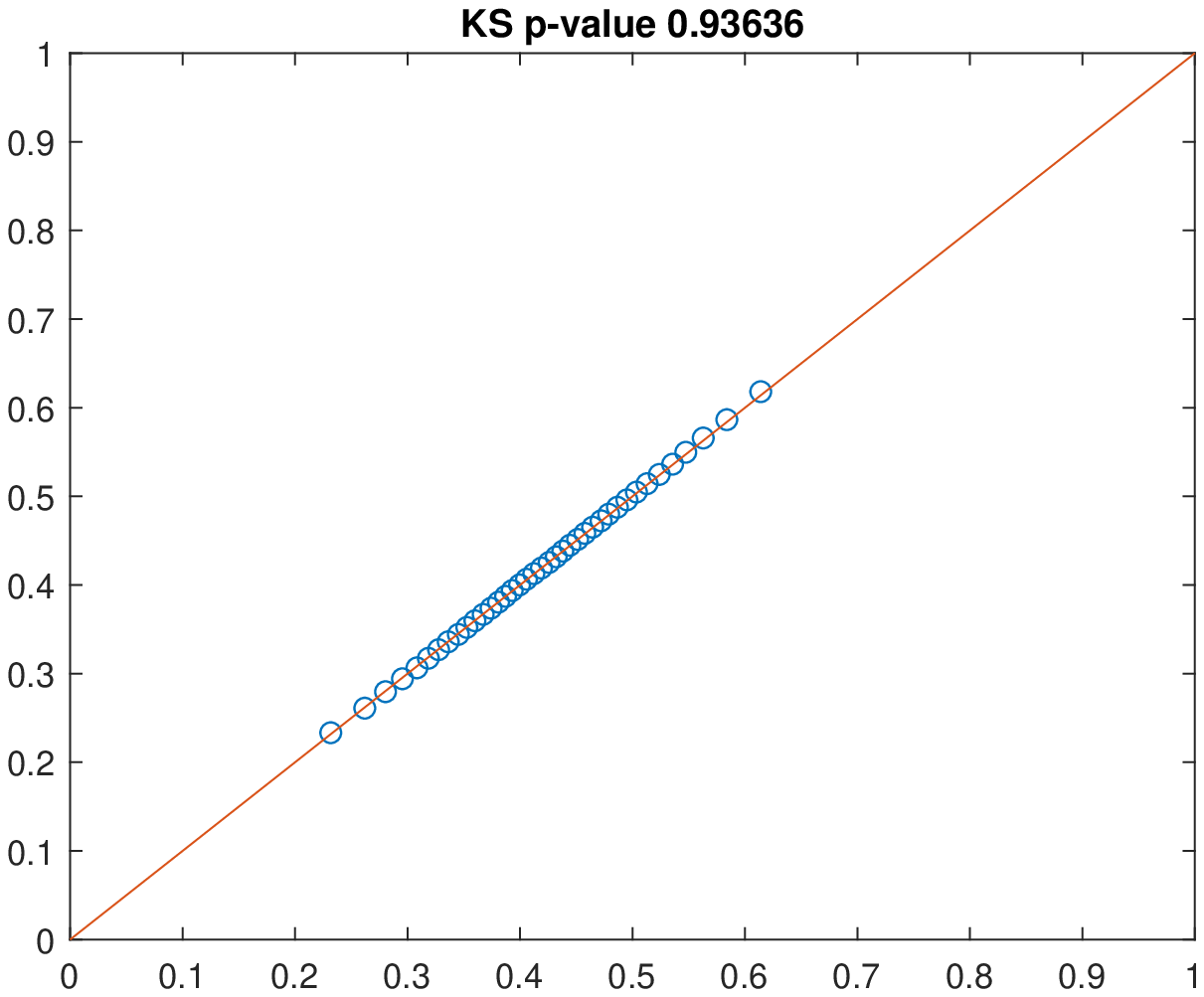} }} \hspace*{5mm}
	\subfloat
	\centering{{\includegraphics[width=.28\linewidth]{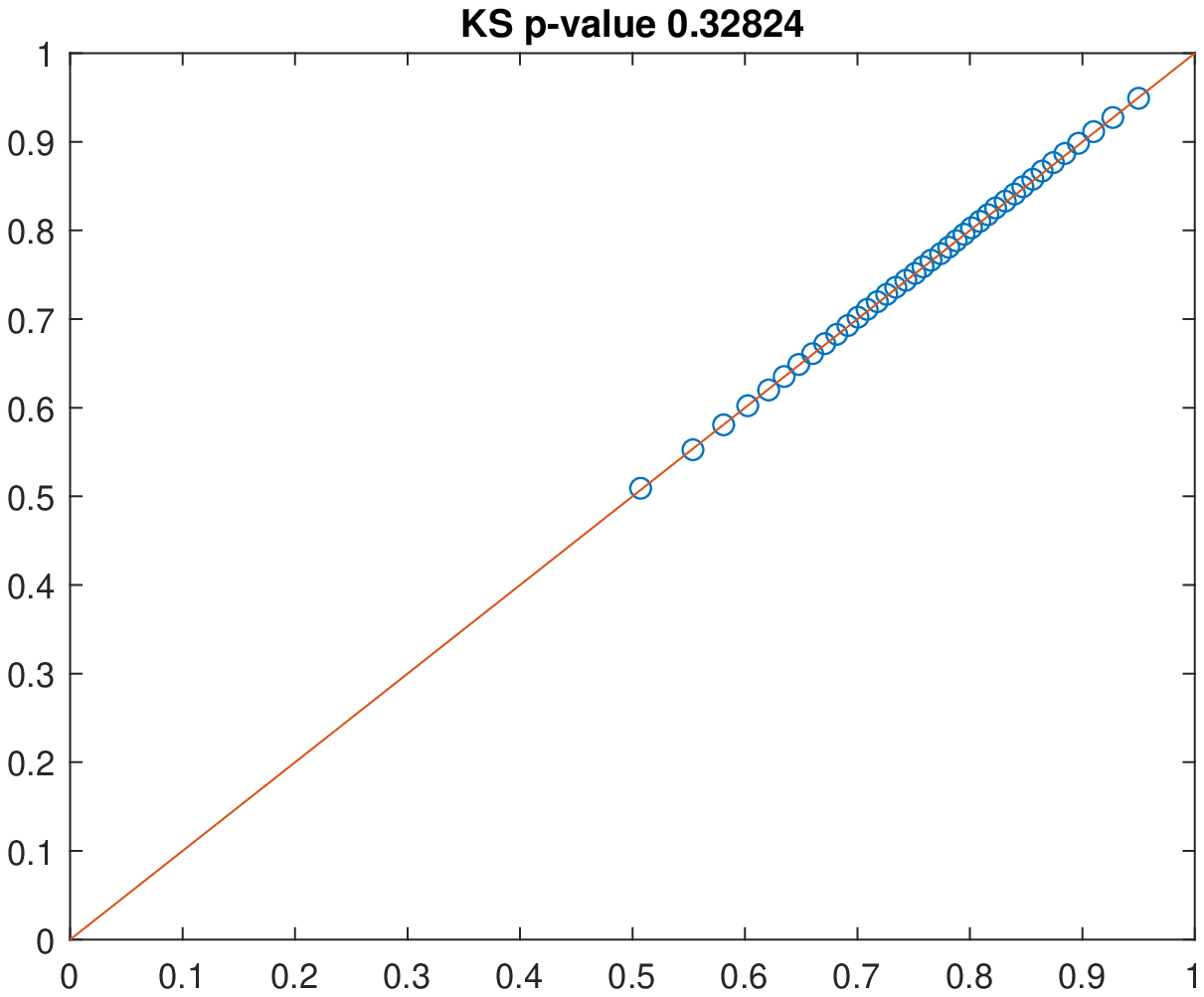} }} \hspace*{5mm}
	\subfloat
	\centering{{\includegraphics[width=.28\linewidth]{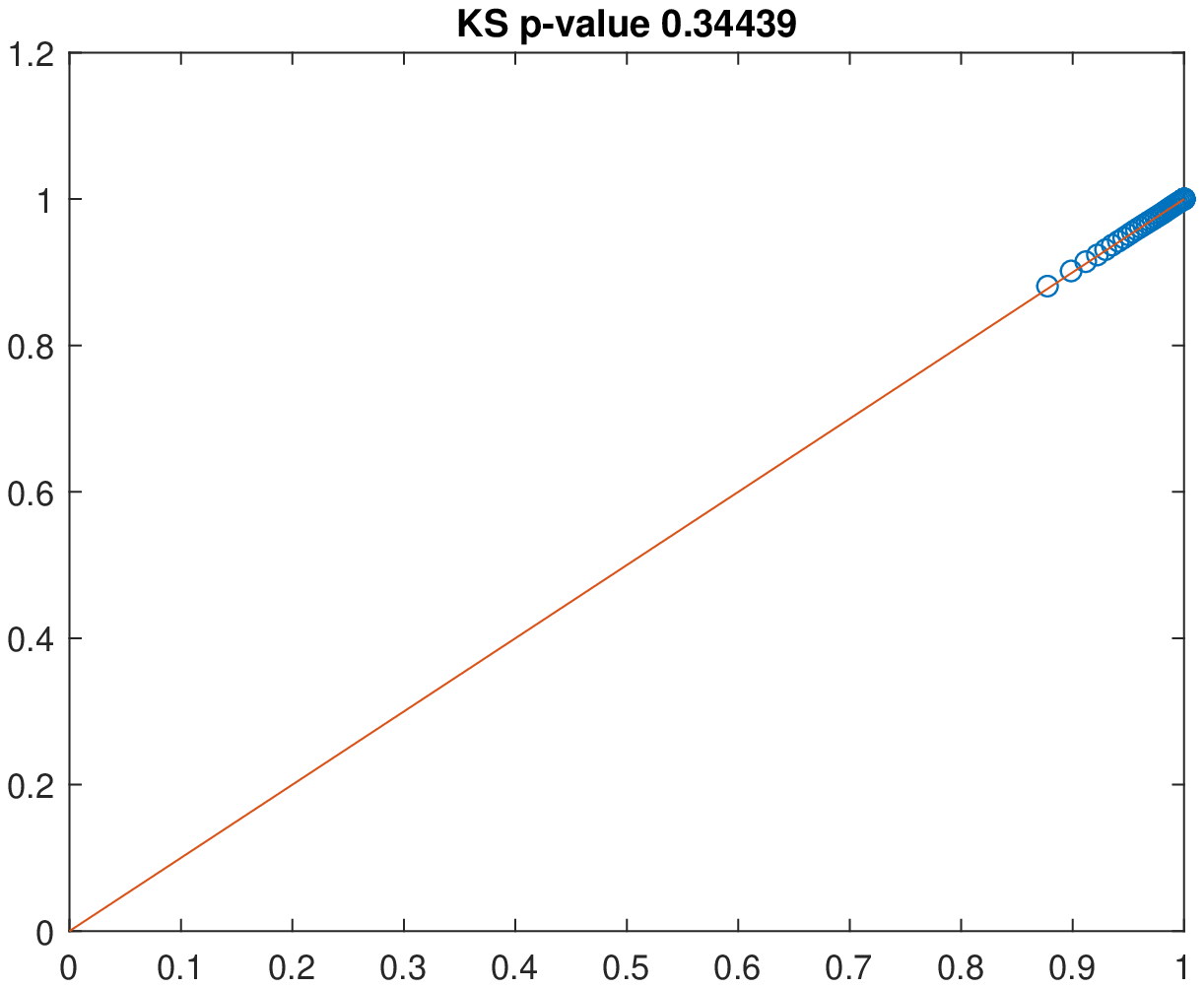}}}
	\caption{(Top row): Histograms for 10,000 samples generated from the law of the Wright--Fisher diffusion bridge `Bridge 1' (allowed to be absorbed at 1) as given in Table \ref{UnconditionedBridges}, sampled at the times given by the corresponding row in Table \ref{UnconditionedBridgesSamplingTimes}. Note that the samples equal to 1 are not included in the above plots, but their relative frequency can be found in Table \ref{AbsorptionProbabilitiesB1}. The truncated transition density is plotted in red. (Bottom row): QQ-plots for the corresponding samples with the $p$-value returned from the Kolmogorov--Smirnov test reported above the plot.}
\end{figure}

\begin{figure}[H]
	\centering
	\subfloat
	\centering{{\includegraphics[width=.28\linewidth]{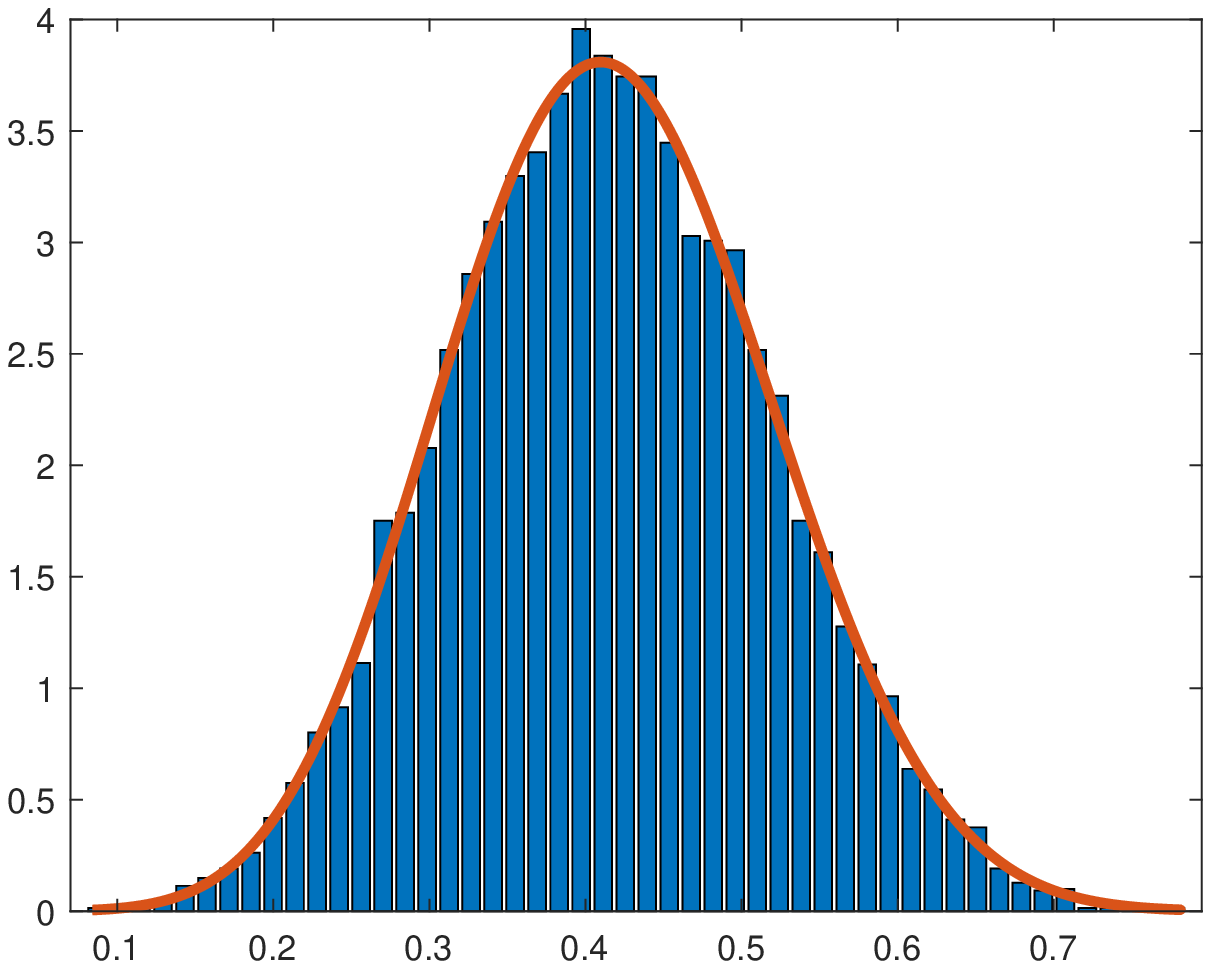} }} \hspace*{5mm}
	\subfloat
	\centering{{\includegraphics[width=.28\linewidth]{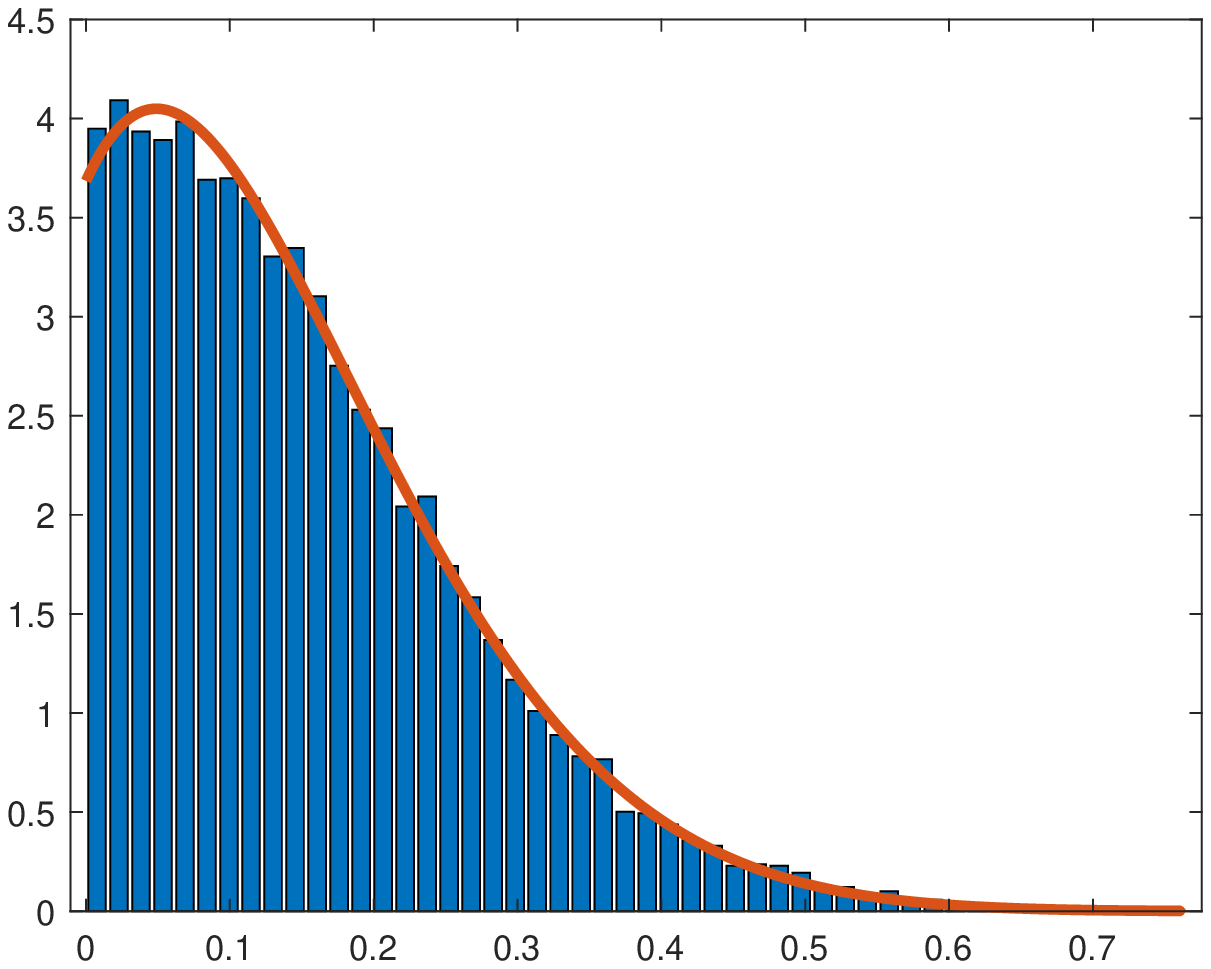} }} \hspace*{5mm}
	\subfloat
	\centering{{\includegraphics[width=.28\linewidth]{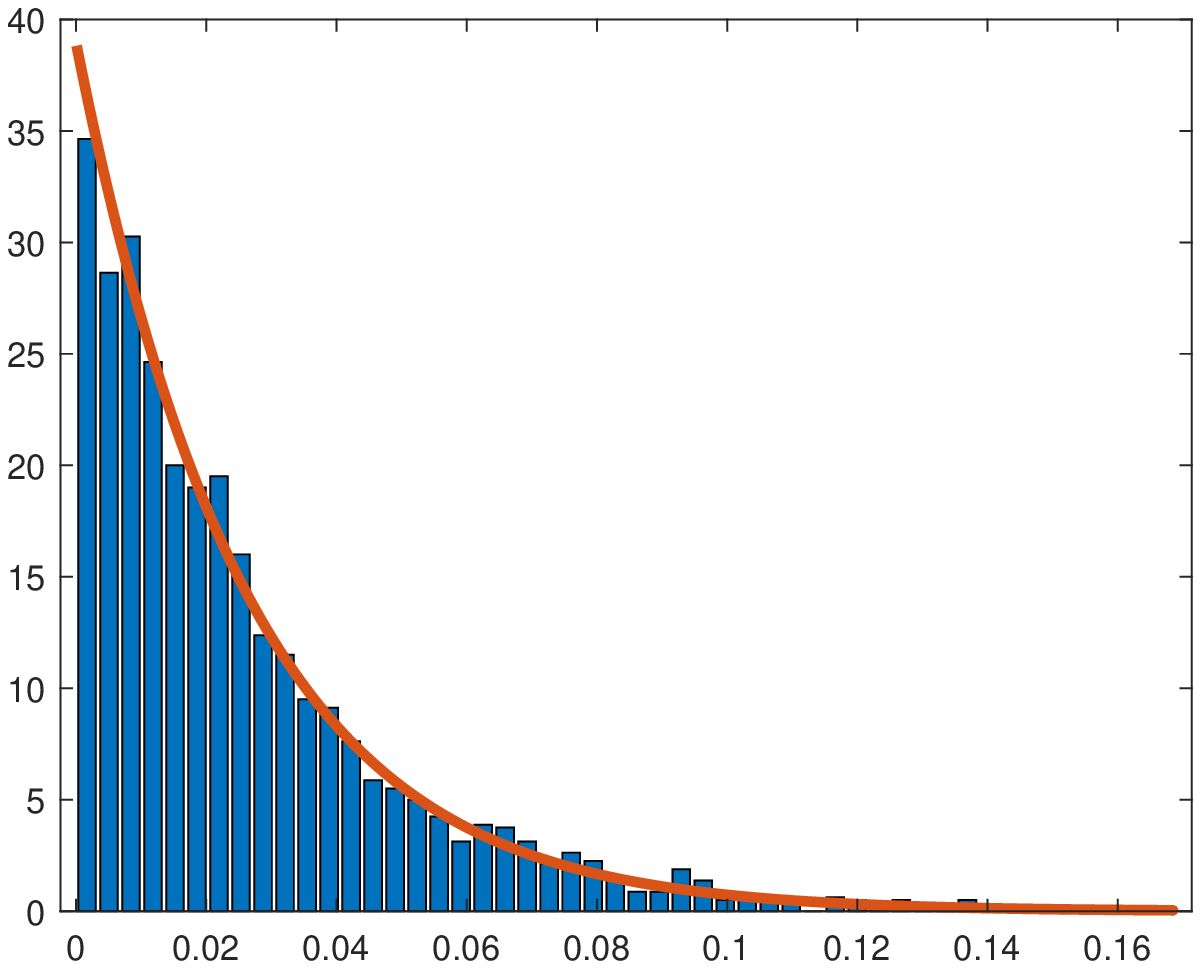}}}
	\\ \vspace*{5mm}
	\centering
	\subfloat
	\centering{{\includegraphics[width=.28\linewidth]{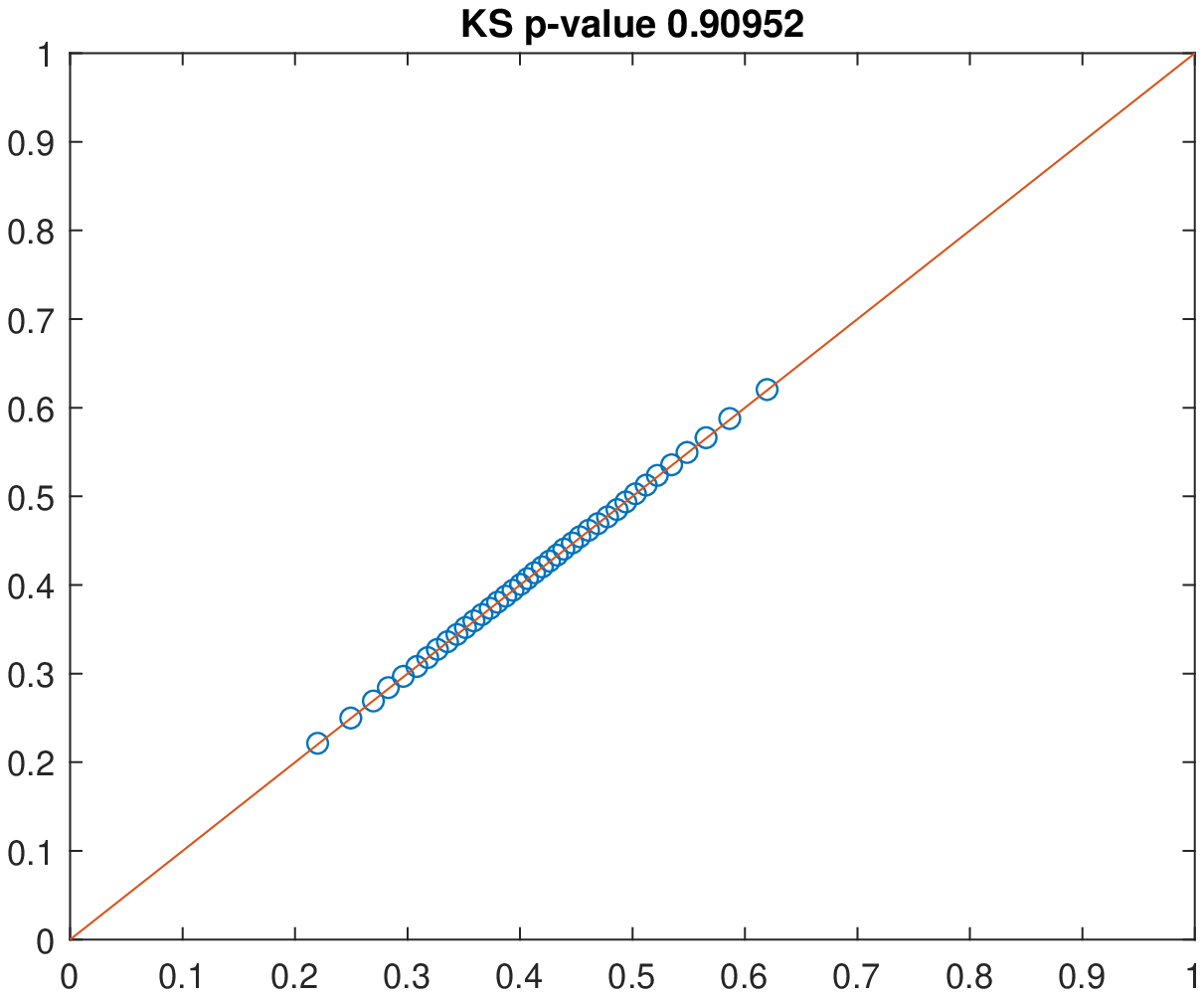} }} \hspace*{5mm}
	\subfloat
	\centering{{\includegraphics[width=.28\linewidth]{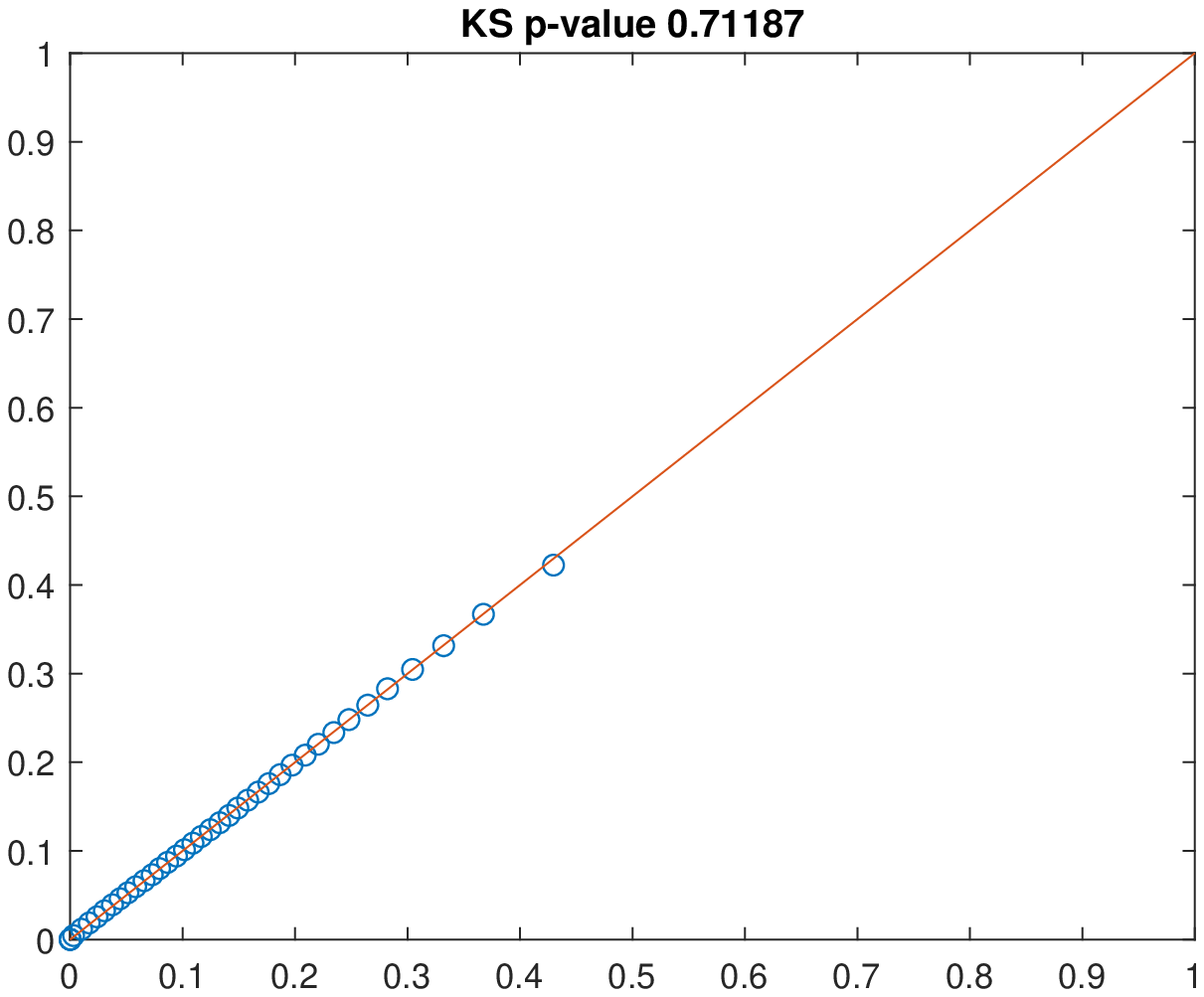} }} \hspace*{5mm}
	\subfloat
	\centering{{\includegraphics[width=.28\linewidth]{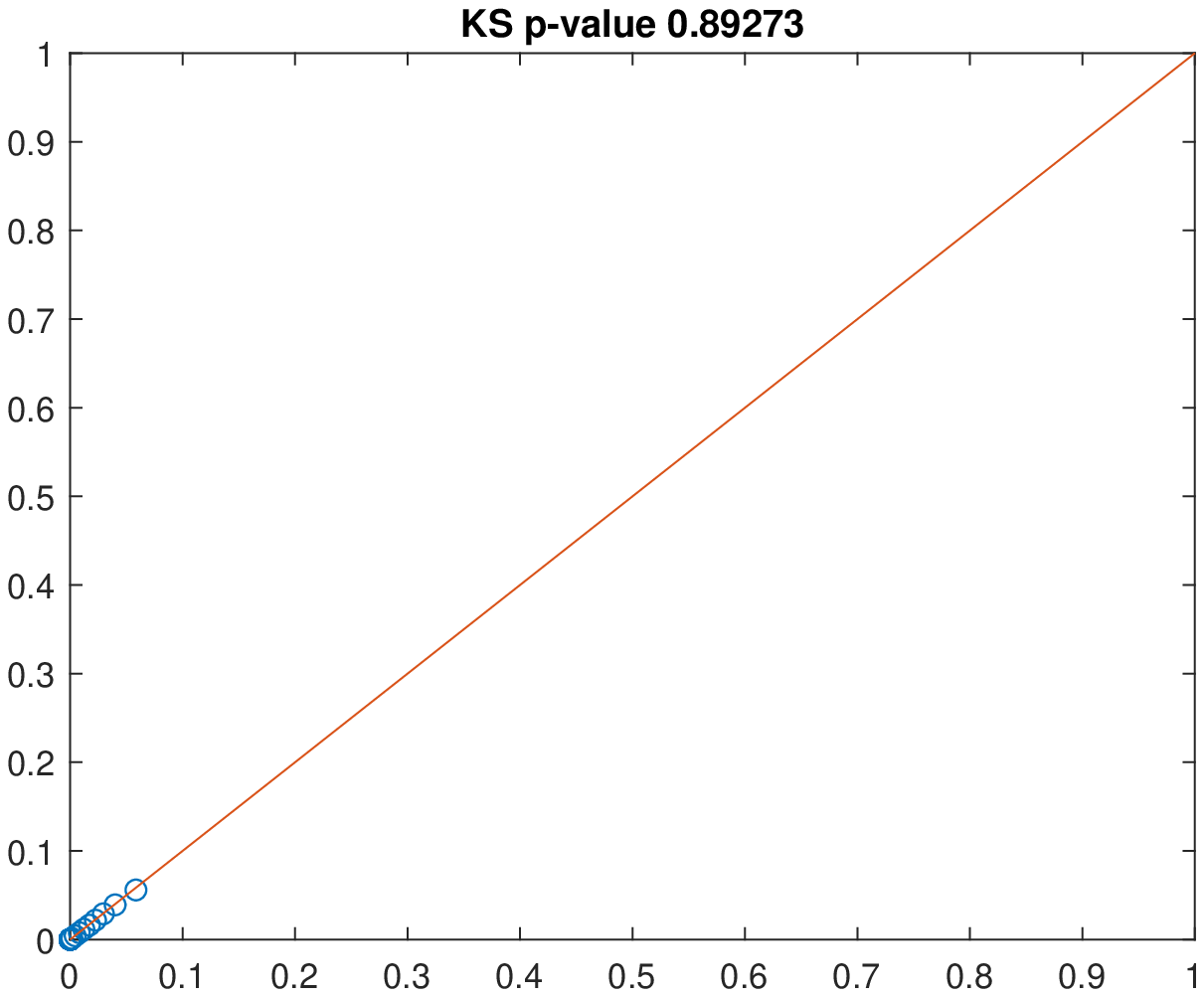}}}
	\caption{(Top row): Histograms for 10,000 samples generated from the law of the Wright--Fisher diffusion bridge `Bridge 2' (allowed to be absorbed at 0) as given in Table \ref{UnconditionedBridges}, sampled at the times given by the corresponding row in Table \ref{UnconditionedBridgesSamplingTimes}. Note that the samples equal to 0 are not included in the above plots, but their relative frequency can be found in Table \ref{AbsorptionProbabilitiesB2}. The truncated transition density is plotted in red. (Bottom row): QQ-plots for the corresponding samples with the $p$-value returned from the Kolmogorov--Smirnov test reported above the plot.}
\end{figure}

\subsection{Non-neutral diffusions and diffusion bridges}
\noindent Non-neutral Wright--Fisher paths can be generated (as described in Section \ref{NonneutralPaths}) through the use of neutral paths coupled with an appropriate Poisson point process. This technique was proposed in \cite{JenkinsSpano} and is used (without any alteration) in the current implementation of EWF to return non-neutral draws from the laws of both diffusions and diffusion bridges. Thus, although EWF does allow for non-neutral draws under a very broad class of selective regimes (and instructions on how to do this can be found in the respective configuration files), we omit the resulting output. 

\bibliographystyle{natbib}
\bibliography{main}

\begin{thebibliography}{}

\bibitem[Bollback {\em et~al.}(2008)Bollback, York, and Nielsen]{Bollback}
Bollback, J.~P.  {\em et~al.} (2008).
\newblock Estimation of $2{N}_es$ from temporal allele frequency data.
\newblock {\em Genetics\/}, {\bf 179}(1), 497--502.

\bibitem[Dangerfield {\em et~al.}(2012)Dangerfield, Kay, Macnamara, and
  Burrage]{Dangerfield}
Dangerfield, C.~E.  {\em et~al.} (2012).
\newblock A boundary preserving numerical algorithm for the {W}right--{F}isher
  model with mutation.
\newblock {\em BIT Numerical Mathematics\/}, {\bf 52}(2), 283--304.

\bibitem[Fages {\em et~al.}(2019)Fages {\em et~al.}]{Fages2019}
Fages, A. {\em et~al.} (2019).
\newblock Tracking five millennia of horse management with extensive ancient
  genome time series.
\newblock {\em Cell\/}, {\bf 177}, 1419 -- 1435.e31.

\bibitem[Fitzsimmons {\em et~al.}(1993)Fitzsimmons, Pitman, and
  Yor]{Fitzsimmons92}
Fitzsimmons, P.  {\em et~al.} (1993).
\newblock {M}arkovian bridges: construction, {P}alm interpretation, and
  splicing.
\newblock In {\em Seminar on Stochastic Processes, 1992\/}, pages 101--134.
  Springer.

\bibitem[Griffiths(1979)Griffiths]{Griffiths79}
Griffiths, R. (1979).
\newblock A transition density expansion for a multi-allele diffusion model.
\newblock {\em Advances in Applied Probability\/}, {\bf 11}(2), 310--325.

\bibitem[Griffiths(1984)Griffiths]{Griffiths84}
Griffiths, R.~C. (1984).
\newblock Asymptotic line-of-descent distributions.
\newblock {\em J. Math. Biol.}, {\bf 21}(1), 67--75.

\bibitem[Griffiths {\em et~al.}(2018)Griffiths, Jenkins, and
  Span\`o]{GriffithsJenkinsSpano}
Griffiths, R.~C.  {\em et~al.} (2018).
\newblock {W}right--{F}isher diffusion bridges.
\newblock {\em Theor. Popul. Biol.}, {\bf 122}, 67--77.

\bibitem[Jenkins and Span\`o(2017)Jenkins and Span\`o]{JenkinsSpano}
Jenkins, P.~A. and Span\`o, D. (2017).
\newblock {Exact simulation of the {W}right--{F}isher diffusion}.
\newblock {\em Ann. Appl. Probab.}, {\bf 27}(3), 1478--1509.

\bibitem[Ludwig {\em et~al.}(2009)Ludwig, Pruvost, Reissmann, Benecke,
  Brockmann, Casta{\~n}os, Cieslak, Lippold, Llorente, Malaspinas, Slatkin, and
  Hofreiter]{Ludwig}
Ludwig, A.  {\em et~al.} (2009).
\newblock Coat color variation at the beginning of horse domestication.
\newblock {\em Science\/}, {\bf 324}(5926), 485--485.

\bibitem[Malaspinas {\em et~al.}(2012)Malaspinas, Malaspinas, Evans, and
  Slatkin]{Malaspinas}
Malaspinas, A.-S.  {\em et~al.} (2012).
\newblock Estimating allele age and selection coefficient from time-serial
  data.
\newblock {\em Genetics\/}, {\bf 192}(2), 599--607.

\bibitem[Mathieson and McVean(2013)Mathieson and McVean]{MathiesonMcVean}
Mathieson, I. and McVean, G. (2013).
\newblock Estimating selection coefficients in spatially structured populations
  from time series data of allele frequencies.
\newblock {\em Genetics\/}, {\bf 193}(3), 973--984.

\bibitem[Schraiber {\em et~al.}(2016)Schraiber, Evans, and Slatkin]{Schraiber}
Schraiber, J.~G.  {\em et~al.} (2016).
\newblock {B}ayesian inference of natural selection from allele frequency time
  series.
\newblock {\em Genetics\/}, {\bf 203}(1), 493--511.

\bibitem[Steinr{\"u}cken {\em et~al.}(2016)Steinr{\"u}cken, Jewett, and
  Song]{SteinruckenTDF}
Steinr{\"u}cken, M.  {\em et~al.} (2016).
\newblock {SpectralTDF}: transition densities of diffusion processes with
  time-varying selection parameters, mutation rates and effective population
  sizes.
\newblock {\em Bioinformatics\/}, {\bf 32}(5), 795--797.

\bibitem[Tavar{\'e}(1984)Tavar{\'e}]{Tavare84}
Tavar{\'e}, S. (1984).
\newblock Line-of-descent and genealogical processes, and their applications in
  population genetics models.
\newblock {\em Theoretical population biology\/}, {\bf 26}(2), 119--164.

\bibitem[Wutke {\em et~al.}(2016)Wutke {\em et~al.}]{Wutke2016}
Wutke, S. {\em et~al.} (2016).
\newblock Spotted phenotypes in horses lost attractiveness in the middle ages.
\newblock {\em Scientific Reports\/}, {\bf 6}.

\end{thebibliography}

\end{document}